\documentclass[11pt]{article}

\usepackage{bold-extra}
\usepackage[nottoc]{tocbibind}
\usepackage[normalem]{ulem}
\usepackage{tocloft}
\usepackage[toc]{multitoc}

\usepackage[affil-it]{authblk}
\usepackage{microtype}
\usepackage[margin=1in]{geometry}
\usepackage{amssymb}
\usepackage{tikz}
\usepackage{hyperref}
\hypersetup{
    colorlinks=true,
    linkcolor=blue,
    filecolor=magenta,      
    urlcolor=cyan,
    }
\usepackage{svg}
\usepackage{bigints}
\usepackage{amsthm}
\usepackage{float} 
\usepackage{mathtools}
\usepackage[font=small,labelfont=bf]{caption}
\usepackage{xparse}
\usepackage{subcaption}
\usepackage{suffix}
\usepackage{tikz}
\usepackage{verbatim}

\usepackage[linesnumbered,boxed]{algorithm2e}
\SetAlCapSkip{1em}
\SetAlCapNameFnt{\small}
\SetAlCapFnt{\small}
\SetAlgorithmName{Algorithm}{}{}

\newcommand{\cH}{\mathcal H}

\DeclareMathOperator{\sdp}{\it \mathrm{SDP}}

\numberwithin{equation}{section}
\newtheorem{theorem}{Theorem}[section]

\newtheorem{lemma}[theorem]{Lemma}

\newtheorem*{claim*}{Claim}
\newtheorem{conjecture}[theorem]{Conjecture}
\newtheorem{question}[theorem]{Question}
\newtheorem*{conjecture*}{Conjecture}
\newtheorem*{definition*}{Definition}
\newtheorem*{theorem*}{Theorem}
\newtheorem*{remark*}{Remark}
\newtheorem*{question*}{Question}

\newtheorem{corollary}[theorem]{Corollary}
\newtheorem{proposition}[theorem]{Proposition}
\newtheorem{definition}[theorem]{Definition}

\newtheorem{remark}[theorem]{Remark}

\newtheorem*{lem*}{Lemma}

\theoremstyle{definition}

\theoremstyle{remark}

\newcommand{\gwb}{\mathrm{GWB}}

\newcommand{\R}{\mathbb{R}}
\newcommand{\C}{\mathbb{C}}

\newcommand{\N}{\mathbb{N}}

\newcommand{\tr}{\operatorname{tr}}

\newcommand{\abs}[1]{\lvert #1 \rvert}

\newcommand{\ip}[2]{\langle #1 , #2\rangle}
\newcommand{\bigip}[2]{\bigl\langle #1, #2 \bigr\rangle}
\newcommand{\Bigip}[2]{\Bigl\langle #1, #2 \Bigr\rangle}

\newcommand{\X}{\mathcal{X}}

\newcommand{\Z}{\mathbb{Z}}

\newcommand{\M}{\mathrm{M}}

\newcommand{\fid}{\mathrm{fid}}
\newcommand{\expect}{\mathbb{E}}

\newcommand{\microspace}{\mspace{.5mu}} %
\newcommand{\ket}[1]{\ensuremath{\lvert\microspace #1
    \microspace\rangle}} %
\newcommand{\bra}[1]{\ensuremath{\langle\microspace #1
    \microspace\rvert}} %    
\newcommand{\ketbra}[2]{\ensuremath{\lvert\microspace #1
    \microspace\rangle\! \langle \microspace #2 \microspace \rvert}} %

\newcommand{\paren}[1]{(#1)}

\newcommand{\bigparen}[1]{\big(#1\big)}
\newcommand{\Bigparen}[1]{\Big(#1\Big)}

\newcommand{\class}[1]{\mathsf{#1}} %

\newcommand{\NP}{\class{NP}} %
\newcommand{\Pp}{\class{P}} %

\newcommand{\QMA}{\class{QMA}} %
\newcommand{\QMIP}{\class{QMIP}} %
\newcommand{\MIP}{\class{MIP}} %
\newcommand{\RE}{\class{RE}} %

\newcommand{\linlegacy}[2]{\mathrm{Max}\textit{-}{#1}\textit{-}\mathrm{Lin}({#2})}

\newcommand{\lin}[2]{\mathrm{Max}\textit{-}\mathrm{Lin}({#2})}
\newcommand{\slin}[2]{\mathrm{SMax}\textit{-}\mathrm{Lin}({#2})}
\newcommand{\homlin}[2]{\mathrm{HMax}\textit{-}\mathrm{Lin}({#2})}
\newcommand{\cut}[1]{\mathrm{Max}\textit{-}{#1}\textit{-}\mathrm{Cut}}
\newcommand{\maxcut}{\mathrm{Max}\textit{-}\mathrm{Cut}}
\newcommand{\ugame}[1]{\mathrm{Unique}\textit{-}\mathrm{Game}({#1})}

\newcommand{\nlin}[2]{\mathrm{NC}\textit{-}\mathrm{Max}\textit{-}\mathrm{Lin}({#2})}
\newcommand{\ulin}[2]{\mathrm{Unitary}\textit{-}\mathrm{SMax}\textit{-}\mathrm{Lin}({#2})}
\newcommand{\nslin}[2]{\mathrm{NC}\textit{-}\mathrm{SMax}\textit{-}\mathrm{Lin}({#2})}
\newcommand{\nhomlin}[2]{\mathrm{NC}\textit{-}\mathrm{HMax}\textit{-}\mathrm{Lin}({#2})}
\newcommand{\ncut}[1]{\mathrm{NC}\textit{-}\mathrm{Max}\textit{-}{#1}\textit{-}\mathrm{Cut}}

\newcommand{\sopt}{\mathrm{S}\textit{-}\mathrm{OPT}}
\newcommand{\ssdp}{\mathrm{S}\textit{-}\mathrm{SDP}}
\newcommand{\nsopt}{\mathrm{NC}\textit{-}\mathrm{S}\textit{-}\mathrm{OPT}}
\newcommand{\nopt}{\mathrm{NC}\textit{-}\mathrm{OPT}}

\newcommand{\rel}{\delta}
\newcommand{\Rel}{\Delta}

\newcommand{\val}{\mathrm{val}}
\newcommand{\Val}{\mathrm{VAL}}

\newcommand{\eps}{\varepsilon}
\newcommand{\poly}{\mathrm{poly}}

\newcommand{\coRE}{\mathsf{coRE}}

\newcommand{\Mod}[1]{\ (\mathrm{mod}\ #1)}

\renewcommand{\hat}[1]{\widehat{#1}} 
\renewcommand{\tilde}[1]{\widetilde{#1}}

\newcommand{\Id}{1}

\let\epsilon=\varepsilon %

\newcommand{\set}[1]{\{#1\}}

\WithSuffix\newcommand\QMIP*{\ensuremath{\class{QMIP}^*}} %

\WithSuffix\newcommand\MIP*{\ensuremath{\class{MIP}^*}} %

\def\01{\{0,1\}}

\newcommand{\braket}[2]{\langle #1 | #2 \rangle}

\DeclareMathOperator{\SDP}{SDP}

\DeclareMathOperator{\opt}{\mathrm{OPT}}

\DeclareMathOperator{\argmax}{argmax}

\usepackage{mathrsfs,scalerel}

\newcommand{\ttt}[1]{\texttt{#1}}

\newcommand{\mc}[1]{\mathcal{#1}}
\newcommand{\mrm}[1]{\mathrm{#1}}
\newcommand{\scr}[1]{\mathscr{#1}}

\renewcommand{\Re}{\mathrm{Re}}

\DeclarePairedDelimiter\parens{\lparen}{\rparen}
\DeclarePairedDelimiter\squ{[}{]}
\DeclarePairedDelimiter\floor{\lfloor}{\rfloor}
\DeclarePairedDelimiter\absm{\lvert}{\rvert}
\DeclarePairedDelimiter\norm{\|}{\|}

\DeclareMathOperator*{\expec}{\scalerel*{\vphantom{\mathbb{E}}\text{\raisebox{0.3pt}{\scalebox{0.9}{$\mathbb{E}$}}}}{\sum}}

\DeclareDocumentCommand\ketbra{s m g}{
	\IfBooleanTF{#1}{
		\IfNoValueTF{#3}{
			\vphantom{#2}\left\lvert{#2}\middle\rangle\!\middle\langle{#2}\right\rvert
		}{
			\vphantom{#2#3}\left\lvert{#2}\middle\rangle\!\middle\langle{#3}\right\rvert
		}
	}{
		\IfNoValueTF{#3}{
			\vphantom{#2}\left\lvert\smash{#2}\middle\rangle\!\middle\langle\smash{#2}\right\rvert
		}{
			\vphantom{#2#3}\left\lvert\smash{#2}\middle\rangle\!\middle\langle\smash{#3}\right\rvert
		}
	}
}

\DeclareDocumentCommand\braket{s m g g}{
	\IfBooleanTF{#1}{
		\IfNoValueTF{#3}{
			\vphantom{#2}\left\langle{#2}\middle\vert{#2}\right\rangle
		}{
			\IfNoValueTF{#4}{
				\vphantom{#2#3}\left\langle{#2}\middle\vert{#3}\right\rangle
			}{
				\vphantom{#2#3#4}\left\langle{#2}\middle\rvert{#3}\middle\rvert{#4}\right\rangle
			}
		}
	}{
		\IfNoValueTF{#3}{
			\vphantom{#2}\left\langle\smash{#2}\middle\vert\smash{#2}\right\rangle
		}{
			\IfNoValueTF{#4}{
				\vphantom{#2#3}\left\langle\smash{#2}\middle\vert\smash{#3}\right\rangle
			}{
				\vphantom{#2#3#4}\left\langle\smash{#2}\middle\rvert\smash{#3}\middle\rvert\smash{#4}\right\rangle
			}
		}
	}
}

\DeclareDocumentCommand\group{s m g}{
	\IfBooleanTF{#1}{
		\IfNoValueTF{#3}{
			\vphantom{#2}\left\langle{#2}\right\rangle
		}{
			\vphantom{#2#3}\left\langle{#2}:{#3}\right\rangle
		}
	}{
		\IfNoValueTF{#3}{
			\vphantom{#2}\left\langle\smash{#2}\right\rangle
		}{
            \vphantom{#2#3}\left\langle\smash{#2}:\smash{#3}\right\rangle
		}
	}
}

\DeclareDocumentCommand\set{s o m g}{
	\IfBooleanTF{#1}{
		\IfNoValueTF{#4}{
			\left\{{#3}\right\}
		}{
			\left\{{#3}\middle\vert{#4}\right\}
		}
	}{
		\IfNoValueTF{#2}{
			\IfNoValueTF{#4}{
				\{{#3}\}
			}{
				\{{#3}\vert{#4}\}
			}
		}{
			\IfNoValueTF{#4}{
				#2\{{#4}#2\}
			}{
				#2\{{#3}#2\vert{#4}#2\}
			}
		}
	}
}

\makeatletter
\DeclareFontFamily{OMX}{MnSymbolE}{}
\DeclareSymbolFont{MnLargeSymbols}{OMX}{MnSymbolE}{m}{n}
\SetSymbolFont{MnLargeSymbols}{bold}{OMX}{MnSymbolE}{b}{n}
\DeclareFontShape{OMX}{MnSymbolE}{m}{n}{
    <-6>  MnSymbolE5
   <6-7>  MnSymbolE6
   <7-8>  MnSymbolE7
   <8-9>  MnSymbolE8
   <9-10> MnSymbolE9
  <10-12> MnSymbolE10
  <12->   MnSymbolE12
}{}
\DeclareFontShape{OMX}{MnSymbolE}{b}{n}{
    <-6>  MnSymbolE-Bold5
   <6-7>  MnSymbolE-Bold6
   <7-8>  MnSymbolE-Bold7
   <8-9>  MnSymbolE-Bold8
   <9-10> MnSymbolE-Bold9
  <10-12> MnSymbolE-Bold10
  <12->   MnSymbolE-Bold12
}{}

\let\llangle\@undefined
\let\rrangle\@undefined
\DeclareMathDelimiter{\llangle}{\mathopen}%
                     {MnLargeSymbols}{'164}{MnLargeSymbols}{'164}
\DeclareMathDelimiter{\rrangle}{\mathclose}%
                     {MnLargeSymbols}{'171}{MnLargeSymbols}{'171}
\makeatother

\begin{document}

\title{Approximation Algorithms for Noncommutative CSPs}
\author[1]{Eric Culf}
\author[2]{Hamoon Mousavi}
\author[3]{Taro Spirig}
\affil[1]{Faculty of Mathematics and Institute for Quantum Computing, University of Waterloo\footnote{\ttt{eculf@uwaterloo.ca}}}
\affil[2]{Simons Institute for Theoretical Computer Science, University of California at Berkeley\footnote{\ttt{hmousavi@berkeley.edu}}}
\affil[3]{QMATH, Department of Mathematical Sciences, University of Copenhagen\footnote{\ttt{tasp@math.ku.dk}}}
\maketitle

\begin{abstract}
\emph{Noncommutative} constraint satisfaction problems (NC-CSPs) are \emph{higher-dimensional operator} extensions of classical CSPs. Despite their significance in quantum information, their approximability remains largely unexplored. A notable example of a noncommutative CSP that is not solvable in polynomial time is NC-Max-$3$-Cut. We present a $0.864$-approximation algorithm for this problem. Our approach extends to a broader class of both classical and noncommutative CSPs. We introduce three key concepts: \emph{approximate isometry}, \emph{relative distribution}, and \emph{$\ast$-anticommutation}, which may be of independent interest.
\end{abstract}
\tableofcontents

\section{Introduction}\label{sec:introduction}
\subsection{Motivation}\label{sec:motivation-background} 
A notable example of a constraint satisfaction problem (CSP) is $\maxcut$, where we are given a graph $G = (V,E)$ and are asked to partition the vertices into two subsets such that the number of edges crossing the partition is maximized. It is natural to express this as a polynomial optimization, illustrated in Figure \ref{fig:maxcut-illustration}.

\begin{figure}[h]
\centering
\begin{subfigure}{0.49\textwidth}
    \centering
	\begin{tikzpicture}
        \node[circle,fill=green, inner sep=0.05cm] (1) at (0,0) {};
        \node[circle,fill=blue, inner sep=0.05cm] (2) at (1,1) {};
        \node[circle,fill=blue, inner sep=0.05cm] (3) at (1.366,-0.366) {};
        \node[circle,fill=green, inner sep=0.05cm] (4) at (2.41,1) {};
        \node[circle,fill=green, inner sep=0.05cm] (5) at (-1.22,0.71) {};
        \node[circle,fill=blue, inner sep=0.05cm] (6) at (-1.22,-0.71) {};
        \draw[] (1) -- (2);
        \draw[] (1) -- (3);
        \draw[] (2) -- (3);
        \draw[] (2) -- (4);
        \draw[] (1) -- (5);
        \draw[] (1) -- (6);
        \draw[thick, red, dashed] (2.45,0) -- (2.2,0.5) .. controls (1.5,2) and (0.75,2)  .. (0.5,0.5) .. controls (0.25,-1) and (-0.5,-1) .. (-1,-0.5) -- (-1.7,0.2);
        
        \node at (1) [above=0.7cm,red]{$-1$};
        \node at (3) [below right=0.2cm,red]{$+1$};
    \end{tikzpicture}
    \caption{Example of a cut}
\end{subfigure}
\hfill
\begin{subfigure}{0.49\textwidth}
    \begin{equation*}
    \openup\jot
    \begin{aligned}[t]
    \text{ maximize:}\quad &\sum_{(i,j) \in E} \frac{1-x_ix_j}{2} \\
            \text{subject to:}\quad & x_i \in \{-1,+1\}.
    \end{aligned}
    \end{equation*}
    \vspace{0cm}
    \caption{Max-Cut as a polynomial optimization}
\end{subfigure}
\caption{An assignment of $\pm 1$ to each variable $x_i$ indicates which side of the partition the corresponding vertex $i$ resides. Therefore, the expression $\tfrac{1-x_ix_j}{2}$ in the objective function is $1$ if the edge $(i,j)$ crosses the partition and zero otherwise.}
\label{fig:maxcut-illustration}
\end{figure}

What if instead of letting the variables represent unit scalars $\pm 1$, we let them represent unitary matrices\footnote{A unitary matrix is a matrix $X$ which satisfies $X X^* = X^* X=1$ where $X^*$ denotes the conjugate transpose of $X$ and $1$ denotes the identity matrix.} with $\pm 1$ eigenvalues? To get a number back out, we take their trace inner product. More precisely, we consider the following \emph{noncommutative optimization} 
\begin{equation}\label{eq:ncmaxcut-in-intro}
\openup\jot
\begin{aligned}[t]
\text{ maximize:}\quad &\sum_{(i,j)\in E} \frac{1-\bigip{X_i}{X_j}}{2} \\
        \text{subject to:}\quad & X_i\text{ unitary with eigenvalues }\pm 1.\\
\end{aligned}
\end{equation}
Here $\ip{X}{Y} = \tr(X^* Y)$ is the dimension-normalized Hilbert-Schmidt inner product. Importantly, we are putting no restriction on the dimension of the operators. Note that (classical) $\maxcut$ is the one-dimensional restriction. $\maxcut$ is $\NP$-hard, so surely the noncommutative variant should be at least as hard! Indeed, since the dimension is not bounded, it is not even clear if the noncommutative problem is decidable. Surprisingly, this turns out to be false as shown by Tsirelson \cite{tsirelson1,tsirelson2}: The value of Eq. \ref{eq:ncmaxcut-in-intro} can be computed in polynomial time. We present a proof of this fact in Section \ref{sec:approximate-isometries}. 

Figure \ref{fig:transition-max-cut} illustrates our knowledge of the complexity of both classical and noncommutative $\maxcut$. The $0.878$-approximation algorithm\footnote{The \emph{approximation ratio} of an algorithm $\mathcal{A}$ for a CSP such as $\maxcut$ is the quantity $\inf_I \frac{\mathcal{A}(I)}{\opt(I)}$
where $I$ ranges over all possible instances of the problem and $\opt(I)$ is the optimal value of the instance $I$. An $\alpha$-approximation is a polynomial-time algorithm with approximation ratio $\alpha$.} for $\maxcut$ due to Goemans and Williamson~\cite{goemansmaxcut} is known to be tight by Khot \emph{et al.} \cite{khot} assuming the Unique Games Conjecture (UGC). 

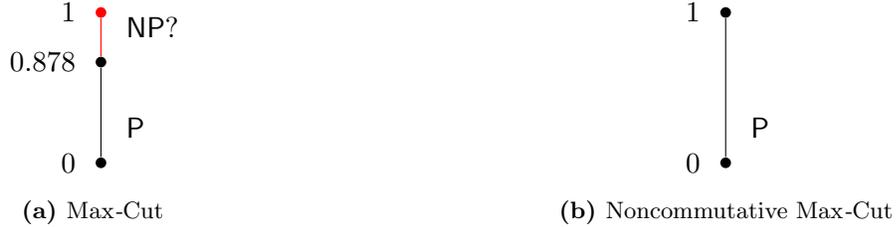
\begin{figure}[h]
\centering
\begin{subfigure}{0.49\textwidth}
    \centering
	\begin{tikzpicture}
        \node[circle,fill=black, inner sep=0.05cm] (1) at (0,0) {};
        \node[circle,fill=black, inner sep=0.05cm] (2) at (0,2*0.67) {};
        \node[circle,fill=red, inner sep=0.05cm] (3) at  (0,2)  {};
        \draw[] (1) -- (2);
        \draw[red] (2) -- (3);
        
        \node at (1) [left=0.2cm]{$0$};
        \node at (2) [left=0.2cm]{$0.878$};
        \node at (3) [left=0.2cm]{$1$};

        \node at (1) [above right =0.2cm]{$\Pp$};
        \node at (2) [above right =0.2cm]{$\NP$?};        
    \end{tikzpicture}
    \caption{$\maxcut$}
\end{subfigure}
\hfill
\begin{subfigure}{0.49\textwidth}
    \centering
	\begin{tikzpicture}
        \node[circle,fill=black, inner sep=0.05cm] (1) at (0,0) {};
        \node[circle,fill=black, inner sep=0.05cm] (2) at  (0,2)  {};
        \draw[] (1) -- (2);
        
        \node at (1) [left=0.2cm]{$0$};
        \node at (2) [left=0.2cm]{$1$};

        \node at (1) [above right =0.2cm]{$\Pp$};
    \end{tikzpicture}
    
    \caption{Noncommutative $\maxcut$}
\end{subfigure}
\caption{Transitions in complexity for classical and noncommutative $\maxcut$.}
\label{fig:transition-max-cut}
\end{figure}

Next, let us consider a simple modification of $\maxcut$ where our initial reaction regarding undecidability does end up being true. This modification is $\cut{3}$, where we are again given a graph and asked to partition the vertices into not two but three subsets such that the number of edges crossing the partitions is maximized. Some readers might have seen this problem in the context of $3$-Colouring instead: Indeed $G$ is $3$-colourable if and only if the value of $\cut{3}$ is $|E|$. This problem can also be written compactly as an optimization problem over ternary variables, as illustrated in Figure \ref{fig:maxcut-3-illustration}. 

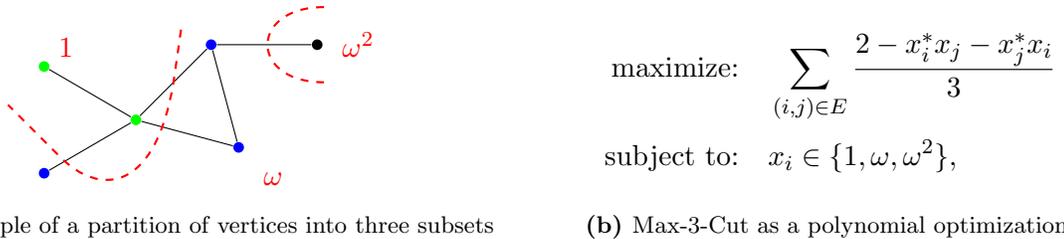
\begin{figure}[h]
\centering
\begin{subfigure}{0.49\textwidth}
    \centering
	\begin{tikzpicture}
        \node[circle,fill=green, inner sep=0.05cm] (1) at (0,0) {};
        \node[circle,fill=blue, inner sep=0.05cm] (2) at (1,1) {};
        \node[circle,fill=blue, inner sep=0.05cm] (3) at (1.366,-0.366) {};
        \node[circle,fill=black, inner sep=0.05cm] (4) at (2.41,1) {};
        \node[circle,fill=green, inner sep=0.05cm] (5) at (-1.22,0.71) {};
        \node[circle,fill=blue, inner sep=0.05cm] (6) at (-1.22,-0.71) {};
        \draw[] (1) -- (2);
        \draw[] (1) -- (3);
        \draw[] (2) -- (3);
        \draw[] (2) -- (4);
        \draw[] (1) -- (5);
        \draw[] (1) -- (6);
        \draw[thick, red, dashed] (0.6,1.2) -- (0.5,0.5) .. controls (0.25,-1) and (-0.5,-1) .. (-1,-0.5) -- (-1.7,0.2);
        \draw[thick, red, dashed] (2.5,1.5) .. controls (1.5,1.5) and (1.5,0.5) .. (2.5,0.5);
        
        \node at (1) [above left=0.7cm,red]{$1$};
        \node at (3) [below right=0.2cm,red]{$\omega$};
        \node at (4) [right=0.2cm,red]{$\omega^2$};
    \end{tikzpicture}
    \caption{Example of a partition of vertices into three subsets}
\end{subfigure}
\hfill
\begin{subfigure}{0.49\textwidth}
    \begin{equation*}
    \openup\jot
    \begin{aligned}[t]
    \text{ maximize:}\quad &\sum_{(i,j) \in E} \frac{2 - x_i^* x_j - x_j^*x_i}{3} \\
            \text{subject to:}\quad & x_i \in \{1,\omega,\omega^2\},
    \end{aligned}
    \end{equation*}
    \vspace{0cm}
    \caption{Max-$3$-Cut as a polynomial optimization}
\end{subfigure}
\caption{$\cut{3}$ and its presentation as a polynomial optimization: Here $1,\omega,\omega^2$ are the $3$rd roots of unity and $x^*$ is the complex conjugate. The term $\tfrac{2 - x_i^* x_j - x_j^*x_i}{3}$ is $1$ if $x_i\neq x_j$ and $0$ otherwise.}
\label{fig:maxcut-3-illustration}
\end{figure}

Now consider the noncommutative problem defined in the same way as before: Instead of $3$rd roots of unity, we allow unitaries with eigenvalues that are $3$rd roots of unity
\begin{equation}\label{eq:ncmax3cut-in-intro}
\openup\jot
\begin{aligned}[t]
\text{ maximize:}\quad &\sum_{(i,j)\in E} \frac{2-\bigip{X_i}{X_j}-\bigip{X_j}{X_i}}{3} \\
        \text{subject to:}\quad & X_i\text{ unitary with eigenvalues } 1,\omega,\omega^2.\\
\end{aligned}
\end{equation}
The value of this problem is known to be uncomputable following the work of Ji \cite{ji2013binary}, the celebrated $\MIP^*=\RE$ theorem of Ji et al. \cite{ji_mip_re}, and a recent work of Harris \cite{harris2023universality}:
\begin{theorem}[Ji, Ji \emph{et al.}, and Harris \cite{ji2013binary,ji_mip_re,harris2023universality}]\label{thm:harris} Given a graph $G = (V,E)$, figuring out whether the value of noncommutative $\cut{3}$ is $|E|$ or strictly less than $|E|$ is undecidable. 
\end{theorem}
But how hard is it to approximate the value of Eq. \ref{eq:ncmax3cut-in-intro}? For the classical $\cut{3}$, there is a $0.836$-approximation algorithm due to a sequence of works by Frieze and Jerrum \cite{frieze}, de Klerk \emph{et al.} \cite{klerk}, and Goemans and Williamson \cite{goemansmax3cut}. We show in this paper that we can actually do better in the case of the noncommutative $\cut{3}$:
\begin{theorem}\label{thm:main-motivation}
There exists a polynomial-time $0.864$-approximation algorithm for the value of noncommutative $\cut{3}$. 
\end{theorem}

Is this the best we can do? For the classical $\cut{3}$, if some plausible conjectures, including UGC, are true, then $0.836$ is the best approximation ratio (see Khot \emph{et al.} \cite{khot}). So the complexity of the classical problem is $\Pp$ up to the approximation ratio $0.836$ and it is conjectured to be $\NP$-hard beyond this constant. In the noncommutative setting, is a similar conjecture plausible?

\begin{question}\label{q:uncomputability}
Is approximating noncommutative $\cut{3}$ beyond the ratio $0.864$ uncomputable? Do we need some noncommutative extension of UGC to prove this?\footnote{We suspect that proving sharp complexity transitions for noncommutative $2$-CSPs (such as $\maxcut$ and $\cut{3}$) would require some additional complexity theoretic assumptions. This suspicion stems from the analogous development in the classical theory.}
\end{question}
Noncommutative CSPs can also be equivalently viewed as nonlocal games which are studied in quantum information theory (See Section \ref{sec:nonlocal_games} for more details). It is interesting that these natural problems from quantum information theory might be polynomial-time up to some constant approximation ratio and uncomputable beyond that constant. See Figure \ref{fig:transition-max-3-cut} where this possible \emph{sharp complexity transition} is illustrated both for classical and noncommutative $\cut{3}$.

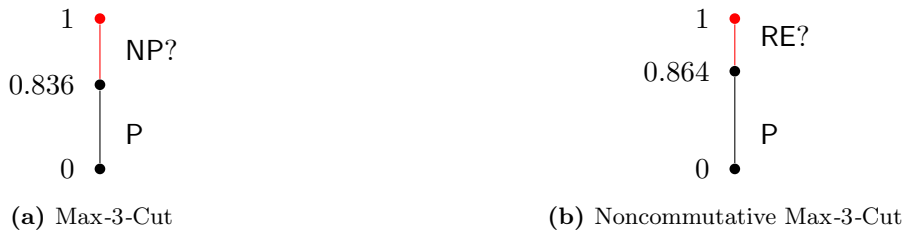
\begin{figure}[h]
\centering
\begin{subfigure}{0.49\textwidth}
    \centering
    \begin{tikzpicture}
        \node[circle,fill=black, inner sep=0.05cm] (1) at (0,0) {};
        \node[circle,fill=black, inner sep=0.05cm] (2) at (0,2*0.56) {};
        \node[circle,fill=red, inner sep=0.05cm] (3) at  (0,2)  {};
        \draw[] (1) -- (2);
        \draw[red] (2) -- (3);
        
        \node at (1) [left=0.2cm]{$0$};
        \node at (2) [left=0.2cm]{$0.836$};
        \node at (3) [left=0.2cm]{$1$};

        \node at (1) [above right =0.2cm]{$\Pp$};
        \node at (2) [above right =0.2cm]{$\NP$?};        
    \end{tikzpicture}
    \caption{$\cut{3}$}
\end{subfigure}
\hfill
\begin{subfigure}{0.49\textwidth}
    \centering
    \begin{tikzpicture}
        \node[circle,fill=black, inner sep=0.05cm] (1) at (0,0) {};
        \node[circle,fill=black, inner sep=0.05cm] (2) at (0,2*0.65) {};
        \node[circle,fill=red, inner sep=0.05cm] (3) at  (0,2)  {};
        \draw[] (1) -- (2);
        \draw[red] (2) -- (3);
        
        \node at (1) [left=0.2cm]{$0$};
        \node at (2) [left=0.2cm]{$0.864$};
        \node at (3) [left=0.2cm]{$1$};

        \node at (1) [above right =0.2cm]{$\Pp$};
        \node at (2) [above right =0.2cm]{$\RE$?};        
    \end{tikzpicture}
    
    \caption{Noncommutative $\cut{3}$}
\end{subfigure}
\caption{Conjectured transition in complexity for classical and noncommutative $\cut{3}$. $\RE$ is the class of problems reducible to the \emph{Halting problem}.}
\label{fig:transition-max-3-cut}
\end{figure}

$\maxcut$ and $\cut{3}$ are examples of $2$-CSPs. These are CSPs where every constraint involves only two variables at a time. By contrast, the phenomenon of complexity transition is better understood in the case of $3$-CSPs, both in the classical and noncommutative settings.\footnote{$k$-CSPs correspond to $k$-player nonlocal games. Read more about the connection to nonlocal games in Section \ref{sec:nonlocal_games}.} For example in the \emph{$3$-XOR problem}, we are given a system of linear equations over the boolean field such that every equation only involves three variables at a time. The simple random assignment satisfies half of the constraints in expectation. H\r{a}stad \cite{hastad} showed using the PCP theorem that it is NP-hard to do any better. 

Similar to what we did with $\maxcut$ and $\cut{3}$ one can define a natural noncommutative extension of $3$-XOR CSP.\footnote{The noncommutative extension of CSPs described in this paper, i.e. replacing scalars with unitary operators, cannot be directly applied for $3$-CSPs. This can be addressed but it is beyond the scope of this paper.} It turns out that the simple random assignment is the best we can do also in the noncommutative setting. O'Donnell and Yuen \cite{odonnell-yuen} showed, in a direct generalization of the PCP-based proof of H\r{a}stad, that approximating the noncommutative value beyond the ratio $1/2$ is uncomputable. See Figure \ref{fig:transition-3-xor}.

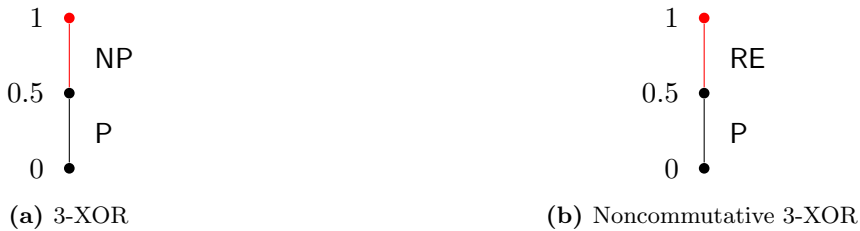
\begin{figure}[h]
\centering
\begin{subfigure}{0.49\textwidth}
    \centering
    \begin{tikzpicture}
        \node[circle,fill=black, inner sep=0.05cm] (1) at (0,0) {};
        \node[circle,fill=black, inner sep=0.05cm] (2) at (0,1) {};
        \node[circle,fill=red, inner sep=0.05cm] (3) at  (0,2)  {};
        \draw[] (1) -- (2);
        \draw[red] (2) -- (3);
        
        \node at (1) [left=0.2cm]{$0$};
        \node at (2) [left=0.2cm]{$0.5$};
        \node at (3) [left=0.2cm]{$1$};

        \node at (1) [above right =0.2cm]{$\Pp$};
        \node at (2) [above right =0.2cm]{$\NP$};        
    \end{tikzpicture}
    \caption{$3$-XOR}
\end{subfigure}
\hfill
\begin{subfigure}{0.49\textwidth}
    \centering
    \begin{tikzpicture}
        \node[circle,fill=black, inner sep=0.05cm] (1) at (0,0) {};
        \node[circle,fill=black, inner sep=0.05cm] (2) at (0,1) {};
        \node[circle,fill=red, inner sep=0.05cm] (3) at  (0,2)  {};
        \draw[] (1) -- (2);
        \draw[red] (2) -- (3);
        
        \node at (1) [left=0.2cm]{$0$};
        \node at (2) [left=0.2cm]{$0.5$};
        \node at (3) [left=0.2cm]{$1$};

        \node at (1) [above right =0.2cm]{$\Pp$};
        \node at (2) [above right =0.2cm]{$\RE$};        
    \end{tikzpicture}
    
    \caption{Noncommutative $3$-XOR}
\end{subfigure}
\caption{Transition in complexity for classical and noncommutative $3$-XOR. Unlike the examples of $\maxcut$ and $\cut{3}$ the transition in complexity is fully settled for both classical and noncommutative $3$-XOR.}
\label{fig:transition-3-xor}
\end{figure}

\subsection{Proof Idea: Approximate Isometries}\label{sec:approximate-isometries}
We introduce in this and the following two sections the main ideas that go into proving Theorem \ref{thm:main-motivation}. In this section, we generalize the construction of Tsirelson used to find optimal solutions for noncommutative Max-Cut.

It is well-known that for every $d$, there exist \emph{pairwise anticommuting} Hermitian and unitary operators $\sigma_1,\sigma_2,\ldots,\sigma_d$. That is in particular $\sigma_i\sigma_j + \sigma_j\sigma_i = 2\delta_{i,j}$. These are known as the \emph{Weyl-Brauer} operators and they can be represented as matrices in $M_{2^{O(d)}}(\C)$. Here $M_D(\C)$ is the algebra of $D$-by-$D$ matrices. It is known that the linear combination $$U(\vec{x}) \coloneqq x_1 \sigma_1 + \cdots + x_d \sigma_d$$ is a Hermitian unitary operator whenever $\vec{x} = (x_1,\ldots,x_d)$ is a real unit vector. This linear map $U$ is also an \emph{isometry}, that is $\ip{U(\vec{x})}{U(\vec{y})} = \ip{\vec{x}}{\vec{y}}$ for all vectors $\vec{x}$ and $\vec{y}$. We sometimes refer to this mapping as the \emph{vector-to-unitary construction}. Let us summarize what we have learnt so far in the following proposition.

\begin{proposition}[Isometric embedding of vectors into unitaries, Tsirelson]\label{prop:isometric-embedding} There exists a linear isometry $U:\R^d \to \mathrm{M}_{2^{O(d)}}(\mathbb{C})$ such that the image of any unit vector is a Hermitian unitary. 
\end{proposition}
The existence of (approximate) isometries similar to the kind introduced in this proposition directly corresponds to efficient algorithms for the CSPs of the previous section. To elucidate this point let us first recall the very simple proof of Tsirelson's theorem on noncommutative $\maxcut$ we advertised earlier.\footnote{This will be a slight modification of the statement and the proof of this remarkable theorem. The theorem was originally stated in the context of \emph{quantum value of XOR nonlocal games} \cite{tsirelson1,tsirelson2} which is a bit more technical to state and prove. The simpler statement given here is for the \emph{synchronous value of XOR games} albeit stated in the language of noncommutative CSPs. To see the connection between noncommutative CSPs and quantum and synchronous values of nonlocal games see Section \ref{sec:nonlocal_games}.} 

\begin{theorem}[Tsirelson's Theorem]\label{thm:tsirelson-intro} For any instance $G=(V,E)$ of $\maxcut$, the optimal value of noncommutative $\maxcut$
\begin{equation}\label{eq:nmaxcut-ideas}
\openup\jot
\begin{aligned}[t]
\text{ maximize:}\quad &\sum_{(i,j)\in E} \frac{1-\ip{X_i}{X_j}}{2} \\
        \text{subject to:}\quad & X_i^*X_i = X_i^2 = 1,\\
\end{aligned}
\end{equation} equals the optimal value of the SDP relaxation
\begin{equation}\label{eq:maxcut-sdp-ideas}
\openup\jot
\begin{aligned}[t]
\text{ maximize:}\quad &\sum_{(i,j)\in E} \frac{1-\ip{\vec{x}_i}{\vec{x}_j}}{2} \\
        \text{subject to:}\quad & \|\vec{x}_i\| = 1 \text{ and } \vec{x}_i \in \mathbb{R}^{|V|}.
\end{aligned}
\end{equation}
Since this SDP can be solved in polynomial time up to any desired accuracy, we can efficiently calculate the value of noncommutative $\maxcut$.
\end{theorem}
We leave it to the reader to verify that \eqref{eq:maxcut-sdp-ideas} is indeed an SDP and that its optimal value is always at least the optimal value of \eqref{eq:nmaxcut-ideas}. For the other direction, suppose $\vec{x}_i \in \R^{|V|}$ are a feasible solution to the SDP \eqref{eq:maxcut-sdp-ideas}. Now the set of Hermitian unitary operators $X_{i} \coloneqq U(\vec{x}_i)$ is a feasible solution to \eqref{eq:nmaxcut-ideas} due to the properties of the map $U$ in Proposition \ref{prop:isometric-embedding}. Also, by the isometric property of $U$, this noncommutative solution has the same objective value in \eqref{eq:nmaxcut-ideas} as the vectors in \eqref{eq:maxcut-sdp-ideas}.

We can now discuss at a high level the proof idea of Theorem \ref{thm:main-motivation} for noncommutative $\cut{3}$ \eqref{eq:ncmax3cut-in-intro}. Similar to $\maxcut$, we first need to come up with a natural SDP relaxation. This part is easy. The following simple modification of \eqref{eq:maxcut-sdp-ideas} is the canonical SDP relaxation for classical $\cut{3}$ used in all the previous work \cite{frieze,klerk,goemansmax3cut}
\begin{equation}\label{eq:vectorcut3-in-approximate-isometry}
\openup\jot
\begin{aligned}[t]
\text{ maximize:}\quad &\sum_{(i,j)\in E} \frac{2-\ip{\vec{x}_i}{\vec{x}_j} - \ip{\vec{x}_j}{\vec{x}_i}}{3} \\
        \text{subject to:}\quad & \|\vec{x}_i\| = 1 \text{ and } \vec{x}_i \in \mathbb{R}^{|V|},\\
                          \quad & \ip{\vec{x}_i}{\vec{x}_j} \geq -\frac{1}{2},
\end{aligned}
\end{equation}
and is easily shown to be an SDP relaxation of noncommutative $\cut{3}$ as well. We will prove a more general statement in Section \ref{sec:csp-definitions}. 

Much the same as in our discussion of Tsirelson's theorem on noncommutative $\maxcut$, we need a map that turns the vector solution of this SDP into a feasible solution of noncommutative $\cut{3}$
\begin{align*}
\vec{x}_1 &\mapsto X_1,\\
&\vdots\\
\vec{x}_{|V|} &\mapsto X_{|V|}.
\end{align*}
To preserve the objective value, we need the map to be an isometry similar to the one in Proposition \ref{prop:isometric-embedding}. However since $X_1,\ldots,X_{|V|}$ are a feasible solution for noncommutative $\cut{3}$ if they are unitaries with eigenvalues that are $3$rd roots of unity, we need this isometry to have an additional feature. In short we need an isometry from $\R^{|V|}$ to $\mathrm{M}_{D}(\mathbb{C})$, for some integer $D$, such that the image of any unit vector is a unitary with eigenvalues that are $3$rd roots of unity. In Section \ref{sec:gwb}, we show that such an isometry cannot exist (for any value of $D$, no matter how large). 

In the absence of such an isometry, we relax our requirements and search for \emph{approximate isometries}. The approximate isometries we will consider uses the power of randomness. Furthermore, we need a \emph{measure} with which we can evaluate the quality of any candidate approximate isometry. With these mind we proceed to define the notion of an $\alpha$-isometry. 

Since our goal is to ultimately solve $\cut{3}$ which has the objective function
\[\sum_{(i,j)\in E} \frac{2-\bigip{X_i}{X_j}-\bigip{X_j}{X_i}}{3}\]
a natural measure quickly emerges. Suppose $U_3:\R^d \to \mathrm{M}_{D}(\mathbb{C})$ is a randomized map with the property that the image of any unit vector is, with probability one, a unitary with eigenvalues that are $3$rd roots of unity. We say that $U_3$ is an $\alpha$-isometry, for some $\alpha \in [0,1]$, if for every pair of unit vectors $\vec{x}$ and $\vec{y}$ such that $\ip{\vec{x}}{\vec{y}}\geq -1/2$,\footnote{Since the vectors in any feasible solution of the SDP in Eq. \ref{eq:vectorcut3-in-approximate-isometry} have inner products $\geq -\frac{1}{2}$, it suffices in the definition of $\alpha$-isometry to consider only vectors $\vec{x}$ and $\vec{y}$ such that $\ip{\vec{x}}{\vec{y}} \geq -\frac{1}{2}$. This turns out to be crucial.} the quantity $$\Val \coloneqq \frac{2 - \Bigip{U_3(\vec{x})}{U_3(\vec{y})} - \Bigip{U_3(\vec{y})}{U_3(\vec{x})}}{3}$$ is, in expectation, at least $\alpha \cdot \val$, where $$\val \coloneqq \frac{2 - \ip{\vec{x}}{\vec{y}} - \ip{\vec{y}}{\vec{x}}}{3}.$$ Indeed if $U_3$ was an isometry then the two quantities $\val$ and $\Val$ would be the same so $\alpha = 1$. In this paper we give a construction of an $0.864$-isometry.

\begin{theorem}[Theorem \ref{thm:main-motivation} restated for approximate isometries]\label{thm:approximate-isometry} There exists a randomized map $U_3 : \R^d \to \mathrm{M}_{2^{O(d)}}(\C)$ that is a $0.864$-isometry such that every unit vector is mapped to a unitary with eigenvalues that are $3$rd roots of unity.
\end{theorem}

The $0.864$-approximation algorithm for the value of noncommutative $\cut{3}$ can be phrased simply as: On instance $I$, solve the SDP relaxation to obtain its optimal value $\sdp(I)$ and output $0.864\cdot\sdp(I)$. This is justified since the existence of an $0.864$-isometry guarantees that there exists a feasible solution with at least this value. Stated in the language of \emph{decision problems with a promise}, we have the following corollary
\begin{corollary}\label{cor:main-motivation} Given a graph $G = (V,E)$, it can be decided in polynomial-time whether the value of noncommutative $\cut{3}$ is $|E|$ or at most $0.864|E|$. 
\end{corollary}
Contrast this with Theorem \ref{thm:harris} on the undecidability of noncommutative $\cut{3}$.

Let us now sketch the construction of the randomized map $U_3$. The actual construction is more technical and is done in Algorithm \ref{alg:homlink}. The randomization is in the sampling of a Haar random unitary $R$ independent of the input vectors. Once a sample $R$ is fixed, then on a real vector $\vec{x} \in \R^d$, the following algorithm calculates $U_3(\vec{x})$. 

\begin{enumerate}
    \item First calculate the linear combination of Weyl-Brauer operators $X = x_1 \sigma_1 + \ldots + x_d \sigma_d$. This is guaranteed to be a unitary.
    \item Then let $\hat{X} = RX$.
    \item Finally, let $U_3(\vec{x})$ be the closest \emph{unitary with eigenvalues that are $3$rd roots of unity} to the unitary~$\hat{X}$.
\end{enumerate}

\begin{remark*}
    The algorithm above is not efficient. Nevertheless, $U_3$ does not need to be efficient for Corollary \ref{cor:main-motivation} to hold. The algorithm of Corollary \ref{cor:main-motivation} does not use this algorithm as a subroutine, but only requires its existence.
\end{remark*}

Step $1$ is exactly the same as the construction of the isometry by Tsirelson. Steps $2$ and $3$ are a noncommutative extension of the \emph{hyperplane randomized rounding scheme} of Goemans and Williamson. Thus our theorem is a {\bf simultaneous generalization} of these two results. The main innovation in proving Theorem \ref{thm:approximate-isometry}, \emph{i.e.} that $U_3$ is a $0.864$-isometry, is introduced in the next section. The reader interested in only seeing the proof idea of Theorem \ref{thm:approximate-isometry} can safely skip the rest of this section. 

We propose that the notion of an approximate isometry is an interesting mathematical primitive in its own right. Notably, there are numerous examples of approximate isometries in the classical literature on approximation algorithms, even if not explicitly framed as such. For instance, the Goemans-Williamson hyperplane rounding constructs a $0.878$-isometry of the form $\R^d \to \{\pm 1\}$. 

So far, we have focused on approximate isometries from real vector spaces. A natural question arises: does there exist an isometry from complex vector spaces? More specifically, is there an isometry $\C^d \to M_{D}(\C)$, for some $D$ possibly much larger than $d$, such that the image of any unit vector is a unitary operator? This question is particularly relevant when designing approximation algorithms for CSPs where the canonical SDP relaxation is complex. This includes several interesting CSPs, such as Unique Games.

Bri\"{e}t, Regev, and Saket \cite{iop_tight_hardness} showed that the image of $U$ from Proposition~\ref{prop:isometric-embedding} can be very far from being unitary, when $U$'s domain is extended to complex unit vectors. In fact, Kretschmer \cite{kretchmer} showed that no (exact) isometry from a complex vector space exists for any $D$. Intuitively, this suggest that \emph{unitaries cannot fully capture complex inner product spaces}. This then raises the question: 
\begin{question}\label{q:isometry-complex}
What is the best approximate isometry from complex vector spaces? 
\end{question}
This is equivalent to Question \ref{q:unbounded-little-grothendieck-inequality}, which we explored in the section titled ``Grothendieck Inequalities.''Approximate isometries have strong ties with Grothendieck inequalities, a topic further explored in that section (Section \ref{sec:grothendieck}). 

\subsection{Innovation: Relative Distribution}\label{sec:innovation-relative-distribution}
Our proof relies on the concept of \emph{relative distribution}, which we informally introduce in this section. To prove that $U_3$ is an approximate isometry, we need to argue that $\bigip{U_3(\vec{x})}{U_3(\vec{y})}$ is close to $\lambda\coloneqq \ip{\vec{x}}{\vec{y}}$ for any two vectors $\vec{x},\vec{y} \in \R^d$ such that $\ip{\vec{x}}{\vec{y}} \in [-1/2,1]$. The algorithm for $U_3$ has three steps; of these only the last step may change the inner product. 

To elaborate, the first step in calculating $U_3$ is to compute unitary operators $X \coloneqq U(\vec{x})$ and ${Y \coloneqq U(\vec{y})}$ where $U$ is the isometry introduced in Proposition \ref{prop:isometric-embedding}. Since $U$ is an isometry ${\bigip{X}{Y} = \lambda}$. The second step randomizes $X$ and $Y$ by premultiplying them with a Haar random unitary $R$, which also preserves the inner product since $\bigip{RX}{RY} = \ip{X}{Y} = \lambda$. Finally, in the third step we round $RX$ and $RY$ to the nearest order-$3$ unitary operators.\footnote{$X$ is an order-$3$ operator if $X^3$ is the identity. A unitary is order-$3$ if and only if all its eigenvalues are $3$rd roots of unity.} This is where the inner product may change. We need to show that it does not change by too much. 

Next, we discuss in more detail what goes into the rounding in Step $3$. It is known that the nearest order-$3$ unitary in Frobenius norm to any unitary $A$ is simply obtained by rounding every eigenvalue of $A$ to the nearest $3$rd root of unity \cite{hoffman}. We use the notation $\tilde{A}$ to refer to the nearest order-$3$ unitary to $A$. 

With this notation we write $U_3(\vec{x}) = \tilde{RX}$ and $U_3(\vec{y}) = \tilde{RY}$. We show in this paper that in expectation (over the Haar random unitary $R$) the inner product $\ip{\tilde{RX}}{\tilde{RY}}$ is close to $\lambda$. More precisely we prove that
\begin{equation}\label{eq:main-eq-in-innovation}
0.864(2-2\lambda) \leq \expect_R \Bigparen{2 - \ip{\tilde{RX}}{\tilde{RY}} - \ip{\tilde{RY}}{\tilde{RX}}}.
\end{equation}
This inequality implies that $U_3$ is a $0.864$-isometry by the definition of an $\alpha$-isometry.
In order to prove \eqref{eq:main-eq-in-innovation}, we need to understand the distribution of pairs of unitaries $(RX,RY)$ where $X$ and $Y$ are some fixed unitaries and $R$ is a unitary sampled according to the Haar measure. We call this distribution the \emph{fixed inner product distribution} of $X$ and $Y$. 

How do we analyse the rounding step in $U_3$? By the discussion of the previous paragraph, the key to proving \eqref{eq:main-eq-in-innovation} is to understand the joint probability distribution of the eigenvalues of pairs of operators $(RX,RY)$. It turns out that we can make do with less: We only need the distribution of the angle between random eigenvalues of $RX$ and $RY$. 

Sample an eigenvalue $\alpha$ of $RX$ and an eigenvalue $\beta$ of $RY$ with probability equal to $\ip{P_\alpha}{Q_\beta}$ where $P_\alpha$ (resp. $Q_\beta$) is the projection onto the $\alpha$-eigenspace of $RX$ (resp. $\beta$-eigenspace of $RY$). This is a well-defined probability distribution on pairs of eigenvalues of $RX$ and $RY$ --- in particular, the inner product of any pair of projectors is in the interval $[0,1]$. We call the random variable $\theta \coloneqq \beta - \alpha$ the \emph{relative angle} and call its probability distribution the \emph{relative distribution}. The randomness is both in the choice of $R$ and the distribution of pairs of eigenvalues of $RX$ and $RY$ we just defined. The definition of relative distribution is formalized in Section \ref{sec:relative-distribution} using measure theory. We denote the relative distribution of $(X,Y)$ by $\Delta_{X,Y}$.

A simple calculation shows that the expected value of $\exp(i\theta)$, where $\theta$ is the relative angle, is the inner product $\lambda$. Remarkably we are able to say much more: In large dimension, the distribution $\Delta_{X,Y}$ (and not just its mean) depends only on the inner product $\lambda$. Hence, we sometimes write $\Delta_\lambda$ instead of $\Delta_{X,Y}$. Indeed it turns out that the relative distribution is a \emph{wrapped Cauchy distribution} with parameters only depending on $\lambda$. This is captured in the following theorem, whose proof is given in Section \ref{sec:relative-distribution-definition} using tools from free probability.

\begin{theorem}[Cauchy law, informal]\label{thm:cauchy-law} As the dimension of the matrices $X$ and $Y$ tends to infinity, the relative distribution of $X$ and $Y$ converges to the wrapped Cauchy distribution with peak angle $\theta_0 = \measuredangle \lambda$ and scale factor $\gamma = -\ln|\lambda|$ where $\lambda \coloneqq \ip{X}{Y}$. 
\end{theorem}
Here $\measuredangle \lambda$ denotes the phase of $\lambda \in \C$. This theorem is given formally as Theorem~\ref{thm:relative-distribution}.

The probability distribution function (PDF) of the wrapped Cauchy distribution with peak angle $\theta_0$ and scale factor $\gamma$ is
\[\frac{1}{2\pi} \frac{\sinh(\gamma)}{\cosh(\gamma) - \cos(\theta - \theta_0)}.\]
In Figure \ref{fig:relative-in-innovation}, we plot the PDF of the relative distribution $\Delta_\lambda$ for some values of $\lambda$. 

\begin{figure}[h]
\includegraphics[scale=.6]{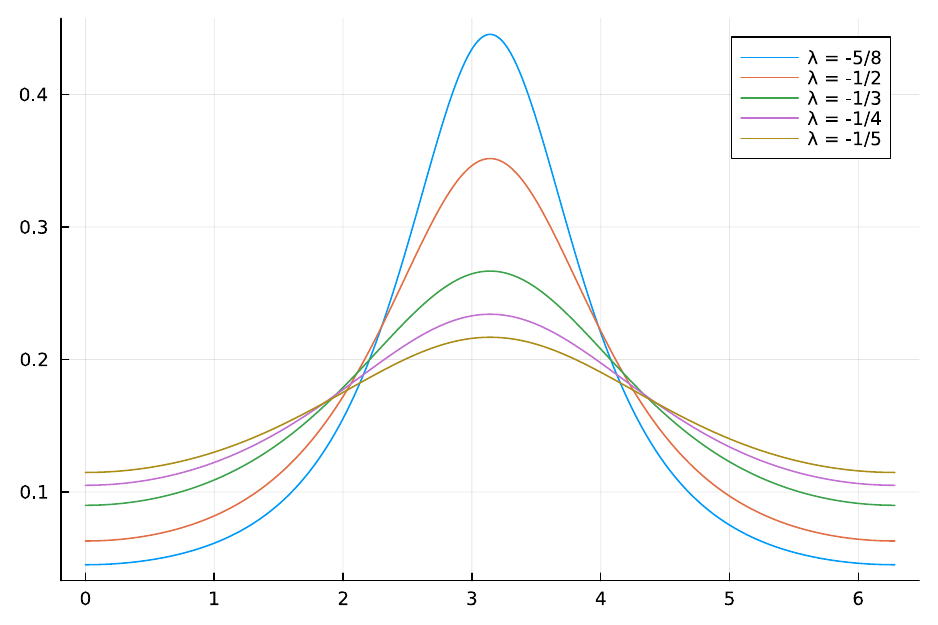}
\centering
\caption{PDF of $\Delta_\lambda$ for some values of $\lambda$. The horizontal axis is the interval $[0,2\pi)$. The distributions are peaked at the angle $\measuredangle \lambda$ which for a negative real number $\lambda$ is always $\pi$.}\label{fig:relative-in-innovation}
\end{figure}

The relative distribution not only helps in estimating the expectation $\expect_R \ip{\tilde{RX}}{\tilde{RY}}$ appearing in \eqref{eq:main-eq-in-innovation}, it allows us to calculate much more complicated expectations $\expect_R \tr\parens[\big]{P(\tilde{RX},\tilde{RY})}$ for any $*$-polynomial of the form $P(x,y) = \sum_{s,t} w_{s,t} (x^*)^sy^t$. In fact, we show that there exists a function $\fid_P$ defined on the interval $[0,2\pi)$, called the fidelity of $P$, such that
\begin{align}\label{eq:integral-formula}
    \expect_R \tr\parens[\big]{P(\tilde{RX},\tilde{RY})}=\int_0^{2\pi}\fid_P(\theta)d\Delta_{X,Y}(\theta),
\end{align}
for every $X$ and $Y$. Since $\Delta_{X,Y}$ only depends on the inner product $\ip{X}{Y}$, the expectation also depends only on the inner product. The result is presented formally as Theorem~\ref{lem:noncommutative-fid-delta-integral} in Section~\ref{sec:fidelity-integral-formula}. 

\begin{center}
\emph{Because of the generality of this integral representation, the proposed framework that combines approximate isometry and relative distribution goes beyond just the analysis of noncommutative $\cut{3}$.}
\end{center}

\vspace{0.2cm}

\noindent\textbf{Vector relative distribution. } We also give a \emph{vector analogue} of the (operator) relative distribution we introduced here which, as one might expect, has applications to classical CSPs. Indeed, the vector relative distribution is implicit in the analysis of approximation algorithms for classical CSPs, for example in \cite{goemansmax3cut,newman}. Using this distribution, and an argument much similar to the analysis of the noncommutative problem, we can recover all the approximation ratios in the work by Newman \cite{newman} (as well as the one in \cite{goemansmax3cut}) for classical $\cut{k}$. 
\begin{center}
\emph{Indeed, the idea of relative distribution simplifies the analysis of these algorithms.}    
\end{center}
We define this distribution and prove an analogue of the Cauchy Law (Theorem~\ref{thm:cauchy-law}) in Appendix~\ref{sec:classical-rel-dist}. 
\begin{remark*}
The preceding discussion indicates that a unified approximation framework may exist for both classical and noncommutative CSPs. 
\end{remark*}

\subsection{$\ast$-Anticommutation}\label{sec:innovation-generalized-anticommutation}
The primary application of the Cauchy Law for us is in the construction of approximate isometries. However, as noted in the statement of Theorem \ref{thm:cauchy-law}, this proof technique requires taking the limit of large dimension. Effectively, this leads to a construction of an approximate isometry $U_3:\R^d \to M_{D}(\C)$ with $D = \infty$. That is still sufficient for Corollary \ref{cor:main-motivation}, as we can approximate the value through a sequence of finite-dimensional solutions. But, can we prove the Theorem \ref{thm:approximate-isometry} with $D = 2^{O(d)}$? Note that an exponential blow-up is necessary even in the case of Proposition \ref{prop:isometric-embedding}, see for example \cite{Slofstra_2011}.

We show that this is possible. To achieve this we generalize the Weyl-Brauer operators introduced in the previous section. In \cite{generalizedchsh}, a variant of the anticommutation relation, $X^*Y = - Y^*X$, was used to construct optimal noncommutative solutions to some examples of CSPs. We refer to this relationship as \emph{$\ast$-anticommutation}. This reduces to usual anticommutation when the operators $X$ and $Y$ are Hermitian. 

We show that for all integers $k$ and $d$ there exist $d$ order-$k$ unitaries $\sigma_1,\ldots,\sigma_d$ that pairwise satisfy the $\ast$-anticommutation relation $\sigma_i^* \sigma_j = -\sigma_j^*\sigma_i$ for all $i\neq j$. We call these the \emph{generalized Weyl-Brauer operators} (GWB) and denote the group they generate by $\gwb_d^k$. 
\begin{theorem}\label{thm:gwp-exists}
Generalized Weyl-Brauer operators exist for every $k$ and $d$. Moreover these operators can be represented on a Hilbert space of dimension $2^{O(kd)}$.
\end{theorem}
The proof represents one of the key technical challenges we addressed in this work. This and the Cauchy Law (Theorem~\ref{thm:cauchy-law}) comprise the technical backbone of this work.

We conclude this section by highlighting a property of GWB that plays an important role in our proofs. Similar to Weyl-Brauer operators (see Proposition \ref{prop:isometric-embedding}), when $\sigma_1,\ldots,\sigma_d$ are generalized Weyl-Brauer operators, the linear combination \[U^{(k)}(\vec{x}) \coloneqq x_1 \sigma_1 + \cdots + x_d \sigma_d\] remains a unitary operator for any real unit vector $\vec{x}$, and the mapping is isometric. Moreover $U^{(k)}(\vec{x})$ is almost an order-$k$ unitary. Together, this means that the linear combination of generalized Weyl-Brauer operators nearly preserves all the properties of the original operators. See Figure \ref{fig:almost-order-k} for an illustration of this. This is made formal in Section~\ref{sec:construct-gwb} and proved in Corollary \ref{cor:fin-dim-weak-k-power}.

\begin{figure}[h]
\centering
\begin{subfigure}{0.49\textwidth}
    \includegraphics[scale=.5]{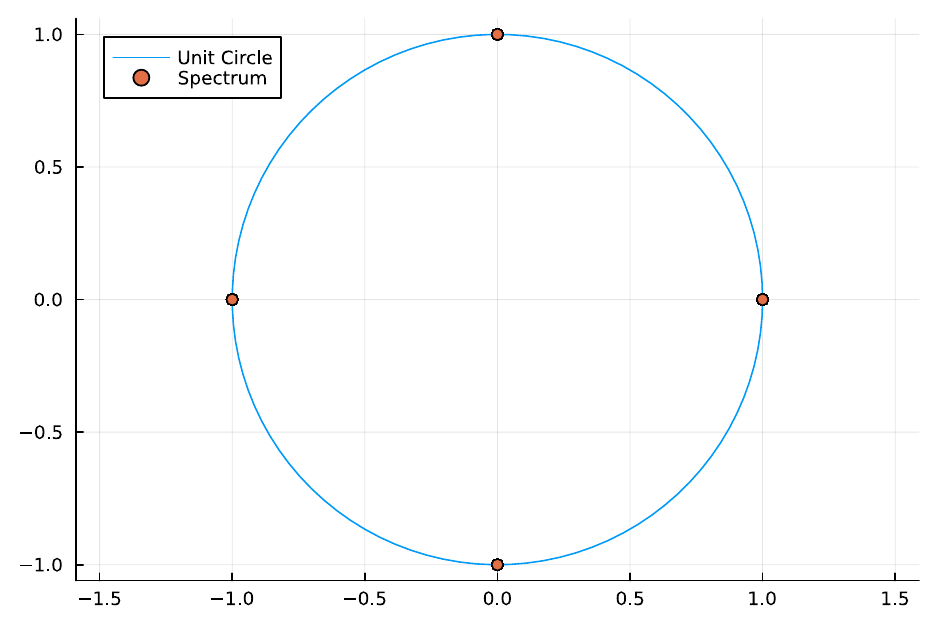}
    \caption{Spectrum of an operator in $\gwb_d^4$}
\end{subfigure}
\hfill
\begin{subfigure}{0.49\textwidth}
    \includegraphics[scale=.5]{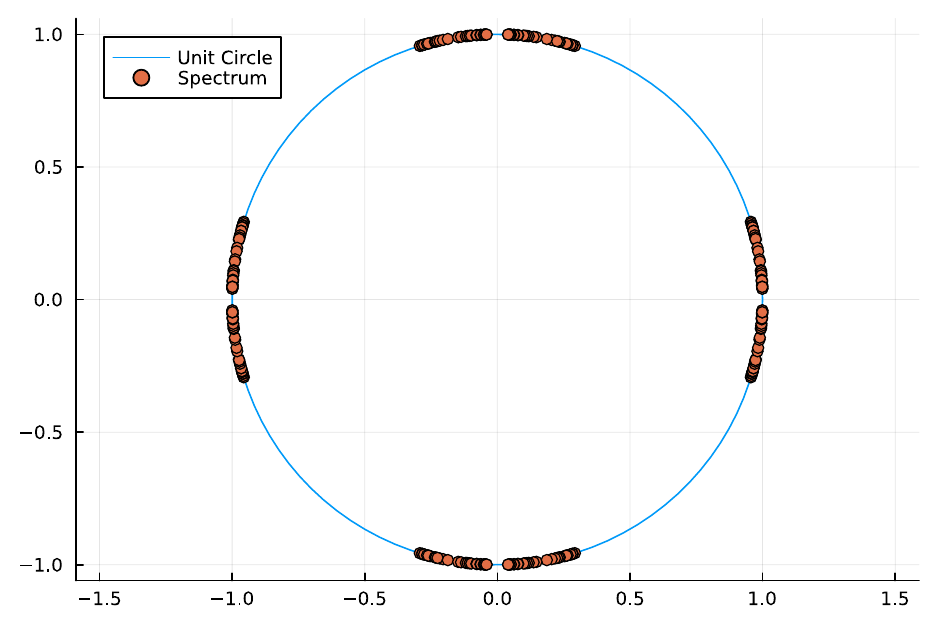}
    \vspace{0cm}
    \caption{Spectrum of a linear combination of $\gwb_d^4$ operators}
\end{subfigure}
\caption{Linear combination of order-$4$ generalized Weyl-Brauer operators are unitaries that are \emph{almost} order-$4$: their spectrum is a perturbation of $4$th roots of unity.}
\label{fig:almost-order-k}
\end{figure}

\subsection{Future Directions}\label{sec:open-problems}
This work opens up numerous new directions for further research. Among these, we think the following could be most helpful in improving our understanding of CSPs: 

\begin{enumerate}
\item {\bf Gapped Hardness}: Deciding if the (normalized) value of noncommutative $\cut{3}$ is $1$ or strictly less than $1$ is $\RE$-hard \cite{ji2013binary,ji_mip_re,harris2023universality}. In fact, it is also $\coRE$-hard \cite{pi2}. This means that the problem is strictly harder than even the Halting problem. Does there exist a constant $c$, such that deciding if the value is $1$ or at most $1-c$ is still $\RE$-hard? In this paper, we show that deciding whether the value is $1$ or at most $0.864$ can be done in polynomial time. Could this constant $c$ match the gap of our algorithm? In other words, could it be that $c = 1 - 0.864$? See Figure~\ref{fig:transition-max-3-cut}.

\item {\bf Integrality Gaps}: What is the integrality gap (Definition \ref{def:integrality-gap}) of the canonical SDP relaxation~\eqref{eq:vectorcut3-in-approximate-isometry} for noncommutative $\cut{3}$~\eqref{eq:ncmax3cut-in-intro}? Does the integrality gap match the approximation ratio of $0.864$? Instances of integrality gap arguments in the noncommutative setting appear in~\cite{iop_tight_hardness,bandeira,quantum-max-cut-integrality-gap}. 

Note that a positive answer to the gapped hardness question (the previous problem) would imply that $0.864$ is the best approximation ratio among all algorithms for $\ncut{3}$. In comparison, in this question, we ask a weaker question: Is $0.864$ the best approximation ratio among all algorithms based on the canonical SDP relaxation?

\item {\bf Rounding Higher-Level SDPs}: There is a whole hierarchy of SDP relaxations for noncommutative polynomial optimization~\cite{NPA}. Here we gave a rounding scheme from the first level of this SDP hierarchy. Could we achieve a better approximation ratio using the second level? A similar hierarchy exists in the commutative setting and Raghavendra \cite{prasad} showed that in the case of classical CSPs, assuming UGC, the first level always gives the best approximation ratio. Could we expect noncommutative CSPs to behave differently? 

\item {\bf Rounding to Classical Solution}: Given a noncommutative solution to a CSP, can we extract a good classical solution from it? More precisely, does there exist a rounding from noncommutative solutions to classical solutions such that, if the value of the noncommutative solution is $\Val$, the value of the rounded classical solution is at least $\alpha \cdot \Val$ for some constant~$\alpha$?

All the best known algorithms for classical CSPs are based on rounding vector solutions of SDPs. Could rounding operator solutions do as well or even better? See Section \ref{sec:brave-new-world} for more discussion on this.

A related question is: What is the ratio between the optimal classical value and the optimal noncommutative value for the CSPs discussed in this paper? Similar questions have been asked in the context of nonlocal games for the ratio of classical and quantum values, see for example~\cite{briet2009multiplayer,rosicka2024constructive}.

\item {\bf Other CSPs}: The construction of approximate isometries discussed earlier crucially relies on the vectors being real. Thus, as it currently stands, our algorithm works only for the classes of CSPs introduced in Section \ref{sec:csp-definitions}, \emph{i.e.} \emph{homogeneous CSPs} and \emph{smooth CSPs}. These are subclasses of linear $2$-CSPs where the canonical SDP relaxation can always be assumed to be real. Can we extend the framework of approximate isometries to all linear $2$-CSPs? See also Question~\ref{q:isometry-complex} and Section~\ref{sec:grothendieck}.

\item {\bf Non-synchronous Regime}: From the perspective of nonlocal games, our algorithm approximates the \emph{synchronous quantum value}. Could this framework be extended to also approximate the \emph{quantum value}? See Section~\ref{sec:nonlocal_games} where we define the quantum and synchronous quantum values of a nonlocal game.

\item {\bf Finite-Dimensionality}: We mentioned that noncommutative $\cut{3}$ is $\coRE$-hard. This implies that there is an instance for which the optimal value is not attained in any finite dimensional Hilbert space. Can we explicitly construct an example? 

Explicit examples of general nonlocal games that exhibit this remarkable behavior are known. See, for example~\cite{slofstra_set_of_quantum_correlations,Coladangelo_separation,coladangelo-inherently-infinite}.

In contrast, for noncommutative $\maxcut$, there always exists a finite-dimensional optimal solution. In fact, an optimal solution can always be found in a Hilbert space of exponential dimension~\cite{Slofstra_2011}.

\item {\bf Sum-of-Squares}: We also mentioned that $\cut{3}$ is $\RE$-hard. This means that there is an instance for which no sum-of-squares can certify the optimal value. Can we construct an explicit example?

In contrast, for noncommutative $\maxcut$, a sum-of-squares of quadratic degree is sufficient to certify the optimal value.

\item {\bf Parallel Repetition}: The $n$-th \emph{parallel repetition} of a nonlocal game $\mc{G}$ is the game $\mc{G}^n$ where the referee samples $n$ pairs of question $(i_1,j_1),\ldots,(i_n,j_n)$ independently and sends $(i_1,\ldots,i_n)$ to Alice and $(j_1,\ldots,j_n)$ to Bob. Alice and Bob each respond with $n$ answers and they win if and only if they would have won each game individually. 

Let $\val(\mc{G})$ denote the quantum value of the game. We know that $\val(\mc{G}^n) = \val(\mc{G})^n$ when $\mc{G}$ is an XOR game \cite{cleve2008perfect}. We also know that for Unique Games $(1-\eps)^n \leq \val(\mc{G}^n) \leq (1-\eps/4)^n$ when $\val(\mc{G}) = 1-\eps$ \cite{kempe}. Let us call an inequality of this form the \emph{strong parallel repetition} property. Do games arising from $\cut{k}$ CSPs also satisfy the strong parallel repetition property? 

This amounts to showing that the SDP relaxation of these CSPs satisfies a certain \emph{multiplicative property}. See \cite{feige-parallel,cleve2008perfect,mittal-product,kempe} for more on this.
\end{enumerate}

This work offers additional open problems. They appear in the paper as Questions \ref{q:uncomputability}, \ref{q:isometry-complex}, \ref{q:dimension-efficient-algorithm}, \ref{q:reduce-randomness}, \ref{q:efficient-computation-of-spectrum}, \ref{q:efficient-implementation-on-quantum-computers}, \ref{q:best-ratio-for-cut-5}, \ref{q:noncommutative-smooth-csp}, \ref{q:classical-smooth-csp}, \ref{q:relative-dist-for-3}, \ref{q:which-ncsps-are-approximable}, \ref{q:complex-grothendieck-for-csp-laplacian}, \ref{q:grothendieck-constant-for-order-k-problems}, \ref{q:unbounded-little-grothendieck-inequality}, and \ref{q:order-k-unbounded-little-grothendieck}.

\subsection{Overview}

A sketch of the proof of our result for $\cut{3}$ is given in Section \ref{sec:case-study}. The full proof is given in Section \ref{sec:algorithm}. 

In Section \ref{sec:csp-definitions} we introduce the classes of CSPs we consider in this work together with their canonical SDP relaxations. \emph{Homogeneous CSPs} are introduced in Section \ref{sec:def-homogeneous-csps} together with our quantitative results about them. The analysis of our algorithm for this class is given in Section \ref{sec:algorithm-homogeneous}. Similarly, \emph{smooth CSPs} are introduced in Section \ref{sec:smooth-csps} and the analysis of our algorithm is given in Section \ref{sec:algorithm-smooth}. 

The main ingredient in the analysis of the algorithm, \emph{relative distribution}, is defined and characterized in Section \ref{sec:relative-distribution-definition}. The proof of the Cauchy Law is given in Theorem \ref{thm:relative-distribution}. We see how to apply the Cauchy Law in Section \ref{sec:fidelity-integral-formula}.

The aforementioned sections comprise the \emph{analytic approach} to constructing approximate isometries. Sections \ref{sec:gwb}, \ref{sec:algebraic-reldist}, and \ref{sec:efficient-algorithm}, develop the \emph{algebraic approach}. Section \ref{sec:gwb} constructs representations of the operators used in this approach, \emph{i.e.} \emph{generalized Weyl-Brauer} operators. Section \ref{sec:algebraic-reldist} introduces the notion of \emph{algebraic relative distribution}, which allows us to bring tools from the analytic approach to this context. In Section \ref{sec:efficient-algorithm}, the dimension-efficient algorithm based on these algebraic ideas is presented.

In Section \ref{sec:connection-with-previous-work}, we explore the connection between our work with related topics. In Section \ref{sec:brave-new-world} we give a brief history of approximation algorithms and hardness results known for classical CSPs. There we explore the possibility of a deeper link between classical CSPs and their noncommutative variants. In Sections \ref{sec:nonlocal_games} and \ref{sec:entangled-unique-games}, we first give the nonlocal game viewpoint on noncommutative CSPs, then present previous algorithmic results on approximating their values. In Section \ref{sec:grothendieck} we discuss the close ties between our work on noncommutative CSPs and the framework of Grothendieck inequalities. Finally Section \ref{sec:quantum-max-cut} discusses a famous variant of noncommutative Max-Cut which is an important topic studied in Hamiltonian complexity.

The technical background and notations are collected in the preliminaries in Section \ref{sec:prilim}.

\paragraph{Acknowledgments.} 
We thank Tarun Kathuria and Nikhil Srivastava for pointing to us the connection with free probability. We also thank Henry Yuen for many valuable comments on the first version of this paper. EC is supported by a CGS-D from Canada's NSERC. EC thanks Richard Cleve and William Slofstra for their invaluable support; and Archishna Bhattacharyya, Alex Frei, and Romi Lifshitz for helpful discussions. HM is supported by DOE NQISRC QSA grant \#FP00010905. This material is based upon work supported by the U.S. Department of Energy, Office of Science, National Quantum Information Science Research Centers, Quantum Systems Accelerator. Part of this work was done while TS was at MIT and supported by the Hasler Foundation. TS thanks Anand Natarajan for hosting him during that time. TS is supported by the European Union under the Grant Agreement No 101078107, QInteract and VILLUM FONDEN via Villum Young Investigator grant (No 37532) and the QMATH Centre of Excellence (Grant No 10059).

\section{Proof Sketches}\label{sec:case-study}
The goal of this section is to see the main ideas of the paper in action on the example of noncommutative $\cut{3}$ (or $\ncut{3}$ for short). For a full formal treatment of these ideas and the way they are applied to the more general settings of homogeneous and smooth CSPs see Sections \ref{sec:gwb}, \ref{sec:relative-distribution}, \ref{sec:algorithm}, \ref{sec:algebraic-reldist} and \ref{sec:efficient-algorithm}. 

Recall that $\ncut{3}$ is the optimization problem
\begin{equation}\label{eq:ncut3}
\openup\jot
\begin{aligned}[t]
\text{ maximize:}\quad &\sum_{i,j=1}^N w_{ij} \frac{2 - \ip{X_i}{X_j} - \ip{X_j}{X_i}}{3} \\
        \text{subject to:}\quad & X_i^*X_i = X_i^3 = 1.
\end{aligned}
\end{equation}
Here the graph instance has $N$ vertices and there is a weight $w_{ij} \geq 0$ on edge $(i,j)$. The SDP relaxation is
\begin{equation}\label{eq:vectorcut3}
\openup\jot
\begin{aligned}[t]
\text{ maximize:}\quad &\sum_{i,j=1}^N w_{ij} \frac{2-\ip{\vec{x}_i}{\vec{x}_j} - \ip{\vec{x}_j}{\vec{x}_i}}{3} \\
        \text{subject to:}\quad & \|\vec{x}_i\| = 1 \text{ and } \vec{x}_i \in \mathbb{R}^N,\\
                          \quad & \ip{\vec{x}_i}{\vec{x}_j} \geq -\frac{1}{2}.
\end{aligned}
\end{equation}
We see in Section \ref{sec:csp-definitions} that this is indeed a relaxation. 

To solve $\ncut{3}$, our first attempt is to solve the SDP, then apply Tsirelson's \emph{vector-to-unitary construction} (that is the linear map $U$ in Proposition \ref{prop:isometric-embedding}) to the SDP vectors. Doing so we obtain Hermitian unitary operators, but we need order-$3$ unitary operators to be feasible in $\ncut{3}$. To fix this issue, we may try and modify the vector-to-unitary construction by switching the Weyl-Brauer operators with a different set of operators. Namely, we could look for operators $\sigma_1,\ldots,\sigma_N$ such that, for real vectors $\vec{x} = (x_1,\ldots,x_N) \in \R^n$, the operators $U_{\vec{x}} \coloneqq \sum x_i \sigma_i$ satisfy the following two properties:
\begin{enumerate}
    \item Whenever $\vec{x}$ is a unit vector, $U_{\vec{x}}$ is an order-$3$ unitary operator , \emph{i.e.} $U_{\vec{x}}^3 = U_{\vec{x}}^*U_{\vec{x}} = 1$.
    \item The map $\vec{x}\mapsto U_{\vec{x}}$ preserves the inner product, \emph{i.e.} $\ip{U_{\vec{x}}}{U_{\vec{y}}} = \ip{\vec{x}}{\vec{y}}$ for all $\vec{x},\vec{y} \in \R^N$.
\end{enumerate}
However, as we see in Section~\ref{sec:v-u-intro}, no such set of operators $\sigma_1,\ldots,\sigma_N$ exists.

\subsection{Analytic Approach}\label{sec:case-study-analytic-approach}
Nevertheless, let $\vec{x}_1,\ldots,\vec{x}_N$ be any feasible SDP solution in \eqref{eq:vectorcut3} and apply Tsirelson's vector-to-unitary construction to obtain unitary operators $X_1,\ldots,X_N$. These operators, by the isometry property, satisfy $\ip{X_i}{X_j} = \ip{\vec{x}_i}{\vec{x}_j}$; next we round these unitaries to the nearest order-$3$ unitaries. In Frobenius norm, the closest order-$3$ unitary to a unitary $X$ is obtained by rounding the eigenvalues of $X$ to the nearest $3$-rd roots of unity. That is, if $X = \sum_{s} \lambda_s P_s$ is the spectral decomposition of $X$ where $\lambda_s$ are eigenvalues and $P_s$ are projections onto the corresponding eigenspaces, then $\tilde{X} = \sum_{s} \tilde{\lambda}_s P_s$ is the closest order-$3$ unitary to $X$ in Frobenius norm, where $\tilde{\lambda}_s$ is the closest $3$-rd root of unity to $\lambda_s$. Unfortunately, as is, this rounding scheme recovers an approximation ratio that is even less than the ratio $0.836$ of the Goemans and Williamson \cite{goemansmax3cut} algorithm for classical $\cut{3}$. 

To improve on this we use randomization. Sample a unitary $R$ from the Haar measure and let $U_3(\vec{x}_i)$ be the closest order-$3$ unitary to $RX_i$, that is $U_3(\vec{x}_i) = \tilde{RX_i}$. We prove that
\begin{equation}\label{eq:inequality_cut3}
\expect [2 - \ip{U_3(\vec{x}_i)}{U_3(\vec{x}_j)} - \ip{U_3(\vec{x}_j)}{U_3(\vec{x}_i)}] \geq 0.864 \paren{2 - \ip{\vec{x}_i}{\vec{x}_j} - \ip{\vec{x}_j}{\vec{x}_i}}.
\end{equation}
By linearity of expectation we conclude that the rounded solution $U_3(\vec{x}_1),\ldots,U_3(\vec{x}_N)$ on average has a value in \eqref{eq:ncut3} that is at least $0.864$ times the SDP value. The map $\vec{x}_i \mapsto U_3(\vec{x}_i)$ is the approximate isometry in Theorem \ref{thm:approximate-isometry}. To summarize, the algorithm is

\vspace{10pt}
\IncMargin{1em}
\begin{algorithm}[H]\label{alg:max-cut3}
\DontPrintSemicolon
Solve the SDP \eqref{eq:vectorcut3} to obtain vectors $\vec{x}_1,\ldots,\vec{x}_N$.

Apply the vector-to-unitary construction to obtain unitary operators $X_i = U_{\vec{x}_i}$.

Sample a Haar random unitary $R$.

Return order-$3$ unitaries $\tilde{RX}_1,\ldots,\tilde{RX}_N$.
\caption{Algorithm for approximate isometry for $\ncut{3}$}
\end{algorithm}\DecMargin{1em}
\vspace{10pt}

Inspecting the definition of the SDP \eqref{eq:vectorcut3} we have \[\ip{X_i}{X_j} = \ip{\vec{x}_i}{\vec{x}_j} \geq -1/2,\] for all $i,j$. Therefore to prove the $0.864$ approximation ratio, we just need to prove the following theorem. 
\begin{theorem}\label{thm:ratio-max-cut3}
Let $A$ and $B$ be any two unitaries such that $\lambda \coloneqq \ip{A}{B} \geq -1/2$. Then
\begin{align}
    \expect [2 - \ip{\tilde{X}}{\tilde{Y}} - \ip{\tilde{Y}}{\tilde{X}}] &\geq 0.864 \paren{2 - \ip{A}{B} - \ip{B}{A}}.\label{eq:inequality-that-needs-proof}
\end{align}
where $(X,Y)$ is sampled from $D_{A,B}$. 
\end{theorem}
Recall that $D_{A,B}$ is the \emph{fixed inner product distribution} of the pair of unitaries $A,B$: This is the distribution $(RA,RB)$ where $R$ is sampled from the Haar measure. 

We now sketch the proof of this theorem. Let $(X,Y)$ be a sample from $D_{A,B}$ and note that $\ip{X}{Y} = \ip{A}{B}$. We are done if we prove the inequality 
\begin{align*}
    \expect[2 - \ip{\tilde{X}}{\tilde{Y}} - \ip{\tilde{Y}}{\tilde{X}}] &\geq 0.864 \paren{2 - \ip{A}{B} - \ip{B}{A}}\\
    &= 0.864 \paren{2 - \ip{X}{Y} - \ip{Y}{X}}.
\end{align*}
Let the spectral decompositions of $X$ and $Y$ be $X = \sum_r \alpha_r P_r$ and $Y = \sum_s \beta_s Q_s$. We now can write 
\begin{equation}\label{eq:innerproduct-formula}
2 - \ip{X}{Y} - \ip{Y}{X} = \sum_{r,s} (2 - \alpha_r^*\beta_s - \beta_s^*\alpha_r)\ip{P_r}{Q_s},
\end{equation}
and
\begin{equation}\label{eq:tilde-formula}
2 - \ip{\tilde{X}}{\tilde{Y}} - \ip{\tilde{Y}}{\tilde{X}} = \sum_{r,s} (2 - \tilde{\alpha}_r^*\tilde{\beta}_s - \tilde{\beta}_s^*\tilde{\alpha}_r)\ip{P_r}{Q_s},
\end{equation}
where as before $\tilde{\alpha}_r$ and $\tilde{\beta}_s$ are the closest $3$-rd roots of unity to $\alpha_r$ and $\beta_s$, respectively. The two quantities 
\begin{align}
&2 - \alpha_r^*\beta_s - \beta_s^*\alpha_r,\label{eq:two-quantities-1}\\
&2 - \tilde{\alpha}_r^*\tilde{\beta}_s - \tilde{\beta}_s^*\tilde{\alpha}_r,\label{eq:two-quantities-2}
\end{align}
constitute the only differences in \eqref{eq:innerproduct-formula} and \eqref{eq:tilde-formula}. One way to prove \eqref{eq:inequality-that-needs-proof} is to show that \eqref{eq:two-quantities-1} and \eqref{eq:two-quantities-2} are close to each other. The \emph{fidelity function} helps us quantify this closeness. 

\begin{paragraph}{Fidelity.}
To compare \eqref{eq:two-quantities-1} and \eqref{eq:two-quantities-2} we need to understand the distribution over $(\alpha_r,\beta_s)$ of pairs of eigenvalues of $X$ and $Y$ sampled from $D_{A,B}$. For now consider the following simpler distribution. Fix an angle $\theta$, suppose $\alpha$ is sampled uniformly at random from the unit circle, and let $\beta = \alpha \exp(i\theta)$. In other words $(\alpha,\beta)$ has uniform distribution over all pairs of points on the unit circle that are a phase $\theta$ apart, \emph{i.e.} $\alpha^*\beta = \exp(i\theta)$. We would like to compare the two quantities
\begin{align*}
2 - \alpha^*\beta - \beta^*\alpha,\\
2 - \tilde{\alpha}^*\tilde{\beta} - \tilde{\beta}^*\tilde{\alpha},
\end{align*}
on average. Of course the first quantity is simply
\[2 - \alpha^*\beta - \beta^*\alpha = 2 - 2\cos(\theta).\]
We define the \emph{fidelity at angle} $\theta$ to be the average of the second quantity \[\fid(\theta) \coloneqq \expect[2 - \tilde{\alpha}^*\tilde{\beta} - \tilde{\beta}^*\tilde{\alpha}],\]
where the randomness is that of $(\alpha,\beta)$. We give the formal definition of fidelity in a more general setting in Definition \ref{def:fidelity}.
The fidelity ends up being a very simple piecewise linear function of $\theta$ (this is formally stated and proved in Lemma \ref{lem:general-fidelity}). We draw its plot in Figure \ref{fig:fidelity_3}, where we compare it with the plot of $2 - 2\cos(\theta)$. When $\theta \in \{0, 2\pi/3,4\pi/3\}$ (the angles of $3$-rd roots of unity $\{1,\omega,\omega^2\}$) we see that the two functions are the same. For other angles the fidelity function provides an approximation for $2-2\cos(\theta)$. 
\begin{figure}[h]
\includegraphics[scale=.6]{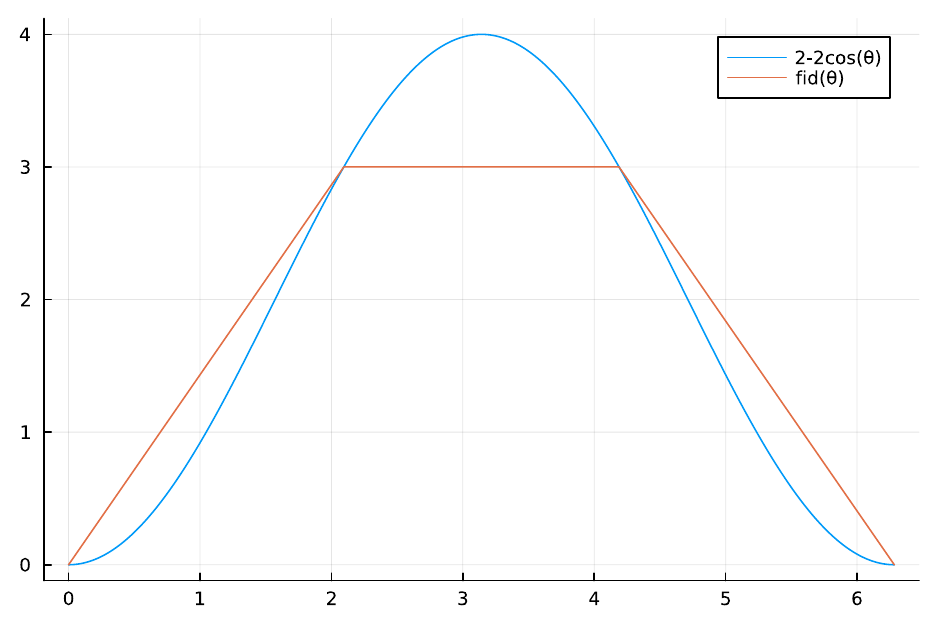}
\centering
\caption{Fidelity as a function of the relative angle $\theta \in [0,2\pi)$ compared with the function $2-2\cos(\theta)$.}\label{fig:fidelity_3}
\end{figure}
\end{paragraph}
 
In the definition of fidelity we assumed $(\alpha,\beta)$ is uniformly distributed. However, the distribution of eigenvalues $(\alpha_r,\beta_s)$ may not be uniform. To understand the distribution of pairs of eigenvalues we need the definition of relative distribution.
\begin{paragraph}{Relative distribution.}
In Section \ref{sec:innovation-relative-distribution} we introduced $\Delta_{A,B}$, the relative distribution of $(A,B)$, as the distribution of the angle $\theta$ in the random process
\begin{enumerate}
    \item Sample $(X,Y)$ from $D_{A,B}$.
    \item Sample a pair of eigenvalues $(\alpha,\beta)$ of $X$ and $Y$ with probability $\ip{P}{Q}$, where $P$ is the projection onto the $\alpha$-eigenspace of $X$ and $Q$ is the projection onto the $\beta$-eigenspace of $Y$.
    \item Let $\theta \in [0,2\pi)$ be the relative angle between $\alpha$ and $\beta$, that is $\theta \coloneqq \measuredangle\alpha^*\beta$.
\end{enumerate}
Fix a $\theta \in [0,2\pi)$ and a sample $(X,Y) \in D_{A,B}$ and let $X = \sum_r \alpha_r P_r$ and $Y = \sum_s \beta_s Q_s$ be the spectral decompositions. The \emph{relative weight of eigenspaces\footnote{This is somewhat reminiscent of the construction of Nussbaum-Szko\l{}a distributions \cite{NussbaumSkola}.} of $X$ and $Y$ on the angle $\theta$} is given by \[w_{X,Y}(\theta) \coloneqq \sum_{\substack{r,s:\\ \alpha_r^*\beta_s = \exp(i\theta)}} \ip{P_r}{Q_s}.\]
If there is no pair of eigenvalues $\alpha_r$ of $X$ and $\beta_s$ of $Y$ such that $\alpha_r^*\beta_s = \exp(i\theta)$, we let $w_{X,Y}(\theta) = 0$. It should be intuitively clear (and this is made formal in Definition \ref{def:relative-measure}) that $\Delta_{A,B}(\theta)$ is the average of $w_{X,Y}(\theta)$ over all samples $(X,Y) \in D_{A,B}$.\footnote{This relationship can be formally understood by treating the distributions as measures, \emph{i.e.} $\Delta_{A,B}(E) = \int_{U} w_{UA,UB}(E)dU$ where $w_{A,B}(E)=\sum_{\theta\in E}w_{A,B}(\theta)$ for every measurable $E \subset [0,2\pi).$}
In Theorem \ref{thm:relative-distribution}, we showed that, in the limit of large dimension, the relative distribution $\Delta_{A,B}$ approaches the wrapped Cauchy distribution $\Delta_\lambda$ where $\lambda = \ip{A}{B}$. This distribution $\Delta_\lambda$ is defined in Preliminaries \ref{sec:probability-theory}. Note that we can always artificially increase the dimension of $A$ and $B$, by tensoring them with the identity operator, without changing the inner product $\ip{A}{B}$ (and hence without changing the objective value of the solution in the CSP). Theorem \ref{thm:relative-distribution} is therefore proving that $\Delta_{A\otimes I_m,B\otimes I_m}\rightarrow \Delta_\lambda$ as $m\rightarrow \infty$. The plot of the wrapped Cauchy distribution $\Delta_\lambda$ for some real values of $\lambda = \ip{A}{B}$ is given in Figure \ref{fig:relative-in-innovation}. 

\end{paragraph}
It turns out that the expectation $\expect[2 - \ip{\tilde{X}}{\tilde{Y}} - \ip{\tilde{Y}}{\tilde{X}}]$ is a simple formula in terms of the fidelity function and relative distribution:
\begin{lemma}\label{lem:fidelity-rel-dist-formula} It holds that
\begin{align*}
\expect[2 - \ip{\tilde{X}}{\tilde{Y}} - \ip{\tilde{Y}}{\tilde{X}}] &= \int_\theta \fid(\theta) d\Delta_{A,B}(\theta).
\end{align*}
\end{lemma}
We prove a stronger version of this in Theorem \ref{lem:noncommutative-fid-delta-integral}. Since we can always increase the dimension of $A$ and $B$, in this integral formula we can replace $\Delta_{A,B}$ with the wrapped Cauchy distribution $\Delta_\lambda$.
Now to prove \eqref{eq:inequality_cut3} we just need to show that
\[\int_\theta \fid(\theta)d\Delta_\lambda(\theta) \geq 0.864 \paren{2 - 2\lambda},\]
for all $\lambda \in [-1/2,1]$. We prove this inequality by elementary means in Section \ref{sec:algorithm} and Appendix~\ref{sec:increasing-ratio}. This completes the sketch of the proof of Theorem \ref{thm:ratio-max-cut3}.

\subsection{Algebraic Approach}\label{sec:case-study-algebraic-approach}
The strength of the analytic approach is its generality: it can be applied to any set of unitary operators $X_1,\ldots,X_N$ to find a nearby set of order-$3$ unitaries. The drawback of the analytic approach is that, as in the example we just saw, the approximation ratio of $0.864$ is obtained only in the limit of large dimension. This is due to the fact that $\Delta_{A \otimes I_m,B\otimes I_m} \rightarrow \Delta_\lambda$ only as $m\rightarrow \infty$. 

We now present an outline of the algebraic approach that resolves this issue. Using this approach, we prove that an approximate solution of the same quality as the one in the analytic approach exists on a Hilbert space of dimension $2^{O(N)}$. The details of the algebraic approach are presented in Sections \ref{sec:algebraic-reldist} and \ref{sec:efficient-algorithm}.

The key idea is to change the definition of Weyl-Brauer operators in the vector-to-unitary construction. The generalized Weyl-Brauer (GWB) of order $3$, introduced in Section \ref{sec:innovation-generalized-anticommutation}, are a set of operators $\sigma_1,\ldots,\sigma_N$ that are unitary matrices of order $3$, i.e. $\sigma_i^*\sigma_i = \sigma_i^3 = 1$ and they pairwise satisfy the relation $\sigma_i^* \sigma_j = - \sigma_j^*\sigma_i$ for all $i\neq j$. In Corollary \ref{cor:fin-dim-vector-to-unitary} we prove that this set of operators exists on a Hilbert space of dimension $2^{O(N)}$. Similarly we can define order-$k$ generalized Weyl-Brauer operators denoted by $\gwb^k$ for any $k$.

Let $\vec{x}\mapsto U_{\vec{x}} \coloneqq \sum_i x_i \sigma_i$ be the vector-to-unitary construction where $\sigma_i$ are the $\gwb^k$ operators. Then, as before, we can show that $U_{\vec{x}}$ is a unitary whenever $\vec{x} \in \R^N$ is a unit vector. Additionally, in Section \ref{sec:gwb}, we show that this generalized vector-to-unitary construction satisfies a \emph{strong isometry} property, \emph{i.e.} $\ip{U_{\vec{x}}^s}{U_{\vec{y}}^s} = \ip{\vec{x}}{\vec{y}}^s$ for all $0\leq s < k$ and all vectors $\vec{x},\vec{y}\in \R^N$. This is a key property used in our analysis.

Now consider the distribution $(U_{O\vec{x}},U_{O\vec{y}})$ where $O$ is a Haar random orthogonal matrix acting on $\R^N$. Note the similarity to the fixed inner product distribution $(R U_{\vec{x}},RU_{\vec{y}})$ where $R$ is Haar random unitary matrix acting on $\C^{2^{O(N)}}$ (this is the distribution used in the analytic approach). Notice also the difference between these two distributions: the order of the randomness step and the vector-to-unitary step is switched. In particular in $(U_{O\vec{x}},U_{O\vec{y}})$ we are using far less randomness. Despite this, using the strong isometry property, in Section \ref{sec:algebraic-reldist} on \emph{algebraic relative distribution}, we prove that the two distributions are the same as far as the relative distribution is concerned. This paves the path for the following algorithm.

In Section \ref{sec:efficient-algorithm} we show that, for sufficiently large $k$, a slightly modified version of Algorithm \ref{alg:efficient-nmax-cut-3} achieves an approximation ratio of $0.864$:

\vspace{10pt}
\IncMargin{1em}
\begin{algorithm}[H]\label{alg:efficient-nmax-cut-3}
\DontPrintSemicolon
Solve the SDP relaxation to obtain vectors $\vec{x}_1,\ldots,\vec{x}_N \in \R^N$.

Sample an orthogonal matrix $O$ acting on $\R^N$ from the Haar measure.  

Apply the vector-to-unitary construction using $\gwb^k$ to obtain unitary operators $X_i = U_{O\vec{x}_i}$.

Round to the nearest order-$3$ unitary and return $\tilde{X}_1,\ldots,\tilde{X}_N$.
\caption{Dimension-efficient algorithm}
\end{algorithm}\DecMargin{1em}
\vspace{10pt}

The dimension of the Hilbert space of this approximate solution is $2^{O(kN)}$ where $k$ is a sufficiently large constant. 
\begin{question}\label{q:dimension-efficient-algorithm}
Could it be that even for $k = 3$, Algorithm \ref{alg:efficient-nmax-cut-3} achieves the approximation ratio $0.864$? This indeed matches numerical evidence.
\end{question}

As mentioned above, the primary difference between Algorithms \ref{alg:max-cut3} and \ref{alg:efficient-nmax-cut-3}, is that the order of the randomness and vector-to-unitary steps is switched. In Algorithm \ref{alg:max-cut3}, we first apply the vector-to-unitary construction (producing exponential-sized matrices), then sample a Haar random unitary acting on a Hilbert space of large dimension. In Algorithm \ref{alg:efficient-nmax-cut-3}, we first sample from the Haar measure on a Hilbert space of small dimension $N$, then apply the vector-to-unitary construction.  

\begin{question}\label{q:reduce-randomness}
Can we reduce the amount of randomness needed in Algorithm \ref{alg:max-cut3}? How about Algorithm \ref{alg:efficient-nmax-cut-3}? Could we switch Haar random unitaries with unitary designs in Algorithm \ref{alg:max-cut3}? Could we derandomize these algorithms altogether?
\end{question}

Both algorithms apply the nonlinear operation $X \mapsto \tilde{X}$ on matrices of exponential dimension in $N$. We need the spectral decomposition of $X$ to calculate $\tilde{X}$. So it is natural to ask the following question about the vector-to-unitary construction.
\begin{question}\label{q:efficient-computation-of-spectrum}
Is there an efficient algorithm that given $\vec{x}$ and $t$ finds the $t$-th eigenvalue (in some ordering) of the operator $U_{\vec{x}}$? 
\end{question}

Given that the solutions are exponential-sized, one cannot hope to even write down the matrices $\tilde{X}_1,\ldots,\tilde{X}_N$. From the nonlocal games perspective, these operators are observables of Alice and Bob in a quantum strategy with $O(N)$ qubits. Could these players be implemented efficiently on quantum computers?
\begin{question}\label{q:efficient-implementation-on-quantum-computers}
Could the order-$3$ observable $\tilde{U_{\vec{x}}}$ be implemented using a polynomial-sized quantum circuit?
\end{question}

\section{Preliminaries}\label{sec:prilim}

\subsection{Notation}
We use shorthand $[n]$ for the set $\{1,\ldots,n\}$. For $z\in\C$, we write $z^\ast$ for its complex conjugate, $\Re(z)$ for its real part, and $\measuredangle z$ for its argument, \emph{i.e.} angle in the interval $[0,2\pi)$ such that $z = |z|e^{i \measuredangle z}$. Write $S^1\subseteq\C$ for the circle of radius $1$. We often represent $S^1$ by the interval $[0,2\pi)$. When doing so, we always assume addition is modulo $2\pi$.

The $k$-th roots of unity are $\set{e^{\frac{2\pi i l}{k}}}{l=0,\ldots,k-1}\subseteq S^1$ and a primitive $k$-th root of unity is a $k$-th root that is not an $l$-th root for any $l<k$; we write $\omega_k=e^{2\pi i/k}$, and drop the subscript when clear from the context. We let $\Omega_k$ denote $k$-simplex that is the convex hull of $k$th roots of unity.

We use vector notation $\vec{x}$ for elements of $\R^d$ or $\C^d$, and write the standard inner product $\ip{\vec{x}}{\vec{y}}=\sum_{i=1}^d x_i^\ast y_i$. We denote the \emph{Euclidean norm} of a vector $\vec{x}$ as $\norm{\vec{x}}=\sqrt{\ip{\vec{x}}{\vec{x}}}$.

We denote the set of $d\times d$ over a field $\mathbb{F}$ matrices by $\mathrm{M}_d(\mathbb{F})$. Write $\Id\in\mathrm{M}_d(\mathbb{F})$ for the identity element. Alternatively we write $I_d$ for the identity element when we want to be explicit about its dimension. More generally we denote the set of $d\times d'$ matrices over a field $\mathbb{F}$ by $\mathrm{M}_{d\times d'}(\mathbb{F})$. Sometimes we want to index entries of a matrix by finite sets that are not $[n]$; we write $\mathrm{M}_{E\times F}(\mathbb{F})$ for matrices in $\mathbb{F}$ with rows and columns indexed by sets $E$ and $F$, respectively. Important subsets of $\mathrm{M}_d(\C)$ are the positive (semidefinite) matrices $\mathrm{M}_d(\C)_{\geq 0}$ and the positive definite matrices $\mathrm{M}_d(\C)_{>0}$; write $X\geq 0$ to mean $X\in\mathrm{M}_d(\C)_{\geq 0}$ when the dimension of $X$ is evident. Write also $\mathrm{GL}_d(\C)$ for the group of invertible matrices and $\mathrm{U}_d(\C)$ for the subgroup of unitary matrices. As for complex numbers, we write $X^\ast$ for the adjoint (hermitian conjugate) of a matrix $X$. We write $\tr:\mathrm{M}_d(\C)\rightarrow\C$ for the normalised trace $\tr(X)=\frac{1}{d}\sum_{i=1}^d X_{ii}$; this gives rise to the Hilbert-Schmidt inner product $\ip{X}{Y}=\tr(X^\ast Y)$. We make use of the \emph{operator norm} $\norm{X}_{\mathrm{op}}=\sup\set[\big]{\norm{Xv}}{\norm{v}\leq 1}$ and the \emph{little Frobenius norm} $\norm{X}=\sqrt{\ip{X}{X}}$.

For a Hilbert space $\mc{H}$, we denote the set of bounded linear operators on $\mc{H}$ as~$\mrm{B}(\mc{H})$. These are exactly the linear operators that are continuous with respect to the norm topology on $\mc{H}$. If $\mc{H}$ is finite-dimensional, then $\mrm{B}(\mc{H})\cong\mrm{M}_d(\C)$ for $d=\dim\mc{H}$; we often make no distinction between these.

\subsection{Probability Theory}\label{sec:probability-theory}
A \emph{measurable space} is a pair $(\Omega,\scr{S})$ where the set $\Omega$ is called the sample space and $\scr{S}$, a $\sigma$-algebra on $\Omega$, is called the event space. For a topological space $\Omega$, there is a canonical choice of $\sigma$-algebra to give it a measurable structure, the \emph{Borel algebra} $\scr{B}(\Omega)$. See~\cite{Axl20} for a reference on measure theory.

A \emph{probability space} is $(\Omega,\scr{S},\mc{D})$ where $\mc{D}$ is a measure on $(\Omega,\scr{S})$ such that $\mc{D}(\Omega)=1$. We call $\mc{D}$ a \emph{distribution} on $\Omega$. 

Given another measurable space $(M,\scr{F})$, a \emph{random variable} on $(\Omega,\scr{S},\mc{D})$ is a measurable function $X:\Omega\rightarrow M$, \emph{i.e.} a function for which $X^{-1}(E) \in \scr{S}$ for all $E \in \scr{F}$. This induces a distribution on $M$ where $\mu_X(E)\coloneqq\mc{D}(X^{-1}(E))$ for all $E\in\scr{F}$; this is called the \emph{distribution} of $X$. We also use probability notation for this distribution, writing $\Pr\squ*{X\in E}\coloneqq\mu_X(E)$. Two random variables $X,Y$ are \emph{independent} if for all measurable $E,F$ we have $\Pr\squ*{X\in E\land Y\in F}=\Pr\squ*{X\in E}\Pr\squ*{Y\in F}$. 

If $M$ is in fact a measure space with its own measure $\nu$ --- for example if $M$ is a finite-dimensional real vector space with the Lebesgue measure $\Lambda$ --- then we can compare this to the measure $\mu_X$. If $\mu_X$ is absolutely continuous with respect to $\nu$ (for all $E\in\scr{F}$, $\nu(E)=0$ implies $\mu_X(E)=0$), then we can define the Radon-Nikodym derivative of $\mu_X$ with respect to $\nu$, which we refer to as the \emph{probability density function (PDF)} $p_X=\frac{d\mu_X}{d\nu}:M\rightarrow\R_{\geq 0}$. Another way to state this relationship is $\mu_X(E) = \int_{E} p_X(x)d\nu(x)$. More generally, we can also consider the PDF of any distribution $\mc{D}$ on $M$ as the Radon-Nikodym derivative $p_{\mc{D}}=\frac{d\mc{D}}{d\nu}$, assuming $\mc{D}$ is absolutely continuous with respect to $\nu$.

If $M$ is a real or complex vector space, the \emph{expectation} of $X$ is given by the integral $\expect X=\int_\Omega X(\omega)d\mu(\omega)=\int_Mxd\mu_X(x)$. Also, for a function $f:M\rightarrow\C$, the composition $f\circ X$ is also a random variable, so we can always write $\expect_X f\circ X = \expect f\circ X =\int_\Omega f(X(\omega))d\mu(\omega)$. We write simply $f(X)$ when referring to $f\circ X$.

\begin{definition}\label{def:converge-in-distribution}
    A sequence of random variables $(X_n)$ \emph{converges in distribution} to a random variable $X$ if for all measurable $E\subseteq M$,
    $$\lim_{n\rightarrow\infty}\Pr\squ*{X_n\in E}=\Pr\squ*{X\in E}.$$
    This is equivalent to $\expect f(X_n)\rightarrow\expect f(X)$ for all bounded continuous functions $f:M\rightarrow\R$.
\end{definition}

It is not necessary to make reference to an underlying sample space to define and work with a random variable, so we generally avoid it. Instead, we define random variables according to a distribution. For a distribution $\mc{D}$ (probability measure) on a measurable space $M$, we say a random variable $X:\Omega\rightarrow M$ is distributed according to $\mc{D}$ if $\mu_X=\mc{D}$. We denote this $X\sim\mc{D}$. Below, we present the standard distributions we use.
\begin{itemize}
    \item \emph{Dirac Delta Distribution}: For any measurable space $(M,\scr{F})$ and $m\in M$, this is the distribution
    \begin{align}
        \delta_{m}(E)=\begin{cases}1&\quad m\in E,\\0&\quad\text{otherwise}.\end{cases}\label{eq:dirac}
    \end{align}

    \item \emph{Complex Normal Distribution}: $\mc{CN}(m,\sigma)$ for mean $m\in\C$ and standard deviation $\sigma>0$ is the distribution on $\C$ with PDF
    \begin{align}
        p_{\mc{CN}(m,\sigma)}(z)=\frac{1}{2\pi\sigma^2}e^{-\frac{|z-m|^2}{2\sigma^2}},
    \end{align}
    with respect to the Lebesgue measure on $\C$. We call $\mc{CN}=\mc{CN}(0,1)$ the \emph{standard complex normal distribution}.
    \item \emph{Vector Complex Normal Distribution}: $\mc{CN}(d,\vec{m},\Sigma)$ for dimension $d\in\N$, mean $\vec{m}\in\C^d$, and covariance matrix $\Sigma\in\mathrm{M}_d(\C)_{>0}$ is the distribution on $\C^d$ with PDF
    \begin{align}
        p_{\mc{CN}(d,\vec{m},\Sigma)}(\vec{z})=\frac{1}{(2\pi\det\Sigma)^d}e^{-\frac{1}{2}\ip{\vec{z}-\vec{m}}{\Sigma^{-1}(\vec{z}-\vec{m})}},
    \end{align}
    with respect to the Lebesgue measure on $\C^d$. We call $\mc{CN}(d)=\mc{CN}(d,\vec{0},\Id)$ the \emph{standard vector complex normal distribution}. Note that if $\Sigma$ is diagonal, then the components of a random variable $\vec{Z}\sim\mc{CN}(d,\vec{m},\Sigma)$ are independent.
    \item \emph{Wrapped Cauchy Distribution}: $\mc{W}(\theta_0,\gamma)$ for peak position $\theta_0\in[0,2\pi)$ and scale factor $\gamma>0$ is the distribution on $[0,2\pi)\cong S^1$ with PDF
    \begin{align}
        p_{\mc{W}(\theta_0,\gamma)}(\theta)=\frac{\sinh(\gamma)}{2\pi(\cosh(\gamma)-\cos(\theta-\theta_0))}, \label{eq:cauchy}
    \end{align}
    with respect to the Lebesgue measure on $[0,2\pi)$. For $\lambda\in\C$ with $|\lambda|\leq1$ we define the distribution $\Delta_\lambda \coloneqq \mc{W}(\measuredangle \lambda,-\ln|\lambda|)$ when $|\lambda|<1$ and $\Delta_\lambda\coloneqq\delta_{\measuredangle\lambda}$ when $|\lambda|=1$.
    
    \item \emph{Haar Distribution}: For any compact topological group $G$, there exists a unique distribution $\mathrm{Haar}(G)$ on $(G,\scr{B}(G))$ such that for all $E\subseteq G$ Borel-measurable and $g\in G$, $$\mathrm{Haar}(G)(E)=\mathrm{Haar}(G)(gE)=\mathrm{Haar}(G)(Eg).$$ We make use of $\mathrm{Haar}(\mathrm{U}_d(\C))$.
\end{itemize}

We will be largely interested in random variables $X$ on the circle $[0,2\pi)\cong\R/2\pi\Z\cong S^1$. Here, we take the \emph{characteristic function} of $X$ to be the Fourier series $$\chi_X:\Z\rightarrow\C,\quad\chi_X(n)=\expect(e^{i n X}).$$ If $\mu_X$ is absolutely continuous, then we have $\chi_X=\hat{p}_X$, the Fourier series of the PDF. The \emph{Fourier series} of an integrable function $f:[0,2\pi)\rightarrow\C$ is the function $\hat{f}:\Z\rightarrow\C$ defined as $$\hat{f}(n)=\int_0^{2\pi}f(\theta)e^{in\theta}d\theta.$$ 

The characteristic function can similarly be defined for distributions. For example the characteristic function of the Dirac delta distribution $\delta_{\theta_0}$ on $[0,2\pi)$ is
$$\chi_{\delta_{\theta_0}}(n)=e^{in\theta_0};$$
and the characteristic function of the wrapped Cauchy distribution $\mc{W}(\theta_0,\gamma)$ is
$$\chi_{\mc{W}(\theta_0,\gamma)}(n)=e^{-|n|\gamma+in\theta_0}.$$
This implies that $\chi_{\Delta_\lambda}(n)=\lambda^n$ for $n\geq 0$ and $\chi_{\Delta_\lambda}(n)=(\lambda^\ast)^{-n}$ for $n<0$.

\begin{theorem}[Parseval]
    Let $f,g:[0,2\pi)\rightarrow\C$ be square-integrable functions on the circle ($f,g\in L^2[0,2\pi)$). Then
    \begin{equation}\label{eq:parseval}
    \int_0^{2\pi}f(\theta)^\ast g(\theta)d\theta = \frac{1}{2\pi}\sum_{n\in\Z}\hat{f}(n)^\ast\hat{g}(n).
    \end{equation}
    In particular, the series $\frac{1}{2\pi}\sum_{n\in\Z}\hat{f}(n)^\ast\hat{g}(n)$ converges.
\end{theorem}

This theorem states that the operation of taking the Fourier series is an isometry $L^2[0,2\pi)\rightarrow\ell^2\Z$ with the correct normalisation on the inner products. Using Parseval, when $X$ is a random variable on $[0,2\pi)$ with square-integrable PDF and $f:[0,2\pi)\rightarrow\C$ is a square-integrable function, we have
\begin{align}\label{eq:parseval-for-distributions}
    \int_0^{2\pi} f(\theta)d\mu_X(\theta)=\frac{1}{2\pi}\sum_{n\in\Z}\hat{f}(-n)\chi_{X}(n).
\end{align}

\begin{theorem}[L\'evy's continuity theorem on the circle]\label{thm:levy}
    A sequence of random variables $(X_n)$ on the $[0,2\pi)$ converges in distribution if and only if the sequence of characteristic functions $(\chi_{X_n})$ converges pointwise.
\end{theorem}
The above theorem is a consequence of the Stone-Weierstrass theorem.

\subsection{Free Probability}\label{sec:free-probability}
We will also make use of results from free probability theory. This is the study of noncommutative random variables and our main reference for this section is~\cite{NS06}.

A \emph{noncommutative probability space} is a unital algebra $\mc{A}$ equipped with a linear functional $\tau:\mc{A}\rightarrow\C$ such that $\tau(1)=1$, called the \emph{expectation}. We will always assume, as is commonly done, that $\mc{A}$ is a $C^\ast$-algebra, \emph{i.e.} $\mc{A}$ has a conjugate-linear involution $\ast$ and is complete with respect to a submultiplicative norm satisfying $\norm{a^\ast a}=\norm{a}^2$; and that $\tau$ is a faithful tracial state, \emph{i.e.} that $\tau(a^\ast a)\geq 0$, if $\tau(a^\ast a)=0$ then $a=0$, and $\tau(ab)=\tau(ba)$. Elements of a noncommutative probability space are called \emph{noncommutative random variables}. As a first example, it is worth noting that commutative probability theory, as seen in the previous section, fits as a special case of noncommutative probability theory. In fact, given a probability space $(\Omega,\scr{S},\mc{D})$, its set of (bounded, complex-valued) random variables corresponds to the noncommutative probability space $\mc{A}=L^\infty(\Omega,\mc{D})$ with $\tau(X)=\int X(\omega)d\mc{D}(\omega)=\expect X$. Another example of noncommutative probability space is $(\mathrm{M}_d(\C),\tr)$.

A \emph{mixed $\ast$-moment} of a collection of  noncommutative random variables $a_1,\ldots,a_n \in\mc{A}$ is an expectation of the form $\tau(a_{i_1}^{e_1}\cdots a_{i_m}^{e_m})$, where $i_j\in[n]$ and $e_j\in\{1,\ast\}$. For a single element $a\in\mc{A}$, the \emph{$\ast$-moments} of $a$ are the mixed $\ast$-moments of $a$. For normal $a\in\mc{A}$, \emph{i.e.} $aa^\ast=a^\ast a$, the $\ast$-moments are simply $\tau((a^\ast)^ka^{l})$ for integers $k,l\geq 0$. In this case, the idea of a \emph{probability distribution} for random variables generalizes quite directly. There exists a unique measure $\mu_a$ on (some compact subset of) $\C$ such that all the $\ast$-moments can be written as
$$\tau((a^\ast)^ka^l)=\int (z^\ast)^k z^ld\mu_a(z).$$
Knowing the $\ast$-moments, the expectation of any element of $C^\ast\!\!\group{a}$, the unital subalgebra generated by $a$, can be computed. 

The $\ast$-algebra of univariate $\ast$-polynomials, denoted $\C^\ast\!\!\group{x}$, is the $\ast$-algebra generated by the free variable $x$. This extends to multivariate $\ast$-polynomials $\C^\ast\!\!\group{x_1,\ldots,x_n}$.

\begin{definition}
    A set of noncommutative random variables $a_1,\ldots,a_n\in\mc{A}$ is \emph{freely independent} (or \emph{free} for short) if, for any $i_1,\ldots,i_m\in[n]$ such that $i_j\neq i_{j+1}$ and $\ast$-polynomials $P_1,\ldots,P_m\in\C^\ast\!\!\group{x}$ such that $\tau(P_j(a_{i_j}))=0$, we have that
    $$\tau(P_1(a_{i_1})\cdots P_m(a_{i_m}))=0.$$
\end{definition}
The usual commutative notion of independence guarantees the above property only in the case where all the $i_j$ are distinct, which can be supposed due to commutativity. If the $\ast$-moments of a set of free random variables are known, then any of the mixed $\ast$-moments can be computed. 

\begin{definition}
    A collection of sequences $(a_{1,n}),\ldots,(a_{l,n})$ of noncommutative random variables such that $a_{i,n}\in(\mc{A}_n,\tau_n)$ \emph{converges in $\ast$-distribution} to a collection of noncommutative random variables $a_1,\ldots,a_l\in(\mc{A},\tau)$ if all the mixed $\ast$-moments converge, \emph{i.e.} $$\lim_{n\rightarrow\infty}\tau_n(a_{i_1,n}^{e_1}\cdots a_{i_m,n}^{e_m})=\tau(a_{i_1}^{e_1}\cdots a_{i_m}^{e_m}).$$ We write this as $a_{1,n},\ldots,a_{l,n}\rightarrow a_1,\ldots,a_l$ ($\ast$-dist).
\end{definition}

The idea of a Haar-random unitary generalizes naturally to the context of free probability. We say a noncommutative random variable $u\in\mc{A}$ is a \emph{Haar unitary} if $u$ is unitary ($u^\ast u=uu^\ast=1$) and the $\ast$-moments $\tau(u^k)=\delta_{k,0}$ for all $k\in\Z$. This generalizes the Haar distribution on the unitary group as a Haar-random unitary $U$ on $\mathrm{U}_d(\C)$ satisfies $\int\tr(U^k)dU=\delta_{k,0}$. A Haar-random unitary can be seen as a noncommutative random variable on the space $\mc{A}=\mathrm{M}_d(L^\infty(\Omega,\mc{D}))$ for some probability space $(\Omega,\scr{S},\mc{D})$. This space has natural expectation $\tau(a)=\int\tr(a(\omega))d\mc{D}(\omega)$.

We make use of the following theorem.
\begin{theorem}[23.13 in \cite{NS06}]\label{thm:asymptotically-free}
    Let $(d_n)$ be a sequence of natural numbers such that $d_n\rightarrow\infty$. Let $(U_n)$ be a sequence of random variables such that $U_n\sim\mathrm{Haar}(\mathrm{U}_{d_n}(\C))$, and let $(A_n)$ be a sequence of matrices such that $A_{n}\in\mathrm{M}_{d_n}(\C)$. Suppose that the limiting $\ast$-moments $\lim_{n\rightarrow\infty}\tr(A_n^{e_1}\cdots A_n^{e_m})$ for $e_i\in\{1,\ast\}$ all exist. Then, there exists a noncommutative probability space $\mc{A}$ and freely independent random variables $u,a\in\mc{A}$ such that $U_n,A_n\rightarrow u,a$ ($\ast$-dist).
\end{theorem}

Note that, by the convergence in $\ast$-distribution, we have in particular that $u$ is Haar unitary and the $\ast$-moments of $a$ are the limiting $\ast$-moments of $A_n$. In this context, we say that $(U_n)$ and $(A_n)$ are \emph{asymptotically free}.

\subsection{Group Theory}\label{sec:group-presentations}

For any group $G$ and subset $A\subseteq G$, we write $\group{A}$ for the subgroup generated by $A$, and $\llangle A \rrangle$ for its normal closure. The identity element of a group is denoted $e$. For two group elements $g,h$ we let $[g,h]$ denote the commutator $ghg^{-1}h^{-1}$.

\vspace{0.2cm}

\noindent\textbf{Presentations. } We present groups we work with by means of generators and relations. A set of \emph{generators} is a set of symbols $S$, and a set of \emph{relators} is a set $R$ of words in $S\cup S^{-1}$, where $S^{-1}=\set*{x^{-1}}{x\in S}$. Each relator $r\in R$ corresponds to the \emph{relation} $r=e$. The free group generated by $S$, $\langle S\rangle$ is the group of words in $S\cup S^{-1}$ under concatenation --- the identity element is the empty word $e$ and inverses exist since we include $S^{-1}$ in the generating set. Then, the \emph{group generated by generators $S$ and relators $R$} is the quotient group $\group{S}{R}:=\langle S\rangle/\langle\!\langle R\rangle\!\rangle$. As noted, we can equivalently phrase the relators as relations on the words in $S$. A group $G$ is \emph{finitely-presented} if $G=\group{S}{R}$ where $S$ and $R$ are both finite sets.

\vspace{0.2cm}

\noindent\textbf{Free and amalgamated products. } Given two groups $G$ and $H$, their \emph{free product} is the group $G\ast H$ generated by $G\cup H$ with no relations between the elements of $G$ and $H$ except that the identities $e_G=e_H$. The elements of $G\ast H$ are words in $G\cup H$. Given subgroups $A\leq G$ and $B\leq H$ that are isomorphic, the \emph{amalgamated product} $G\ast_A H$ is the quotient $(G\ast H)/\langle\!\langle R\rangle\!\rangle$ where $R$ is the set of relators $a^{-1}\phi(a)$ for $a\in A$, where $\phi:A\rightarrow B$ is the isomorphism. 

\vspace{0.2cm}

\noindent\textbf{Representations. } A ($\C$-linear) \emph{representation} of a group $G$ is a homomorphism $\varrho:G\rightarrow\mathrm{GL}_d(\C)$. Suppose $G$ is a finite group. Then, $\varrho$ is equivalent to a unitary representation $\varrho(G)\in\mathrm{U}_d(\C)$ by conjugation, so we may assume any representation is unitary. Also, $\varrho$ decomposes as a direct sum of irreducibles (representations with no nontrivial invariant subspaces). Every finite group has finitely many classes of irreducible representations, modulo equivalence by conjugation. For conciseness, we refer to these classes as \emph{irreps}. The number of these irreps is equal to the number of conjugacy classes, subsets $\mathrm{Cl}(g)=\set*{hgh^{-1}}{G}$ that partition $G$. See \cite{Ser77} for a full treatment.

\vspace{0.2cm}

\noindent\textbf{Group algebras. } For any group $G$, the \emph{group algebra} $\C[G]$ is the set of finitely-supported functions $G\rightarrow\C$. We write an element $x\in\C[G]$ as $x=\sum_{g\in G}x(g)g$. $\C[G]$ is an algebra as it has a natural structure as a $\C$-vector space, and it also has a product given by convolution: $xy=\sum_{g,h\in G}x(g)y(h)gh=\sum_{g\in G}\sum_{h\in G}x(gh^{-1})y(h)g$. The multiplicative identity in $\C[G]$ is the group identity $e$. Further, $\C[G]$ is a $\ast$-algebra with involution $x^\ast=\sum_gx(g)^\ast g^{-1}$. The $\ast$-homomorphisms $\C[G]\rightarrow\mathrm{B}(\mc{H})$ correspond exactly to the unitary representations of $G$.

\subsection{Von Neumann Algebras}
A \emph{von Neumann algebra} is a unital $\ast$-subalgebra $\mc{M}\subseteq \mrm{B}(\mc{H})$ for some Hilbert space $\mc{H}$ that is equal to its double commutant $\mc{M}=\mc{M}''$, where $S':=\set*{x\in\mrm{B}(\mc{H})}{xy=yx\;\forall\,y\in S}$ for any $S\subseteq\mrm{B}(\mc{H})$. The double commutant condition is equivalent to having $\mc{M}$ be closed in the strong, weak, or weak-$\ast$ operator topologies on $\mathrm{B}(\mc{H})$. In particular, the weak-$\ast$ topology is the smallest topology making any bounded linear functional continuous. A sequence $(x_n)$ converges to $x$ in the weak-$\ast$ topology if, for any bounded linear functional $\phi:\mrm{B}(\mc{H})\rightarrow\C$, the sequence $\phi(x_n)\rightarrow\phi(x)$ in $\C$.

We focus mainly on \emph{finite} von Neumann algebras: $\mc{M}$ is finite iff it has a weak-$\ast$ continuous tracial state $\tau:\mc{M}\rightarrow\C$, that is a faithful positive linear functional such that $\tau(xy)=\tau(yx)$. A functional $\tau$ is weak-$\ast$ continuous if for any weak-$\ast$ convergent sequence $(x_n)\rightarrow x$, then $\tau(x_n)\rightarrow\tau(x)$. An important fact about finite von Neumann algebras is that any compression remains finite. That is, for any nonzero projection $p\in\mc{M}$, \emph{i.e.} $p=p^\ast=p^2$, $p\mc{M}p\subseteq\mrm{B}(p\mc{H})$ is a von Neumann algebra with trace~$\frac{1}{\tau(p)}\tau|_{p\mc{M}p}$, so it is finite. Note that if $p$ is also central, then $p\mc{M}p = p\mc{M}$.  

The most important examples of von Neumann algebras we use are von Neumann group algebras. Let $G$ be an (infinite discrete) group, and consider the evaluation map $\tau:\C[G]\rightarrow\C$, $x\mapsto x(e)$. $\tau$ is a unital linear functional, and it is positive and faithful as
$$\tau(x^\ast x)=\sum_{g,h}x(g)^\ast x(h)\tau(g^{-1}h)=\sum_g|x(g)|^2,$$
which is positive, and $0$ if and only if all the terms are $0$. Also, $\tau$ is tracial as
$$\tau(xy)=\sum_{g}x(g)y(g^{-1})=\sum_gy(g)x(g^{-1})=\tau(yx).$$
We call this the canonical trace on $\C[G]$. Further, $\ip{x}{y}:=\tau(x^\ast y)$ gives an inner product on $\C[G]$. Consider the completion of $\C[G]$ with respect to the topology generated by the inner product, denoted $\ell^2 G$. This is a Hilbert space. $\C[G]$ acts by left multiplication on $\ell^2 G$, therefore $\C[G]$ embeds into $\mrm{B}(\ell^2G)$. We define the \emph{von Neumann group algebra} $\mrm{L}(G)$ by the double commutant of the image of $\C[G]$ in $\mrm{B}(\ell^2G)$. $\mrm{L}(G)$ is a finite von Neumann algebra since the canonical trace $\tau$ of $\C[G]$ extends to a weak-$\ast$ continuous trace on $\mrm{L}(G)$. For finite $G$ we have $\ell^2 G = \C[G] = \mrm{L}(G)$ as every finite vector space is complete.

For a complete introduction to von Neumann algebras, see~\cite{Bla06}

\subsection{Semidefinite Programming}
A \emph{semidefinite program} (SDP) is any optimization problem of the form:
\begin{equation}\label{eq:SDPdef}
\openup\jot
\begin{aligned}[t]
\quad\text{ maximize:}\quad &\ip{C}{X} \\
        \text{subject to:}\quad & \mc{A}_i(X) = B_i \text{ for } i \in \{1,\ldots,m\}\\
        & X \in \mathrm{M}_d(\C)_{\geq 0},
\end{aligned}
\end{equation}
where $C \in \mathrm{M}_d(\C), B_i \in \mathrm{M}_{d_i}(\C)$ are Hermitian and $\mc{A}_i: \mathrm{M}_d(\C) \to \mathrm{M}_{d_i}(\C)$ are linear maps. Since any feasible solution $X$ is positive semidefinite, we can use its Gram matrix characterisation to equivalently write the SDP in a vector form. Gram vectors are $\vec{v}_i\in \mathbb{C}^d$ such that $X_{ij}=\ip{\vec{v}_i}{\vec{v}_j}$, which always exist by spectral decomposition. SDPs can be approximated to any precision $\epsilon$ in time $\poly(d,\log(1/\epsilon))$ \cite{VandenbergheBoydSDP}. For more on SDPs, see \cite{boyd2004convex,watrousAdvQIT}.

\section{Main Definitions and Results} \label{sec:csp-definitions}

In this section, we introduce the classes of CSPs that we study, and state our results in terms of the approximation ratios we obtain for them. At first we define the class of \emph{linear $2$-CSPs} known as $\linlegacy{2}{k}$. We then discuss two natural \emph{noncommutative generalizations} of them. In the following two subsections we introduce the subclasses of \emph{homogeneous CSPs} and \emph{smooth CSPs}, where we also present our main results about them.

$\linlegacy{2}{k}$ \cite{hastad,goemansmax3cut,khot}, or $\lin{2}{k}$ for short (since we only consider $2$-CSPs in this work), is a class of CSPs that includes the Unique Games\footnote{To be more precise, the class $\lin{2}{k}$ only includes those Unique-Games that are linear. However by \cite{khot,mossel}, the inapproximability of linear Unique-Games is equivalent to that of general Unique-Games.} \cite{khot_original} as well as optimization problems such as $\cut{k}$ \cite{papadimitriou,goemansmaxcut,frieze,goemansmax3cut,klerk}. Given $N$ variables $x_1,\ldots,x_N$ taking values in $\Z_k = \{0,1,\ldots,k-1\}$ and a system of linear equations $x_j - x_i = c_{ij}$ and inequations $x_j - x_i \neq c_{ij}$ for constants $c_{ij} \in \Z_k$, each with an associated weight $w_{ij} \geq 0$, the goal of $\lin{2}{k}$ is to maximize the total weight of satisfied constraints.\footnote{We may choose the constraints to be symmetric in the sense that $c_{ij}=-c_{ji}$ and $w_{ij}=w_{ji}$, because these constraints involve the same variables. Hence, we always assume that these symmetry conditions hold.}

$\ugame{k}$ are those instances of $\lin{2}{k}$ with only equation constraints.\footnote{The term unique refers to the fact that in equation constraints involving only two variables, any assignment to one variable determines uniquely the choice of satisfying assignment to the other variable.} On the other end, instances with only inequation constraints and $c_{ij} = 0$ are the $\cut{k}$ problems. 

We can equivalently state a $\lin{2}{k}$ instance in multiplicative form where the variables $x_1,\ldots,x_N$ take values that are $k$-th roots of unity, and the constraints are $x_i^{-1}x_j = \omega^{c_{ij}}$ or $x_i^{-1}x_j \neq \omega^{c_{ij}}$ for constants $c_{ij} \in \Z_k$. Here $\omega$ is a primitive $k$-th root of unity. This multiplicative framing is also sometimes preferred in the literature, see for example~\cite{goemansmax3cut}. When $x,y$ are $k$-th roots of unity, the polynomial $\frac{1}{k}\sum_{s=0}^{k-1} \omega^{-cs}x^{-s}y^s$ indicates whether the equation $x^{-1}y = \omega^c$ is satisfied, since
\begin{align*}
\frac{1}{k}\sum_{s=0}^{k-1} \omega^{-cs}x^{-s}y^s = \begin{cases}0 & x^{-1}y \neq \omega^c,\\ 1 & x^{-1}y = \omega^c.\end{cases}
\end{align*}
So, we can frame $\lin{2}{k}$ as a \emph{polynomial optimization} problem
\begin{equation}\label{eq:lin2k}
\openup\jot
\begin{aligned}[t]
\text{ maximize:}\quad &\sum_{(i,j)\in \mathcal{E}} \frac{w_{ij}}{k} {\sum_{s=0}^{k-1} \omega^{-c_{ij}s} x_i^{-s} x_j^{s}} + \sum_{(i,j)\in \mathcal{I}} w_{ij} \bigparen{1 - \frac{1}{k}{\sum_{s=0}^{k-1} \omega^{-c_{ij}s} x_i^{-s} x_j^{s}}} \\
        \text{subject to:}\quad & x_i^k = 1 \text{ and } x_i \in \C,
\end{aligned}
\end{equation}
where $\mathcal{E}$ is the set of equation constraints and $\mathcal{I}$ is the set of inequation constraints.

$\lin{2}{k}$ is famously $\NP$-hard~\cite{karp}. However, efficient approximation algorithms for these problems have long been known. To understand the quality of these approximations, we need the following definition.
\begin{definition}\label{def:approximation-ratio}
The \emph{approximation ratio} of an algorithm $\mathcal{A}$ for a CSP such as $\lin{2}{k}$ is the quantity $\inf_I \frac{\mathcal{A}(I)}{\opt(I)}$
where $I$ ranges over all possible instances of the problem and $\opt(I)$ is the optimal value of the instance $I$. 
\end{definition}
Most of these approximation algorithms for $\lin{2}{k}$ are based on semidefinite programming (SDP) relaxations. We will see some of these relaxations later in this section. The following definition captures the quality of an SDP relaxation:

\begin{definition}\label{def:integrality-gap} The \emph{integrality gap} of an SDP relaxation for a CSP is $\inf_I\frac{\opt(I)}{\sdp(I)}$, where the infimum ranges over all instances $I$ of the CSP and $\sdp(I)$ is the optimal value of the SDP on the instance $I$.
\end{definition}

We now discuss the noncommutative variants of $\lin{2}{k}$. One such variant is obtained from the classical problem by relaxing the commutativity constraint in the polynomial \eqref{eq:lin2k}. This leads us to the following \emph{noncommutative polynomial optimization} problem
\begin{equation}\label{eq:nlin2k_norm}
\openup\jot
\begin{aligned}[t]
\text{ maximize:}\quad & \norm[\Big]{\sum_{(i,j)\in \mathcal{E}} \frac{w_{ij}}{k} {\sum_{s=0}^{k-1} \omega^{-c_{ij}s} X_i^{-s} X_j^{s}} + \sum_{(i,j)\in \mathcal{I}} w_{ij} \bigparen{1 - \frac{1}{k}{\sum_{s=0}^{k-1} \omega^{-c_{ij}s} X_i^{-s} X_j^{s}}}}_{\mathrm{op}} \\
        \text{subject to:}
                          \quad & X_i^* X_i = X_i^k = 1,
\end{aligned}
\end{equation}
where the maximization ranges over all finite-dimensional Hilbert spaces and unitary operators of order $k$ acting on them. The value of \eqref{eq:nlin2k_norm} is an upper bound on the classical value since \eqref{eq:lin2k} is the restriction of \eqref{eq:nlin2k_norm} to the one-dimensional Hilbert space. 

It turns out that this noncommutative problem is equivalent to an SDP, so it can be solved efficiently. For example in the special case of $\cut{k}$, this SDP is the canonical SDP relaxation of the classical problem. This indicates that there is a close relationship between noncommutative CSPs and classical CSPs, that we will elaborate on in Section \ref{sec:brave-new-world}.

There are many other variations of noncommutative CSPs. For the rest of the paper, we will focus on the following tracial variant of the problem, since it is well-motivated in quantum information:
\begin{equation}\label{eq:nlin2k}
\openup\jot
\begin{aligned}[t]
\text{ maximize:}\quad & \tr\Bigparen{\sum_{(i,j)\in \mathcal{E}} \frac{w_{ij}}{k} {\sum_{s=0}^{k-1} \omega^{-c_{ij}s} X_i^{-s} X_j^{s}} + \sum_{(i,j)\in \mathcal{I}} w_{ij} \bigparen{1 - \frac{1}{k}{\sum_{s=0}^{k-1} \omega^{-c_{ij}s} X_i^{-s} X_j^{s}}}} \\
        \text{subject to:}
                          \quad & X_i^*X_i = X_i^k = 1.\\
\end{aligned}
\end{equation}
From here on this is what we refer to as the noncommutative $\lin{2}{k}$, or $\nlin{2}{k}$ for short. When there is an ambiguity we refer to \eqref{eq:nlin2k_norm} as norm $\lin{2}{k}$ and \eqref{eq:nlin2k} as tracial $\lin{2}{k}$.

\begin{remark*}
There are other noncommutative generalizations of classical CSPs, such as the problem known as Quantum $\maxcut$ (which is closely related to local Hamiltonian problems). See Section~\ref{sec:quantum-max-cut} for more discussion.   
\end{remark*}

An equivalent and fruitful way of looking at CSPs is their formulation as nonlocal games, better known as multiprover interactive proofs in theoretical computer science. For example, although we did not present Unique Games as games here, they were historically presented as a $1$-round-$2$-player nonlocal game \cite{khot_original}. It turns out that the commutative polynomial optimization~\eqref{eq:lin2k} captures what is known as the \emph{synchronous classical value} of the nonlocal game arising from the CSP instance. Analogously the noncommutative polynomial optimization \eqref{eq:nlin2k} captures the \emph{synchronous quantum value} of the same game. For a brief introduction to nonlocal games, see Section \ref{sec:nonlocal_games}, where we also explain the connection between $2$-CSPs and $2$-player nonlocal games.

\subsection{Homogeneous CSPs and Max-$k$-Cut}\label{sec:def-homogeneous-csps}
If in the definition of $\lin{2}{k}$ we restrict $c_{ij}$'s to be all zero we arrive at the subclass of homogeneous $\lin{2}{k}$, or $\homlin{2}{k}$ for short, which is expressed as
\begin{equation}\label{eq:def-homlin2k}
\openup\jot
\begin{aligned}[t]
\text{ maximize:}\quad &\sum_{(i,j)\in \mathcal{E}} \frac{w_{ij}}{k} {\sum_{s=0}^{k-1} x_i^{-s} x_j^{s}} + \sum_{(i,j)\in \mathcal{I}} w_{ij} \parens[\Big]{1 - \frac{1}{k}{\sum_{s=0}^{k-1} x_i^{-s} x_j^{s}}} \\
        \text{subject to:}\quad & x_i^k = 1 \text{ and } x_i \in \C.
\end{aligned}
\end{equation}
The noncommutative $\homlin{2}{k}$, denoted by $\nhomlin{2}{k}$, takes the form
\begin{equation}\label{eq:def-nhomlin2k}
\openup\jot
\begin{aligned}[t]
\text{ maximize:}\quad &\sum_{(i,j)\in \mathcal{E}} \frac{w_{ij}}{k} {\sum_{s=0}^{k-1} \ip{X_i^{s}}{X_j^{s}}} + \sum_{(i,j)\in \mathcal{I}} w_{ij} \parens[\Big]{1 - \frac{1}{k}{\sum_{s=0}^{k-1} \ip{X_i^{s}}{X_j^{s}}}} \\
        \text{subject to:}\quad & X_i^*X_i = X_i^k = 1.
\end{aligned}
\end{equation}
The canonical SDP relaxation of both the classical and noncommutative problems is given by 
\begin{equation}\label{eq:def-homlin2k-sdp}
\openup\jot
\begin{aligned}[t]
\text{ maximize:}\quad &\sum_{(i,j)\in \mathcal{E}} \frac{w_{ij}}{k} \paren{1 + (k-1)X_{ij}} + \sum_{(i,j)\in \mathcal{I}} w_{ij}\frac{k-1}{k}(1 - X_{ij})\\
        \text{subject to:}\quad & X \in \mathrm{M}_N(\R)_{\geq 0},\\
                          \quad & X_{ii} = 1,\\
                          \quad & X_{ij} \geq -\frac{1}{k-1},     
\end{aligned}
\end{equation}
which can be equivalently written in the vector form as
\begin{equation}\label{eq:def-homlin2k-sdp-vector}
\openup\jot
\begin{aligned}[t]
\text{ maximize:}\quad &\sum_{(i,j)\in \mathcal{E}} \frac{w_{ij}}{k} \paren{1 + (k-1)\ip{\vec{x}_i}{\vec{x}_j}} + \sum_{(i,j)\in \mathcal{I}} w_{ij}\frac{k-1}{k}(1 - \ip{\vec{x}_i}{\vec{x}_j})\\
        \text{subject to:}\quad & \vec{x}_i \in \mathbb{R}^{N},\\
                          \quad & \|\vec{x}_i\| = 1,\\
                          \quad & \ip{\vec{x}_i}{\vec{x}_j} \geq -\frac{1}{k-1}.  
\end{aligned}
\end{equation}
\begin{proposition} \eqref{eq:def-homlin2k-sdp} is a relaxation of \eqref{eq:def-nhomlin2k}.
\end{proposition}
\begin{proof}
Suppose $X_1,\ldots,X_n$ is any feasible solution to $\nhomlin{2}{k}$. We construct a feasible solution $X$ in the SDP achieving the same objective value as $X_1,\ldots,X_n$. For $0\leq s<k$ let $W_s$ be the matrix defined so that the $(i,j)$-th entry is $\ip{X_i^s}{X_j^s}$. Then we clearly have $W_s \geq 0$ with diagonal entries that are all $1$. Now let $X = \frac{1}{k-1} \sum_{s=1}^{k-1} W_s$. It should be clear that $X$ is real, symmetric, and positive semidefinite $X \geq 0$. Also clearly $X_{ii} = 1$ for all $i$. Also since $\sum_{s=0}^{k-1} \ip{X_i^s}{X_j^s} \geq 0$ we have $X_{ij} = \frac{1}{k-1} \sum_{s=1}^{k-1} \ip{X_i^s}{X_j^s} \geq -\frac{1}{k-1}.$
Thus $X$ is a feasible solution in the SDP and it achieves the same objective value as that of $X_1,\ldots,X_n$ in $\nhomlin{2}{k}$. 
\end{proof}

$\cut{k}$ is the special case of $\homlin{2}{k}$ with only inequation constraints
\begin{equation}\label{eq:def-maxcutk}
\openup\jot
\begin{aligned}[t]
\text{ maximize:}\quad &\sum_{(i,j)\in E} w_{ij} \parens[\Big]{1 - \frac{1}{k}{\sum_{s=0}^{k-1} x_i^{-s} x_j^{s}}} \\
        \text{subject to:}\quad & x_i^k = 1 \text{ and } x_i \in \C \text{ for all } i \in V,
\end{aligned}
\end{equation}
where $V = \{1,\ldots,N\}$ and $E$ are the vertex and edge sets of the underlying graph $G=(V,E)$ of the instance. The noncommutative variant, denoted $\ncut{k}$, is
\begin{equation}\label{eq:def-nmaxcutk}
\openup\jot
\begin{aligned}[t]
\text{ maximize:}\quad &\sum_{(i,j)\in E} w_{ij} \parens[\Big]{1 - \frac{1}{k}{\sum_{s=0}^{k-1} \ip{X_i^{s}}{X_j^{s}}}} \\
        \text{subject to:}\quad & X_i^*X_i = X_i^k = 1 \text{ for all } i \in V.
\end{aligned}
\end{equation}
The canonical SDP relaxation of both these problems is
\begin{equation}\label{eq:def-maxcutk-sdp}
\openup\jot
\begin{aligned}[t]
\text{ maximize:}\quad &\frac{k-1}{k}\sum_{(i,j)\in E} w_{ij}(1 - X_{ij})\\
        \text{subject to:}\quad & X \in \mathrm{M}_N(\R)_{\geq 0},\\
                          \quad & X_{ii} = 1,\\
                          \quad & X_{ij} \geq -\frac{1}{k-1},     
\end{aligned}
\end{equation}

We mentioned earlier our $0.864$-approximation algorithm for $\ncut{3}$. This means that, if the value of the canonical SDP relaxation~\eqref{eq:vectorcut3-in-approximate-isometry} for an instance $I$ is $\sdp(I)$, the optimal value of $\ncut{3}$ is in between $0.864 \sdp(I)$ and $\sdp(I)$. Compare the ratio $0.864$ for the noncommutative problem with the best approximation ratio for classical $\cut{3}$, which stands at $0.836$~\cite{goemansmax3cut}. These results are comparable because we make use of the same SDP relaxation and every classical solution is a one-dimensional noncommutative solution. 

Our algorithm and its analysis extend directly to all $\ncut{k}$ problems, giving closed-form expressions for the approximation ratios (see Section \ref{sec:algorithm}). In Table \ref{tab:ratios} we list these ratios for some $k$ and compare them with those of the best known algorithms for the classical variant. We also calculate our approximation ratios for homogeneous problems for some values of $k$ in Table \ref{tab:ratios-homogeneous} and compare them with those of the classical variant.\footnote{In this paper we only prove the noncommutative ratios. The classical values that are not cited in Table \ref{tab:ratios-homogeneous} follow from a variant of our approximation framework for noncommutative problems tailored to the commutative special case. The main idea behind this classical variant is explained in Appendix \ref{sec:classical-rel-dist}.}

\begin{table}[h!]
\begin{center}
\begin{tabular}{ p{2.5cm}||p{3.5cm}|p{3.5cm} } 
$\cut{k}$ & Classical & Noncommutative \\
\hline
$k = 2$ & $.878$ \cite{goemansmaxcut} & $1$ \cite{tsirelson1}\\ 
\hline
$k = 3$ & $.836$ \cite{goemansmax3cut} & $.8649$\\
\hline
$k = 4$ & $.857$ \cite{klerk} & $.8642$ \\
\hline
$k = 5$ & $.876$ \cite{klerk} & $.8746$\\
\hline
$k = 10$ & $.926$ \cite{klerk} & $.9195$\\
\end{tabular}
\end{center}
\caption{Approximation ratios for classical and noncommutative $\cut{k}$ for various $k$.}\label{tab:ratios}
\end{table}

\begin{table}[h!]
\begin{center}
\begin{tabular}{ p{2.5cm}||p{3.5cm}|p{3.5cm} } 
$\homlin{2}{k}$ & Classical & Noncommutative \\
\hline
$k=2$ & $.878$ \cite{goemansmaxcut} & $1$ \cite{tsirelson1} \\ 
\hline
$k=3$ & $.793$ \cite{goemansmax3cut} & $.864$ \\
\hline
$k=4$ & $.708$ & $.813$ \\
\hline
$k=5$ & $.629$ & $.738$\\
\hline
$k=10$ & $.383$ & $.475$\\
\end{tabular}
\end{center}
\caption{Approximation ratios for classical and noncommutative $\homlin{2}{k}$ for various $k$.}\label{tab:ratios-homogeneous}
\end{table}

As indicated in Table \ref{tab:ratios}, when $k$ increases, our algorithm for the noncommutative problem does slightly worse than the best known algorithm for the classical problem. Therefore, for those problems, the best noncommutative algorithm known is still the one for the one-dimensional classical problem.

\begin{question}\label{q:best-ratio-for-cut-5}
Is it possible to approximate $\ncut{5}$ with a guarantee that is strictly better than the best-known algorithm for classical $\cut{5}$?
\end{question}

\begin{remark*}
    The reason for focusing on the $\homlin{2}{k}$ is that their canonical SDP relaxation \eqref{eq:def-homlin2k-sdp-vector} is real. As discussed in the introduction, so far we only know approximate isometries from real vector spaces. Since approximate isometries is a crucial ingredient in our approximation framework, one needs new ideas to be able to tackle all $\nlin{2}{k}$.
\end{remark*}

\subsection{Smooth CSPs}\label{sec:smooth-csps}
We now discuss a class of CSPs we call \emph{smooth CSPs} that fits naturally within our approximation framework. The key properties of this variant of $\lin{2}{k}$ are that the objective function is quadratic and rewards assignments based on their proximity to the perfectly satisfying assignment. This variant agrees with $\lin{2}{k}$ when $k\leq 3$. Smooth CSPs are motivated by a connection to Grothendieck inequalities, which is discussed further in Section~\ref{sec:grothendieck}. 

The smooth $\lin{2}{k}$, or $\slin{2}{k}$ for short, is the following optimization problem:
\begin{equation}\label{eq:slin2k}
\openup\jot
\begin{aligned}[t]
\text{ maximize:}\quad &\sum_{(i,j)\in \mathcal{E}} w_{ij} \Bigparen{{1 - \frac{1}{a_k} |x_j - \omega^{c_{ij}}x_i|^2}} + \sum_{(i,j)\in \mathcal{I}} \frac{w_{ij}}{a_k}|x_j - \omega^{c_{ij}}x_i|^2 \\
        \text{subject to:}\quad & x_i^k = 1 \text{ and } x_i \in \C,
\end{aligned}
\end{equation}
where the normalization factor $a_k$ ensures that the terms take only values between $0$ and $1$.\footnote{The normalization constant is $a_k=|1-\omega^{\floor{k/2}}|^2$, which is the largest distance between any two $k$-th roots of unity. For example $a_3 = 3$ and $a_k=4$ when $k$ is even.} Comparing this with \eqref{eq:lin2k}, almost-satisfying assignments are not rewarded in $\lin{2}{k}$ but they receive a non-zero reward in the smooth case.\footnote{The constraint $x_i^{-1}x_j=\omega^{c_{ij}}$ corresponds to $1-\frac{1}{a_k}\abs{x_j-\omega^{c_{ij}}x_i}^2$ in the objective function~\eqref{eq:slin2k}. For this constraint, the assignment $x_i = 1$ and $x_j = \omega^{c_{ij}\pm\delta}$ is almost satisfying if $\delta$ is small. The objective function of the smooth CSP rewards the assignment with small $\delta$ more than the assignment with large~$\delta$.} The objective function of $\lin{2}{k}$ is a polynomial of degree $2\floor{k/2}$ in the $2N$ variables $x_1, \ldots,x_N,x_1^*, \ldots, x_N^*$. On the other hand, the objective function of $\slin{2}{k}$ is always a quadratic polynomial.

The noncommutative analogue of $\slin{2}{k}$, denoted by $\nslin{2}{k}$, is 
\begin{equation}\label{eq:nslin2k}
\openup\jot
\begin{aligned}[t]
\text{ maximize:}\quad &\sum_{(i,j)\in \mathcal{E}} w_{ij} \paren{{1 - \frac{1}{a_k} \|X_j - \omega^{c_{ij}}X_i\|^2}} + \sum_{(i,j)\in \mathcal{I}} \frac{w_{ij}}{a_k}\|X_j - \omega^{c_{ij}}X_i\|^2 \\
        \text{subject to:}\quad & X_i^*X_i = X_i^k = 1,\\
\end{aligned}
\end{equation}
where $\|\cdot\|$ is the dimension-normalized Frobenius norm. This can be equivalently written as
\begin{equation}\label{eq:def-nslin2k-using-inner-product}
\openup\jot
\begin{aligned}[t]
\text{ maximize:}\quad &\sum_{(i,j)\in \mathcal{E}} w_{ij} \parens[\Big]{1 -\frac{2}{a_k} + \frac{2}{a_k}\Re\parens[\Big]{\omega^{-c_{ij}}\ip{X_i}{X_j}}}+ \sum_{(i,j)\in \mathcal{I}} \frac{2w_{ij}}{a_k}\parens[\Big]{1 - \Re\parens[\Big]{\omega^{-c_{ij}}\ip{X_i}{X_j}}} \\
        \text{subject to:}\quad & X_i^*X_i = X_i^k = 1,\\
\end{aligned}
\end{equation} 

Finally consider the unitary relaxation of $\nslin{2}{k}$, denoted by $\ulin{2}{k}$: 
\begin{equation}\label{eq:ulin2k}
\openup\jot
\begin{aligned}[t]
\text{ maximize:}\quad &\sum_{(i,j)\in \mathcal{E}} w_{ij} \paren{{1 - \frac{1}{a_k} \|X_j - \omega^{c_{ij}}X_i\|^2}} + \sum_{(i,j)\in \mathcal{I}} \frac{w_{ij}}{a_k}\|X_j - \omega^{c_{ij}}X_i\|^2 \\
        \text{subject to:}\quad & X_i^*X_i = 1,\\
                           \quad & \ip{X_i}{X_j} \in \Omega_k,
\end{aligned}
\end{equation}
where we replaced the order-$k$ constraint in \eqref{eq:nslin2k} with $\ip{X_i}{X_j} \in \Omega_k$. Here $\Omega_k$ denotes the convex hull of $k$-th roots of unity. This is a relaxation of the noncommutative problem \eqref{eq:nslin2k} because, if $X,Y$ are any order-$k$ unitary operators, their inner product $\ip{X}{Y}$ lies within $\Omega_k$. 

The approximation framework applied to smooth CSPs gives the following inequalities
\begin{align*}
\nslin{2}{3} &\geq 0.862\times \ulin{2}{3},\\
\nslin{2}{4} &\geq 0.919\times \ulin{2}{4},\\
\nslin{2}{5} &\geq 0.928\times \ulin{2}{5},\\
\nslin{2}{6} &\geq 0.951\times \ulin{2}{6},\\
\nslin{2}{7} &\geq 0.959\times \ulin{2}{7}.
\end{align*}
We prove these in Section \ref{sec:algorithm-smooth}.

The concept of unitary relaxation of smooth CSPs extends identically to all CSPs.
\begin{remark*} 
The discussion above shows that in the case of smooth $\lin{2}{k}$ the optimal and unitary values are close. In Table \ref{tab:ratios}, we saw that the SDP and optimal values of $\ncut{k}$ are close. Even though we do not expect the same to be true for homogeneous $\nlin{2}{k}$ (by for example looking at Table \ref{tab:ratios-homogeneous}), we can show that the SDP and unitary values are the same for homogeneous problems. This is due to the construction of approximate isometry, \emph{i.e.} the vector-to-unitary construction.
\end{remark*}

We end this section with a couple open problems.
\begin{question}\label{q:noncommutative-smooth-csp}
The unitary and optimal values of noncommutative smooth CSPs are close. How close are the SDP and unitary values of noncommutative smooth CSPs?
\end{question}
\begin{question}\label{q:classical-smooth-csp}
What are the best approximation ratios for classical smooth CSPs? 
\end{question}

\section{Generalized Weyl-Brauer Operators} \label{sec:gwb}

In this section, we formally define the generalized Weyl-Brauer operators we first introduced in Section \ref{sec:innovation-generalized-anticommutation}. In Section~\ref{sec:v-u-intro}, we define the generalized Weyl-Brauer group and vector-to-unitary constructions. In Sections~\ref{sec:gwb-inf} and~\ref{sec:construct-gwb} we discuss infinite-dimensional and finite-dimensional representations, respectively.

\subsection{Vector-to-Unitary Construction for Higher Orders}\label{sec:v-u-intro}

First, we recall the important properties of the vector-to-unitary construction introduced in Section~\ref{sec:approximate-isometries}. The construction critically hinges on the properties of the Weyl-Brauer operators $\sigma_1,\ldots,\sigma_n$:
\begin{itemize}
    \item Hermitian unitary: $\sigma_i^2=\sigma_i^\ast \sigma_i=\Id$,
    \item Anticommutation: $\sigma_i\sigma_j=-\sigma_j\sigma_i$ for all $i\neq j$.
\end{itemize}
Then, the \emph{vector-to-unitary construction}, mapping real vectors $\vec{x}=(x_1,\ldots,x_n)$ to operators $\sum_ix_i\sigma_i$ is isometric and sends every unit vector to an order-$2$ unitary.

\begin{remark}
One can construct Weyl-Brauer operators using Pauli matrices. Pauli matrices are
\[\sigma_x = \begin{bmatrix}0 &1\\1 &0\end{bmatrix},\sigma_y = \begin{bmatrix}0 &-i\\i &0\end{bmatrix},\sigma_z = \begin{bmatrix}1 &0\\0 &-1\end{bmatrix}\]
Combining these matrices with the $2$-by-$2$ identity matrix $I$ in the following arrangement we can construct any number of pairwise anticommuting Hermitian unitary operators:
\begin{alignat*}{7}
\sigma_1 &\coloneqq \sigma_x \ &\otimes \ &I      \ &\otimes \ &I      \ &\otimes \ &I \ \otimes \ \cdots\\
\sigma_2 &\coloneqq \sigma_y \ &\otimes \ &I      \ &\otimes \ &I      \ &\otimes \ &I \ \otimes \ \cdots\\
\sigma_3 &\coloneqq \sigma_z \ &\otimes \ &\sigma_x \ &\otimes \ &I      \ &\otimes \ &I \ \otimes \ \cdots\\
\sigma_4 &\coloneqq \sigma_z \ &\otimes \ &\sigma_y \ &\otimes \ &I      \ &\otimes \ &I \ \otimes \ \cdots\\
\sigma_5 &\coloneqq \sigma_z \ &\otimes \ &\sigma_z \ &\otimes \ &\sigma_x \ &\otimes \ &I \ \otimes \ \cdots\\
\sigma_6 &\coloneqq \sigma_z \ &\otimes \ &\sigma_z \ &\otimes \ &\sigma_y \ &\otimes \ &I \ \otimes \ \cdots
\end{alignat*}
\end{remark}

Could we come up with a vector-to-unitary construction that produces higher-order unitaries? Namely, could we find operators $\sigma_1,\ldots,\sigma_n$ such that, for every real unit vector $\vec{x}\in\R^n$, the operators $U_{\vec{x}}\coloneqq \sum_i x_i\sigma_i$ are order-$k$ unitaries and the mapping $\vec{x} \mapsto U_{\vec{x}}$ is an isometry? 

For this to happen, we clearly need $\sigma_i$ to be order-$k$ unitaries themselves (to see this let $\vec{x}$ be the standard basis vectors $\vec{e}_i$). Additionally, to guarantee unitarity of $U_{\vec{x}}$, we need for all real unit vectors $\vec{x}$ that
\begin{align*}
    \Id=U_{\vec{x}}^\ast U_{\vec{x}}=\sum_{i,j}x_ix_j \sigma_i^\ast \sigma_j.
\end{align*}
This holds if and only if $\sigma_i^\ast \sigma_j = - \sigma_j^\ast \sigma_i$ for every $i\neq j$. As in Section~\ref{sec:innovation-generalized-anticommutation}, we call this relation \emph{$*$-anticommutation}.

On the other hand, to guarantee the order-$k$ requirement we need
\begin{align*}
    U_{\vec{x}}^k=\sum_{i_1,\ldots,i_k}x_{i_1}\cdots x_{i_k}\sigma_{i_1}\cdots \sigma_{i_k}=\Id.
\end{align*}
It is however not clear how to satisfy this in general due to the normalisation $\|\vec{x}\|^2 = x_1^2 + \cdots + x_n^2 = 1$. We can instead guarantee that $U_{\vec{x}}^k\propto\Id$ by having $\sigma_i\sigma_j=\omega \sigma_j\sigma_i$ when $i<j$, for $\omega$ a primitive $k$-th root of unity. Then, taking the $k$-th power of $U_{\vec{x}}$, the cross terms cancel out and we get
\begin{align*}
    U_{\vec{x}}^k=\sum_{i_1,\ldots,i_k}x_{i_1}\cdots x_{i_k}\sigma_{i_1}\cdots \sigma_{i_k}=\sum_{i}x_i^k \sigma_i^k=\parens*{\sum_{i}x_i^k}\Id.
\end{align*}
But the two anticommutation-like properties 
\begin{itemize}
        \item $\sigma_i\sigma_j=\omega \sigma_j\sigma_i$ for $i < j$ and
        \item $\sigma_i^\ast \sigma_j=-\sigma_j^\ast \sigma_i$
    \end{itemize}
are incompatible. We show this with the following simple lemma.
\begin{lemma}\label{lem:impossiblity}
    Let $k>2$ and let $\omega$ be a primitive $k$-th root of unity. There are no unitary operators $X,Y$ satisfying both
    \begin{itemize}
        \item $XY=\omega YX$
        \item $X^\ast Y=-Y^\ast X$
    \end{itemize}
\end{lemma}

\begin{proof}
    Suppose such operators $X,Y$ exist. Using the second relation, $X^2=YY^\ast X^2=-Y X^\ast YX$. Then using the first relation $YX=\omega^\ast XY$, so $X^2=-\omega^\ast Y X^\ast XY=-\omega^\ast Y^2$. In particular, $X^2$ commutes with $Y$. However, using the first relation, $X^2 Y=\omega XYX=\omega^2 YX^2$, which implies that $\omega^2=1$, so $k\leq 2$, which is a contradiction.
\end{proof}
Though we omit the argument here, it is possible to strengthen Lemma~\ref{lem:impossiblity} to rule out the existence of any linear map $\vec{x} \mapsto U_{\vec{x}}$ that satisfies both the order-$k$ and unitary conditions. Nevertheless, we may choose either one of the two relations in Lemma \ref{lem:impossiblity} to generalize the Weyl-Brauer operators. Depending on the choice, one constructs two very different generalizations. From now on, we focus on the $\ast$-anticommutation.

\begin{definition}\label{def:infinite-gwb}
    The \emph{generalized Weyl-Brauer (GWB) group with $n$ generators} is
    \begin{align}
        \mathrm{GWB}_n=\group*{\sigma_1,\ldots,\sigma_n,J}{J^2,\,[\sigma_i,J],\,J(\sigma_i^{-1}\sigma_j)^2\;\forall\,i\neq j}.
    \end{align}
\end{definition}

Any unitary representation of $\gwb_n$ sending $J\mapsto -1$ gives a set of $n$ operators satisfying $\ast$-anticommutation relation. Conversely, any set of unitary operators satisfying $\ast$-anticommutation provides a representation of the group.

\begin{definition}\label{def:vector-to-unitary-construction}
    A \emph{vector-to-unitary construction} is a linear map $\R^n\to\mathrm{B}(\mc{H})$, $\vec{x}\mapsto\sum_ix_i\pi(\sigma_i)$, where $\pi:\gwb_n\to\mathrm{B}(\mc{H})$ is a unitary representation such that $\pi(J)=-1$.
\end{definition}

\begin{proposition}
    Any vector-to-unitary construction $\R^n\to\mathrm{B}(\mc{H})$ maps unit vectors to unitary operators.
\end{proposition}

\begin{proof}
    Write $u_{\vec{x}}=\sum_ix_i\pi(\sigma_i)$. Then, for any unit vector $\vec{x}\in\R^n$, 
    \begin{align*}
        &u_{\vec{x}}^\ast u_{\vec{x}}=\sum_{i,j}x_ix_j\pi(\sigma_i)^\ast\pi(\sigma_j)=\sum_ix_i^2+\sum_{i<j}x_ix_j(\pi(\sigma_i)^\ast\pi(\sigma_j)+\pi(\sigma_j)^\ast\pi(\sigma_i))=1\\
        &u_{\vec{x}}u_{\vec{x}}^\ast=\sum_{i,j}x_ix_j\pi(\sigma_i)\pi(\sigma_j)^\ast=1,
    \end{align*}
    so $u_{\vec{x}}$ is unitary.
\end{proof}

In Section \ref{sec:construct-gwb}, we introduce the order-$k$ \emph{generalized Weyl-Brauer (GWB) group with $n$ generators}, denoted by $\gwb_n^k$: this is obtained from $\gwb_n$ by adding relations that force the $\sigma_i$ to be order-$k$ as well as some extra relations to make the group finite. The group $\gwb_n^k$ is thus a homomorphic image of $\gwb_n$ (since it is a quotient group), and every representation of $\gwb_n^k$ is a representation of $\gwb_n$. See the formal definition of $\gwb_n^k$ (Definition~\ref{def:order-k-gwb}) in that section. 

\begin{definition}\label{def:order-k-vector-to-unitary-construction}
    An \emph{order-$k$ vector-to-unitary construction} is a vector-to-unitary construction where the representation $\pi$ in Definition \ref{def:vector-to-unitary-construction} is a finite-dimensional representation of $\gwb_n^k$. When it is clear from the context we may drop the order-$k$ attribute.
\end{definition}

We are interested in order-$k$ vector-to-unitary constructions $\vec{x}\mapsto U_{\vec{x}}$ that are isometric, \emph{i.e.} for every $\vec{x},\vec{y}\in\R^n$ we have $\ip{U_{\vec{x}}}{U_{\vec{y}}} = \ip{\vec{x}}{\vec{y}}$. In Section \ref{sec:construct-gwb} we present an order-$k$ vector-to-unitary construction that satisfies a much stronger isometry property. This property states that $\ip{U_{\vec{x}}^s}{U_{\vec{y}}^s} = \ip{\vec{x}}{\vec{y}}^s$ for every integer $1 \leq s < k$. This is an important property that is used in the analysis of our algorithm in Sections \ref{sec:algebraic-reldist} and \ref{sec:efficient-algorithm}.

\begin{definition}\label{def:strongly-isometric}
    An order-$k$ vector-to-unitary construction $U$ is \emph{strongly isometric} if $\ip{U_{\vec{x}}^s}{U_{\vec{y}}^s} = \ip{\vec{x}}{\vec{y}}^s$ for every integer $0 \leq s < k$.
\end{definition}

\begin{remark}
Even though it is impossible to have a construction such that $U_{\vec{x}}$ is always an order-$k$ unitary, there is a sense in which the strong isometry property is the next best thing we could hope for. Let us elaborate. Suppose $X$ and $Y$ are two order-$k$ unitaries. Then there exist some nonnegative numbers $a_0,a_1,\ldots,a_{k-1}$ that sum to $1$ such that 
\begin{equation}\label{eq:order-k-property}
\ip{X^s}{Y^s} = a_0 + a_1 \omega^s + \cdots + a_{k-1}\omega^{s(k-1)}
\end{equation}
for every integer $s$.\footnote{The coefficients are $a_s = \frac{1}{k} \sum_{t=0}^{k-1} \omega^{-ts}\ip{X^t}{Y^t}$. It is not hard to prove that these are nonnegative and sum to $1$.} So the inner products of powers of order-$k$ unitaries are tightly correlated.

Now if $U$ is strongly isometric, even though $X = U_{\vec{x}}$ and $Y = U_{\vec{y}}$ are not going to be order-$k$ unitaries in general, so they may not satisfy \eqref{eq:order-k-property}, we still have tight control over the inner products of their powers, namely that $\ip{X^s}{Y^s} = \ip{X}{Y}^s$ for $0\leq s <k$.
\end{remark}

\subsection{Infinite-Dimensional Case}\label{sec:gwb-inf}
This section could be skipped on a first reading. We supply it here because, for those readers familiar with the basics of von Neumann algebras, it gives insight into the construction of finite-dimensional representations of $\gwb_n^k$ in the next section. We collected the facts we use about von Neumann algebras in the preliminaries. 

The goal of this section is to characterise the infinite group $\gwb_n$ (Definition \ref{def:infinite-gwb}) and study its representation theory. We also exhibit a vector-to-unitary construction $u$ that satisfies $\ip{u_{\vec{x}}^s}{u_{\vec{y}}^s} = \ip{\vec{x}}{\vec{y}}^s$ for all non-negative integers $s$. The algorithmic implications of this construction to CSPs is presented in Section~\ref{sec:infinite-dimensional-algebraic-rel-dist}. 

It is evident that $\mathrm{GWB}_n$ is infinite: in fact, we have a surjective homomorphism $\mathrm{GWB}_n\rightarrow\Z$ such that $J\mapsto 0$ and $\sigma_i\mapsto 1$. In order to get a handle on the structure of $\gwb_n$, we consider an alternate presentation. In this presentation, the $\ast$-anticommutation relation $\sigma_i^{-1}\sigma_j=J\sigma_j^{-1}\sigma_i$ can be replaced by the standard anticommutation relation, and hence we can reason about this group using intuition from Pauli matrices.

\begin{lemma}
    The group $\mathrm{GWB}_n$ admits the presentation
    \begin{align}\label{eq:pauli-presentation}
        \mathrm{GWB}_n=\group*{p_1,\ldots,p_{n-1},c,J}{J^2,\,Jp_i^2,\,[p_i,J],\,[c,J],\,J(p_ip_j)^2\;\forall\,i\neq j}.
    \end{align}
    The new generators are related to the original ones by $c=\sigma_1$ and $p_i=\sigma_1^{-1}\sigma_{i+1}$. 
\end{lemma}
We call this the \emph{Pauli presentation} of $\gwb_n$ because the generators $p_i$ could be represented by Pauli matrices; consider the representation that sends $c \mapsto 1, J \mapsto -1$ and generators $p_i$ to pairwise anticommuting and anti-Hermitian Pauli matrices. In fact it is possible to show that the subgroup $P_n = \group{p_1,\ldots,p_{n-1}}$ is isomorphic to some subgroup of index at most $2$ of the Pauli group on $\floor{\tfrac{n-1}{2}}$ qubits.
\begin{proof}
    Let $c=\sigma_1$ and $p_i=\sigma_1^{-1}\sigma_{i+1}$. This is evidently a set of generators for $\gwb_n$. It suffices to rewrite the relations of $\mathrm{GWB}_n$ in terms of these generators. First, note that the order relation $J^2=e$ remains the same, and the commutation relations $[\sigma_i,J]=e$ becomes commutation with the new generators $[c,J]=e$ and $[p_i,J]=e$. The final relations $J(\sigma_i^{-1}\sigma_j)^2 = e$, if $i=1,j > 1$, become $e=Jp_{j-1}^2$; we get the same relation for $j=1,i > 1$; and for $i,j\neq 1$, $e=J(cp_{i-1}p_{j-1}^{-1}c^{-1})^2=cJ(p_{i-1}p_{j-1})^2c^{-1}$, which is the last relation $J(p_{i-1} p_{j-1})^2=e$.
\end{proof}
The quotient by the subgroup $\group{J}$ is the free product $\gwb_n/\group{J}\cong \Z\ast\Z_2^{n-1}$. So, there are no relations between $c$ and $p_i$ in $\gwb_n/\group{J}$. Also, $c,p_i,p_i^{-1}p_j \notin \{e,J\}$ in $\gwb_n$ unless $i=j$. Finally, since the quotient has trivial center, $\group{J}$ is the centre of $\gwb_n$.

The subgroup $C = \group{c,J}$ is isomorphic to $\Z \times \Z_2$, and $\group{J}$ is a subgroup of both $C$ and $P_n$. So, by the definition of the amalgamated product, we arrive at the characterisation \[\gwb_n\cong (\Z\times\Z_2)\ast_{\Z_2}P_n\] where $\ast_{\Z_2}$ is the amalgamated product over the subgroups $\Z_2\cong\group{J}$. 

Our last goal is to construct a special vector-to-unitary construction $u$ with the strong isometry property $\ip{u_{\vec{x}}^s}{u_{\vec{y}}^s} = \ip{\vec{x}}{\vec{y}}^s$ for all orders mentioned earlier, using the group von Neumann algebra $L(\gwb_n) \subset \mathrm{B}(\cH)$ where $\cH=\ell^2\gwb_n$.

To construct a representation where $J\mapsto-1$, we can take the quotient algebra \[\mc{M}=L(\gwb_n)/\group{e+J}.\] The element $p = (e+J)/2$ is a projection in $L(\gwb_n)$, \emph{i.e.} $p^2 = p$ and $p^* = p$. Since $\group{e+J} = \group{p}$ we have $\mc{M} = L(\gwb_n)/\group{p}$, and since $p$ is central, we have $\mc{M} \cong (e-p)L(\gwb_n)$. 

\begin{definition}\label{def:inf-dim-vector-to-unitary}
    The \emph{infinite-dimensional vector-to-unitary construction} is the map $u:\R^n\to \mc{M}$, defined as
    \begin{align}
        \vec{x}\mapsto u_{\vec{x}}=\sum_{i=1}^nx_i\bar{\sigma}_i,
    \end{align}
    where $\bar{\sigma}_i$ is the representative of the group element $\sigma_i$ in the quotient algebra.
\end{definition}
First, it is direct to see that, via the relations $\bar{\sigma}_i^\ast\bar{\sigma}_j=-\bar{\sigma}_j^\ast\bar{\sigma}_i$ in $M$ (since $\bar{J}=-1$), the vector-to-unitary construction sends unit vectors $\vec{x}\in\R$ to unitary elements in $\mc{M}$:
\begin{align*}
    &u_{\vec{x}}^\ast u_{\vec{x}}=\sum_{i,j}x_ix_j\bar{\sigma}_i^\ast\bar{\sigma}_j=\sum_ix_i^2+\sum_{i<j}x_ix_j(\bar{\sigma}_i^\ast\bar{\sigma}_j+\bar{\sigma}_j^\ast\bar{\sigma}_i)=1\\
    &u_{\vec{x}}u_{\vec{x}}^\ast=\sum_{i,j}x_ix_j\bar{\sigma}_i\bar{\sigma}_j^\ast=\sum_ix_i^2+\sum_{i<j}x_ix_j(\bar{\sigma}_i\bar{\sigma}_j^\ast+\bar{\sigma}_j\bar{\sigma}_i^\ast)=1
\end{align*}

Since $L(\mathrm{GWB}_n)$ is a group von Neumann algebra, it has a trace $\tau$ that extends the canonical trace on the group algebra $\C[\gwb_n]$. Now the algebra $(e-p)L(\gwb_n)$ has a trace \[\tau_M=\frac{\tau|_{(e-p)L(\mathrm{GWB}_n)}}{\tau(e-p)}\] where $\tau|_{(e-p)L(\mathrm{GWB}_n)}$ is the restriction of $\tau$ to the compression $(e-p)L(\mathrm{GWB}_n)$. Thus, on $\mc{M}$, this becomes the tracial state \[\tau_{\mc{M}}(\bar{a})=2\tau((e-p)a)=\tau((e-J)a)\] where $\bar{a}$ is the element of the quotient algebra with representative $a$, since $\tau(p) = \tau(\frac{e+J}{2}) = \frac{1}{2}$. 

The trace $\tau_{\mc{M}}$ gives rise to an inner product on $\mc{M}$. We show that the infinite-dimensional vector-to-unitary construction is an isometry with respect to this inner product.
\begin{lemma}
    For all vectors $\vec{x},\vec{y}\in\R^n$,
    \begin{align}
        \tau_{\mc{M}}(u_{\vec{x}}^\ast u_{\vec{y}})=\ip{\vec{x}}{\vec{y}}
    \end{align}
\end{lemma}

\begin{proof}
    We can simply calculate the inner product using the Pauli presentation of $\mathrm{GWB}_n$: writing $e = p_0$,
    \begin{align*}
        \tau_{\mc{M}}(u_{\vec{x}}u_{\vec{y}})=\sum_{i,j=1}^n x_iy_j\tau((e-J)\sigma_i^{-1}\sigma_j)=\sum_{i,j=1}^n x_iy_j\tau((e-J)p_{i-1}^{-1}p_{j-1})=\sum_{i=1}^n x_iy_i=\ip{\vec{x}}{\vec{y}},
    \end{align*}
    since $p_{i-1}^{-1}p_{j-1}\in\{e,J\}$ in the group $\gwb_n$ if and only if $i=j$.
\end{proof}

In fact, we have the stronger property characterising the inner products of powers of the unitaries.

\begin{proposition}[Strong isometry]\label{prop:power-property}
    For all vectors $\vec{x},\vec{y}\in\R^n$ and $m\in\N$,
    \begin{align}
        \tau_{\mc{M}}((u_{\vec{x}}^{m})^\ast u_{\vec{y}}^m)=\ip{\vec{x}}{\vec{y}}^m.
    \end{align}
\end{proposition}
Note that by traciality of the state and the fact that the vectors are real, we get that $\tau_{\mc{M}}((u_{\vec{x}}^{m})^\ast u_{\vec{y}}^m)=\group{\vec{x},\vec{y}}^{|m|}$ for all $m\in\Z$.
\begin{proof}
    We calculate the inner product as in the previous lemma. Writing again $e=p_0$, we have
    \begin{align*}
        \tau_{\mc{M}}((u_{\vec{x}}^{m})^\ast u_{\vec{y}}^m)&=\sum_{\substack{i_1,\ldots,i_m\\j_1,\ldots,j_m}}x_{i_1}y_{j_1}\cdots x_{i_m}y_{j_m}\tau((e-J)\sigma_{i_m}^{-1}\cdots\sigma_{i_1}^{-1}\sigma_{j_1}\cdots\sigma_{j_m})\\
        &=\sum_{\substack{i_1,\ldots,i_m\\j_1,\ldots,j_m}}x_{i_1}y_{j_1}\cdots x_{i_m}y_{j_m}\tau((e-J)p_{i_m-1}^{-1}c^{-1}\cdots c^{-1}p_{i_1-1}^{-1}p_{j_1-1}c\cdots cp_{j_m-1})\\
        &=\sum_{i_1,\ldots,i_m}x_{i_1}y_{i_1}\cdots x_{i_m}y_{i_m}=\group{x,y}^m,
    \end{align*}
    as, for $(i_1,\ldots,i_m)\neq(j_1,\ldots,j_m)$, $p_{i_m-1}^{-1}c^{-1}\cdots c^{-1}p_{i_1-1}^{-1}p_{j_1-1}c\cdots cp_{j_m-1}\notin\{e,J\}$ owing to the fact that it is not identity in the quotient $\mathrm{GWB}_n/\group{J}$ since there are no relations between the $p_i$ and $c$ there.
\end{proof}

\subsection{Finite-Dimensional Case}\label{sec:construct-gwb}

In this section we add relations to $\gwb_n$ and obtain a finite group $\gwb_n^k$ with order-$k$ generators $\sigma_i$. We then study the representation theory of $\gwb_n^k$ and give an strongly isometric order-$k$ vector-to-unitary construction. 

We construct the group $\gwb_n^k$ in two steps. We first introduce the following variation of $\gwb_n$. Let $f_{i,j,t}\in\Z_2$ for $i,j\in \{1,\ldots,n-1\}$ and $t\in\Z_k$ be such that $f_{i,i,0}=0$ and $f_{i,j,t}=f_{j,i,-t}$. Define $G_{n,f}^k$ as the group
\begin{align}
    \group*{p_1,\ldots,p_{n-1},c,J}{J^2,c^k,Jp_i^2,[p_i,J],[c,J],J^{f_{i,j,t}}[p_i,c^{-t}p_jc^t]}.
\end{align}
The difference between this and $\gwb_n$ is the addition of the order-$k$ relation $c^k=e$ and the commutation relations $J^{f_{i,j,t}}[p_i,c^{-t}p_jc^t] = e$.
We show that this group is finite by introducing a normal form for its elements using the commutation relations in the presentation. Let $\leq$ denote the lexicographic order on $\Z\times\Z_k$. For $\alpha\in\mathrm{M}_{[n-1]\times\Z_k}(\Z_2)$, write \[p^\alpha=\prod_{(i,t)\in [n-1]\times\Z_k}c^{-t}p_i^{\alpha_{i,t}}c^t,\] where the terms are ordered with $\leq$. We also use the shorthand notation $p_{i,t}=c^{-t}p_ic^t$.

Every missing proof in this section is given in Appendix \ref{sec:representation-calculations}.
\begin{lemma}\label{lem:claim1} Every element of $G_{n,f}^k$ is of the form $J^{b}c^sp^\alpha$ for some $(b,s,\alpha) \in \Z_2\times\Z_k\times\mathrm{M}_{[n-1]\times\Z_k}(\Z_2)$. 
\end{lemma}
We call $J^{b}c^sp^\alpha$ the normal form of an element in $G_{n,f}^k$. Since the index set $\Z_2\times\Z_k\times\mathrm{M}_{[n-1]\times\Z_k}(\Z_2)$ is of size $2k\cdot 2^{k(n-1)}$, we conclude that $|G_{n,f}^k| \leq 2k\cdot 2^{k(n-1)}$. In fact, we can prove equality.
\begin{lemma}\label{lem:unique-normal-form}
    $|G_{n,f}^k| = 2k\cdot 2^{k(n-1)}$. In particular every element in $G_{n,f}^k$ has a unique normal form.
\end{lemma}
For the rest of this discussion fix $f_{i,j,t}$  as follows:
\begin{align}
f_{i,j,t} = \begin{cases}
1 & i \neq j, t = 0,\\
1 & i \leq j, t = 1,\\
1 & i \geq j, t = -1,\\
0 & \text{otherwise}.
\end{cases}
\end{align}
With this choice of $f_{i,j,t}$, we have that $G_{n,f}^k$ is a quotient of $\gwb_n$.

For every $i \in [n-1]$ let $r_i \coloneqq J^{k}p_i(c^{-1}p_i c)(c^{-2}p_i c^2)\cdots(c^{-(k-1)}p_i c^{k-1})$. Let $H$ be the subgroup generated by $r_1,\ldots,r_{n-1}$.
\begin{lemma}\label{lem:centrality}
$H$ is a central subgroup, and thus is trivially normal in $G_{n,f}^k$.
\end{lemma}
\begin{definition}\label{def:order-k-gwb}The order-$k$ generalized Weyl-Brauer group with $n$ generators is $\gwb_n^k = G_{n,f}^k/H.$
\end{definition}

Since $\gwb_n^k$ is also a quotient of $\gwb_n$, all the representation of $\gwb_n^k$ are also representations of $\gwb_n$. 

In the same way as for $\gwb_n$, $\gwb_n^k$ can be presented in terms of generators $\sigma_1=c$ and $\sigma_{i+1}=cp_{i}$ for $i \in [n-1]$. In particular, the relations $\sigma_i^k=e$ hold in $\gwb_n^k$. This follows by expressing the relation $r_{i-1} = e$ in terms of these generators.
\begin{theorem}\label{thm:gnk-order}
    The group $\gwb_n^k$ is finite of order $2k\cdot 2^{(k-1)(n-1)}$. Further, the words $p^{\alpha}$ for $\alpha\in\mathrm{M}_{[n-1]\times\Z_k}(\Z_2)$ with $\alpha_{i,k-1}=0$, for all $i \in [n-1]$, are equal to $e$ or $J$ if and only if $\alpha=0$.
\end{theorem}

Note that in this theorem, the non-trivial words are chosen to not involve any of the elements $p_{i,k-1}$ due to the relation $p_{i,k-1}=J^{k-1}p_{i,0}\cdots p_{i,k-2}$ for all $i \in [n-1]$.

\begin{corollary}\label{cor:fin-dim-vector-to-unitary}
    There exists a finite-dimensional representation $\pi$ of $\gwb_n^k$ of dimension $k\cdot2^{(k-1)(n-1)}$, where $\pi(J) = -1$ and $\tr(\pi(p^{\alpha}))=0$ for all nonzero $\alpha\in\mrm{M}_{[n-1]\times\Z_k}(\Z_2)$ such that $\alpha_{i,k-1}=0$ for all $i \in [n-1]$. 
\end{corollary}
This representation gives rise to the order-$k$ vector-to-unitary construction $U_{\vec{x}}=\sum_i x_i\pi(\sigma_i)$ (see Definition~\ref{def:order-k-vector-to-unitary-construction}). In the rest of this paper, when we talk about the order-$k$ vector-to-unitary construction we are always referring to $\vec{x} \mapsto U_{\vec{x}}$ just defined.
\begin{proof}
    The proof is similar to the infinite-dimensional case, except now the Hilbert space $\mc{H} = \ell^2 G = \C[G]$ is finite-dimensional of dimension $2k\cdot 2^{(n-1)(k-1)}$, and the von Neumann group algebra is simply $\C[G]$. The group algebra $\C[G]$ acts by left multiplication on $\mc{H}$, therefore it embeds into $\mathrm{B}(\mc{H})$. Let $\tr$ denote the canonical normalised trace on $\mathrm{B}(\mc{H})$.

    Notice that $q = (1-J)/2$ is a projection with trace $\tau(q) = \frac{1}{2}$. Therefore $\mc{H}_0 = q \mc{H}$ is a subspace of dimension $k\cdot 2^{(n-1)(k-1)}$. Let $\pi:\C[G] \to \mathrm{B}(\mc{H}_0)$, be defined so that $\pi(x) = q x$. Since $q$ is a central projection, it is easy to verify that $\pi$ is a representation and that $\pi(J) = -1$. 
    
    Finally notice that $\tr(\pi(x)) = \tr((1-J)x)$, defines a trace on $\mathrm{B}(\mc{H}_0)$. Now since $p^\alpha\notin\{e,J\}$ by Theorem \ref{thm:gnk-order}, $\tr(\pi(p^\alpha))=0$.
\end{proof}

For simplicity from here on we identify $\sigma_i$ with its image under the representation $\pi(\sigma_i)$.

\begin{corollary}[Strong isometry.]\label{cor:fin-dim-weak-k-power}
    The order-$k$ vector-to-unitary construction $U_{\vec{x}}=\sum_i x_i \sigma_i$ is strongly isometric, \emph{i.e.} it satisfies 
    \begin{align*}
        \tr\squ{(U_{\vec{x}}^s)^* U_{\vec{y}}^{s}}=\ip{\vec{x}}{\vec{y}}^s,
    \end{align*}
    for all $s=0,\ldots,k-1$ and $\vec{x},\vec{y} \in \R^n$.
\end{corollary}

\begin{proof}
    Writing $e = p_0$
    \begin{align*}
    \tr\squ{(U_{\vec{x}}^s)^*U_{\vec{y}}^{s}}&=\sum_{\substack{i_1,\ldots,i_s\\j_1,\ldots,j_s}}x_{i_{1}}y_{j_{1}}x_{i_{2}}y_{j_{2}}\cdots x_{i_{s}}y_{j_{s}} \tr\Bigparen{(1-J)\sigma_{i_s}^{-1}\cdots\sigma_{i_2}^{-1}\sigma_{i_1}^{-1}\sigma_{j_1}\sigma_{j_2}\cdots\sigma_{j_s}}\\
    &=\sum_{\substack{i_1,\ldots,i_s\\j_1,\ldots,j_s}}x_{i_{1}}y_{j_{1}}x_{i_{2}}y_{j_{2}}\cdots x_{i_{s}}y_{j_{s}}\tr\Bigparen{(1-J)p_{i_s-1}^{-1}c^{-1}\cdots p_{i_2-1}^{-1}c^{-1}p_{i_1-1}^{-1}c^{-1}c p_{j_1-1} c p_{j_2-1}\cdots c p_{j_s-1}}\\
    &=\sum_{\substack{i_1,\ldots,i_s\\j_1,\ldots,j_s}}x_{i_{1}}y_{j_{1}}x_{i_{2}}y_{j_{2}}\cdots x_{i_{s}}y_{j_{s}}\tr\Bigparen{(1-J)p_{i_s-1,0}^{-1}p_{i_{s-1}-1,1}^{-1}\cdots p_{i_1-1,s-1}^{-1} p_{j_1-1,s-1} \cdots p_{j_s-1,0}},
    \end{align*}
    where in the last line we used the shorthand notation $p_{i,t} = c^{-t} p_i c^t$. By Theorem \ref{thm:gnk-order} \[p_{i_s-1,0}^{-1}\cdots p_{i_1-1,s-1}^{-1} p_{j_1-1,s-1}\cdots p_{j_s-1,0}\in\{e,J\}\] if and only if $(i_1,\ldots,i_l)=(j_1,\ldots,j_l)$. Thus 
    \begin{align*}
    \tr\squ{(U_{\vec{x}}^s)^*U_{\vec{y}}^{s}} = \sum_{i_1,\ldots,i_s}x_{i_1}y_{i_1}\cdots x_{i_s} y_{i_s}=\ip{\vec{x}}{\vec{y}}^s.
    \end{align*}
\end{proof}

\section{Relative Distribution} \label{sec:relative-distribution}

In this section, we study the relative distribution introduced in Section~\ref{sec:innovation-relative-distribution}.

\subsection{Definition and Characterisation}\label{sec:relative-distribution-definition}
First, we formally define the relative distribution.
\begin{definition}\label{def:weight-measure}
    Let $A,B\in\mathrm{U}_d(\C)$. Let $A=\sum_{s}e^{i\phi_s}\Pi_{A,s}$ and $B=\sum_{t}e^{i\psi_t}\Pi_{B,t}$ be the spectral decompositions of $A$ and $B$.  Define the \emph{weight measure} $w_{A,B}:\scr{B}([0,2\pi))\rightarrow\R_{\geq 0}$ as
    \begin{align}
        w_{A,B}(E)=\sum_{s,t:\;\psi_t-\phi_s\in E}\ip{\Pi_{A,s}}{\Pi_{B,t}}.
    \end{align}
\end{definition}
Since the phases $\psi_t$ and $\phi_s$ are considered as angles, the expression $\psi_t-\phi_s\in E$ is modulo $2\pi$. It follows immediately from the definition that $w_{A,B}$ is a well-defined measure, since it is countably additive. Also, we can see that it is a probability distribution since $w_{A,B}([0,2\pi))=\sum_{s,t}\ip{\Pi_{A,s}}{\Pi_{B,t}}=1$.

Also, it is important to note that the weight measure can be simply expressed in terms of the Dirac delta measure on $[0,2\pi)$ (Eq. \ref{eq:dirac})
$$w_{A,B}=\sum_{s,t}\ip{\Pi_{A,s}}{\Pi_{B,t}}\delta_{\psi_t-\phi_s}.$$

\begin{definition}\label{def:relative-measure}
    Let $A,B\in\mathrm{U}_d(\C)$ and $U\sim\mathrm{Haar}(\mathrm{U}_d(\C))$. Define the \emph{relative distribution} of $(A,B)$ as the probability measure $\Rel_{A,B}:\scr{B}([0,2\pi))\rightarrow\R_{\geq 0}$ such that
    \begin{align}
        \Rel_{A,B}(E)=\expect_U(w_{UA,UB}(E)).
    \end{align}
\end{definition}

\begin{lemma}
    The relative distribution is a well-defined probability distribution.
\end{lemma}
\begin{proof}
    First, to see that $\Delta_{A,B}(E)$ is well-defined for all $E$, we need that $U\mapsto w_{UA,UB}(E)$ is a measurable function. In fact, it is the composition of the multiplication $U\mapsto (UA,UB)$, which is continuous; the diagonalisation $(X,Y)\mapsto(\{\phi_s\},\{\psi_t\},\{\Pi_{X,s}\},\{\Pi_{Y,t}\})$ where $X=\sum_te^{i\phi_s}\Pi_{X,s}$ and $Y=\sum_te^{i\psi_t}\Pi_{Y,t}$, which is continuous almost everywhere; and the map $(\{\phi_s\},\{\psi_t\},\{\Pi_{X,s}\},\{\Pi_{Y,t}\})\mapsto\sum_{s,t:\;\phi_s-\psi_t\in E}\ip{\Pi_{X,s}}{\Pi_{Y,t}}$, which is measurable. Therefore, the Haar integral $\expec_U w_{UA,UB}(E)$ is well-defined.

    Next, since $w_{UA,UB}(E)\geq 0$ for all $U$, so $\Delta_{A,B}(E)\geq 0$. Also, $\Delta_{A,B}$ is countably additive by monotone convergence theorem. As such, $\Delta_{A,B}$ is also a measure. Finally, $\Delta_{A,B}([0,2\pi))=\expec_Uw_{UA,UB}([0,2\pi))=1$, so it is normalised, giving a distribution.
\end{proof}

Although $w_{A,B}$ is not absolutely continuous with respect to the Lebesgue measure $\Lambda$ on $[0,2\pi)$, the relative distribution $\Rel_{A,B}$ in fact is. For any measurable subset $E\subseteq[0,2\pi)$, if $E$ has measure $0$, the subset of $\mathrm{U}_d(\C)$ whose eigenvalues are in $E$ must also have measure $0$. As such $w_{UA,UB}(E)=0$ with probability $1$, so $\Rel_{A,B}(E)=0$. Therefore the Radon-Nikodym derivative $p_{\Rel_{A,B}}\coloneqq\frac{d\Rel_{A,B}}{d\Lambda}$ exists --- this is the PDF of the relative distribution.

\begin{lemma}\label{lem:fourier-coeffs}
    The characteristic function of the relative distribution is
    $$\chi_{\Delta_{A,B}}(n)=\int\tr\squ*{U^{-n}(UD)^{n}}dU,$$
    where $D$ is the diagonalisation of $A^\ast B$.
\end{lemma}

\begin{proof}
    First, we find the characteristic function of the weight distribution:
    \begin{align*}
        \chi_{w_{A,B}}(n)&=\sum_{s,t}\ip{\Pi_{A,s}}{\Pi_{B,t}}\chi_{\delta_{\psi_t-\phi_s}}(n)=\sum_{s,t}\ip{\Pi_{A,s}}{\Pi_{B,t}}e^{in(\psi_t-\phi_s)}\\
        &=\ip{\sum_s e^{in\phi_s}\Pi_{A,s}}{\sum_t e^{in\psi_t}\Pi_{B,t}}=\ip{A^{n}}{B^{n}}=\tr\squ*{A^{-n}B^{n}}.
    \end{align*}
    This immediately gives, by integrating over the Haar measure, that the characteristic function of the relative distribution is $$\chi_{\Delta_{A,B}}(n)=\int\tr\squ*{(UA)^{-n}(UB)^n}dU.$$

    To finish, we can simplify the expression using Haar invariance. First, we can replace $U$ by $UA^\ast$ and get $\chi_{\Delta_{A,B}}(n)=\int\tr\squ*{U^{-n}(UA^*B)^{n}}dU$. Next, we can diagonalise $A^\ast B=VDV^\ast$ and again use Haar invariance to replace $U$ by $VUV^\ast$, giving
    $$\chi_{\Delta_{A,B}}(n)=\int\tr\squ*{(VUV^\ast)^{-n}(VUV^*VDV^\ast)^{n}}dU=\int\tr\squ*{U^{-n}(UD)^n}dU$$
\end{proof}

In order to motivate the theorem below, we can compute some of the values of the characteristic function. First, note that $\chi_{\Delta_{A,B}}(-n)=\chi_{\Delta_{A,B}}(n)^\ast$ since $\Delta_{A,B}$ is real-valued. As such, we need only compute the values of the characteristic function for $n\geq0$. Next, it's easy to calculate the characteristic function for small $n$. Letting $\lambda=\tr(D)=\tr(A^*B)$,
\begin{align*}
    &\chi_{\Delta_{A,B}}(0)=\int\tr\squ*{U^{-0}(UD)^0}dU=\int dU=1\\
    &\chi_{\Delta_{A,B}}(1)=\int\tr\squ*{U^\ast UD}dU=\int\tr(D)dU=\lambda\\
    &\chi_{\Delta_{A,B}}(2)=\int\tr\squ*{(U^\ast)^2(UD)^2}dU=\tr\squ[\Big]{D\int UDU^\ast dU}=\tr\squ*{D\tr(D)}=\lambda^2.
\end{align*}
It is also feasible, though far less direct, to calculate the characteristic function at $n=3$. We get 
$$\chi_{\Delta_{A,B}}(3)=\lambda^3+\frac{\tr(D^3)+\lambda^3-2\tr(D^2)/d}{d^2-1}.$$
We see here that, unlike for smaller $n$, the value is not simply $\lambda^n$. However, it does converge there as $d\rightarrow\infty$. On the other hand, the distribution with characteristic function $\lambda^n$ for $n\geq 0$ and $(\lambda^\ast)^{-n}$ for $n< 0$ is exactly the wrapped Cauchy distribution $\Delta_\lambda = \mc{W}(\measuredangle\lambda,-\ln|\lambda|)$ (see Preliminaries \ref{sec:probability-theory}). As such, we have the following theorem. In the statement of the theorem $I_{m}$ is the identity operator on $\M_{m}(\C)$.

\begin{theorem}[Cauchy Law]\label{thm:relative-distribution} 
Let $A,B \in \mathrm{U}_{d}(\C)$ such that $\lambda = \tr(A^*B) = \ip{A}{B}$. In the limit of $m\to\infty$, the relative distribution $\Delta_{A\otimes I_{m},B\otimes I_{m}}$ converges in distribution to $\Delta_\lambda$.
\end{theorem}

The proof makes use of techniques from free probability (Preliminaries \ref{sec:free-probability}). 

\begin{proof}
    Let $A_m$ be the diagonalisation of $A^\ast B \otimes I_m$, and let $U_m$ be a $dm$-dimensional Haar random unitary. Then, by Theorem \ref{thm:asymptotically-free}, we get $U_m,A_m\rightarrow u,a$ ($\ast$-dist) where $u$ and $a$ are free. As such, for every $n\geq 0$
    $$\chi_{\Delta_{A\otimes I_m,B\otimes I_m}}(n)\rightarrow \tau(u^{-n}(ua)^n ),$$
    as $m\rightarrow\infty$.
    This is a mixed $\ast$-moment of $u,a$, so we can use freeness to compute it in terms of the $\ast$-moments of $u$ and $a$. To rewrite this in terms of free variables of expectation $0$, note that $\tau(a)=\lim_{m\rightarrow\infty}\tr(A_m)=\lambda$. Then
    \begin{align*}
        \tau((ua)^n u^{-n})&=\tau((u(a-\lambda+\lambda))^nu^{-n})=\sum_{e_1,\ldots,e_n\in\{0,1\}}\lambda^{n-\sum e_i}\tau(u(a-\lambda)^{e_1}\cdots u(a-\lambda)^{e_n}u^{-n})
    \end{align*}
    If the $e_i$ are not all $0$, then the argument of $\tau$ is a product of free variables each with expectation $0$, so that term does not contribute. That is, we get
    \begin{align*}
        \tau((ua)^n u^{-n})&=\lambda^n\tau(1)+\sum_{k=1}^n\lambda^{n-k}\sum_{n_1,\ldots,n_k>0,\sum n_k\leq n}\tau(u^{n_1}(a-\lambda)\cdots u^{n_k}(a-\lambda)u^{-\sum n_k})\\
        &=\lambda^n+\sum_{k=1}^n\lambda^{n-k}\sum_{n_1,\ldots,n_k>0,\sum n_k\leq n}0=\lambda^n,
    \end{align*}
    since $\tau(a-\lambda)=0$ and $\tau(u^n)=0$ for all $n\neq 0$. This gives exactly the characteristic function of the wrapped Cauchy distribution. Using Theorem~\ref{thm:levy}, $\Delta_{A\otimes I_m,B\otimes I_m}$ converges in distribution to $\Delta_\lambda$
\end{proof}

We end this section with an open problem about the possibility of extending the relative distribution to the case of three operators. This is the main roadblock in extending our framework to $3$-CSPs.
\begin{question}\label{q:relative-dist-for-3}
Is there an analogue of relative distribution for three operators? 
\end{question}

\subsection{Integral Formula}\label{sec:fidelity-integral-formula}

In this section, we make use of the relative distribution introduced in the previous section to calculate the expected value of polynomials of rounded operators. This was first discussed in Section~\ref{sec:innovation-relative-distribution}. We show that these expected values can be expressed as an integral of a function depending only on the polynomial -- that we call the \emph{fidelity} function -- with respect to the relative distribution (see Eq. \ref{eq:integral-formula} in the introduction).

In the discussion below, for every $k$ we work with a $*$-polynomial of the form $P(x,y) = \sum_{s,t=0}^{k-1} c_{s,t} (x^*)^sy^t$. In our applications, the variables $x,y$ will always be substituted with order-$k$ unitaries. Because of this we also assume that the indices $s,t$ in $c_{s,t}$ are considered modulo $k$, \emph{i.e.} we let $c_{s+ak,t+bk} = c_{s,t}$ for all integers $a,b$.
\begin{definition}\label{def:fidelity}
    For every $k$ and a $*$-polynomial of the form $P(x,y) = \sum_{s,t=0}^{k-1} c_{s,t} (x^*)^sy^t$, the \emph{fidelity function} is the map $\fid_{k,P}:[0,2\pi)\rightarrow\C$ defined as
    \begin{align}
        \fid_{k,P}(\theta)=\frac{1}{2\pi}\int_0^{2\pi}\sum_{s,t=0}^{k-1}c_{s,t}\widetilde{e^{i\varphi}}^{-s}\widetilde{e^{i(\varphi+\theta)}}^td\varphi,
    \end{align}
    where $\tilde{z}$ denotes the closest $k$-th root of unity to $z$, and the integral is with respect to the Lebesgue measure over $[0,2\pi)$. When $k$ is clear from the context we let $\fid_P \coloneqq \fid_{k,P}$.
\end{definition}

\begin{theorem}\label{lem:noncommutative-fid-delta-integral} Let $A,B\in\mathrm{U}_d(\C)$ be two unitaries, and let $U\sim\mathrm{Haar}(\mathrm{U}_d(\C))$ be a Haar-random unitary. Define the random variables $\tilde{X}$ and $\tilde{Y}$ as the closest (in little Frobenius norm) order-$k$ unitaries to $X=UA$ and $Y=UB$, respectively. For any $\ast$-polynomial $P$ as above,
\begin{equation}
\expect P(\tilde{X},\tilde{Y}) = \int \fid_P(\theta) d\Rel_{A,B}(\theta).
\end{equation}
\end{theorem}
\begin{proof}
    Let $X = UA=\sum_i\alpha_iP_{i}$ and $Y = UB=\sum_j\beta_jQ_{j}$ be the spectral decompositions. So, by Hoffmann-Wielandt theorem, $\tilde{X}=\sum_i\tilde{\alpha}_iP_{i}$, and similarly for $Y$. This gives that
    \begin{align*}
        \expect \bigparen{\sum_{s,t}c_{s,t} \ip{\tilde{X}^{s}}{\tilde{Y}^t}}= \sum_{i,j}\sum_{s,t}c_{s,t}\expect \bigparen{\tilde{\alpha}_i^{-s}\tilde{\beta}_j^t\ip{P_{i}}{Q_{j}}}
    \end{align*}
    Using Haar invariance, we can multiply $U$ by any root of unity $\zeta$, which sends $(\alpha_i,\beta_j)$ to $(\zeta\alpha_i,\zeta\beta_j)$. In particular, we can multiply by a random $e^{i\varphi}$, giving
    \begin{align*}
        \expect \bigparen{\sum_{i,j}\sum_{s,t}c_{s,t}\tilde{\alpha}_i^{-s}\tilde{\beta}_j^t\ip{P_{i}}{Q_{j}}}&=\frac{1}{2\pi}\int_0^{2\pi}\expect \bigparen{\sum_{i,j}\sum_{s,t}c_{s,t}\widetilde{e^{i\varphi}\alpha_i}^{-s}\widetilde{e^{i\varphi}\beta_j}^t\ip{P_{i}}{Q_{j}}}d\varphi\\
        &=\expect \bigparen{\sum_{i,j}\sum_{s,t}c_{s,t}\frac{1}{2\pi}\int_0^{2\pi}\widetilde{e^{i\varphi}}^{-s}\widetilde{e^{i\varphi}\alpha_i^\ast\beta_j}^td\varphi\ip{P_{i}}{Q_{j}}}\\
        &=\expect \bigparen{\sum_{i,j}\fid_P(\measuredangle \alpha_i^\ast\beta_j)\ip{P_{i}}{Q_{j}}}.
    \end{align*}
    Now, using the weight measure $w_{UA,UB}$ (Definition \ref{def:weight-measure}), we can view the sum over the eigenvalues as an integral with respect to this measure:
    \begin{align*}
        \sum_{i,j}\fid_P(\measuredangle(\alpha_i^\ast\beta_j))\ip{P_i}{Q_j}=\int\fid_P(\theta)dw_{UA,UB}(\theta).
    \end{align*}
    Finally, using the definition of the relative distribution (Definition \ref{def:relative-measure}), we get that
    \begin{align*}
        \expect_U\int\fid_P(\theta)dw_{UA,UB}(\theta)=\int\fid_P(\theta)d\Delta_{X,Y}(\theta).
    \end{align*}
\end{proof}

Due to the results of the previous section, we know that the relative distribution becomes a wrapped Cauchy distribution in the limit of large dimension. In this case, we can exactly compute the integral formula. First, we need a simpler expression of the fidelity function.

\begin{lemma}\label{lem:general-fidelity}
    The fidelity function defined in Definition~\ref{def:fidelity} simplifies to
    \begin{align}
        \fid_P(\theta)=\sum_{t=0}^{k-1} c_{t,t}\omega^{t\floor{\frac{k}{2\pi}\theta}}\parens*{1+(\omega^t-1)\parens*{\tfrac{k}{2\pi}\theta-\floor*{\tfrac{k}{2\pi}\theta}}}.
    \end{align}
    Its Fourier coefficients are $\widehat{\fid}_P(0)=2\pi c_{0,0}$ and for $n\neq 0$,
    \begin{align}
        \widehat{\fid}_P(n)=\frac{2k^2}{\pi n^2}\sin^2\parens*{\frac{\pi n}{k}} c_{n,n}.
    \end{align}
\end{lemma}

\begin{proof}
    Consider the integral $f_{s,t}(\theta)=\int_0^{2\pi}\widetilde{e^{i\varphi}}^{-s}\widetilde{e^{i(\varphi+\theta)}}^td\varphi$. This simplifies as
    \begin{align*}
        f_{s,t}(\theta)&=\sum_{r=0}^{k-1}\int_{-\frac{\pi}{k}+\frac{2\pi}{k}r}^{\frac{\pi}{k}+\frac{2\pi}{k}r}\widetilde{e^{i\varphi}}^{-s}\widetilde{e^{i(\varphi+\theta)}}^td\varphi=\sum_{r=0}^{k-1}\omega^{(t-s)r}\int_{-\frac{\pi}{k}}^{\frac{\pi}{k}}\widetilde{e^{i\varphi}}^{-s}\widetilde{e^{i(\varphi+\theta)}}^td\varphi=k\delta_{s,t}\int_{-\frac{\pi}{k}}^{\frac{\pi}{k}}\widetilde{e^{i(\varphi+\theta)}}^td\varphi.
    \end{align*}
    Now, $\widetilde{e^{i(\varphi+\theta)}}^t$ is a piecewise constant function (\emph{i.e.} a step function), with the value changing at $x_\theta=\frac{2\pi}{k}\floor{\frac{k}{2\pi}\theta}-\theta+\frac{\pi}{k}$. More precisely, on the interval $\varphi\in(-\frac{\pi}{k},x_\theta]$, $\widetilde{e^{i(\varphi+\theta)}}^t=\widetilde{e^{i(\theta-\frac{\pi}{k})}}^t=\omega^{t\floor{\frac{k}{2\pi}\theta}}$ and on its complement $\varphi\in(x_\theta,\frac{\pi}{k}]$, $\widetilde{e^{i(\varphi+\theta)}}^t=\widetilde{e^{i(\theta+\frac{\pi}{k})}}^t=\omega^{t(\floor{\frac{k}{2\pi}\theta}+1)}$. Therefore, $f_{s,t}(\theta)=k\delta_{s,t}\omega^{t\floor{\frac{k}{2\pi}\theta}}\parens*{x_\theta+\frac{\pi}{k}+\omega^t(\frac{\pi}{k}-x_\theta)}$. Putting the terms together,
    \begin{align*}
        \fid_P(\theta)=\frac{1}{2\pi}\sum_{s,t=0}^{k-1}c_{s,t}f_{s,t}(\theta)=\frac{k}{2\pi}\sum_{t=0}^{k-1}c_{t,t}\omega^{t\floor{\frac{k}{2\pi}\theta}}\parens*{\tfrac{2\pi}{k}+(\omega^t-1)(\tfrac{\pi}{k}-x_\theta)}.
    \end{align*}
    To find the Fourier coefficients, consider the Fourier coefficients of $f_{t,t}$. They are
    \begin{align*}
        \hat{f}_{t,t}(0)&=k\int_0^{2\pi}\int_{-\frac{\pi}{k}}^{\frac{\pi}{k}}\widetilde{e^{i(\theta+\varphi)}}^t d\varphi d\theta=2k\pi \delta_{t,0} \int_{-\frac{\pi}{k}}^{\frac{\pi}{k}} d\varphi = 4\pi^2 \delta_{t,0},
    \end{align*}
    and for $n \neq 0$
    \begin{align*}
        \hat{f}_{t,t}(n)&=k\int_0^{2\pi}\int_{-\frac{\pi}{k}}^{\frac{\pi}{k}}\widetilde{e^{i(\theta+\varphi)}}^t e^{-in\theta}d\varphi d\theta\\&=k\int_0^{2\pi}\int_{-\frac{\pi}{k}}^{\frac{\pi}{k}}\widetilde{e^{i\theta}}^t e^{-in(\theta-\varphi)}d\varphi d\theta\\
        &=2k\frac{\sin\parens*{\frac{\pi n}{k}}}{n}\int_{0}^{2\pi} \widetilde{e^{i\theta}}^t e^{-in\theta} d\theta\\
        &=2k\frac{\sin\parens*{\frac{\pi n}{k}}}{n} \int_{-\frac{\pi}{k}}^{\frac{\pi}{k}} e^{-in\theta} d\theta \sum_{r=0}^{k-1}\omega^{r(t-n)} = \begin{cases} 4k^2\frac{\sin^2\parens*{\frac{\pi n}{k}}}{n^2},&\quad t=n\pmod{k},\\
            0,&\quad\text{otherwise.}
        \end{cases}
    \end{align*}
    Plugging these into $\widehat{\fid}_P(n)=\frac{1}{2\pi}\sum_{t=0}^{k-1} c_{t,t}\hat{f}_{t,t}(n)$ gives the wanted result.
\end{proof}

\begin{proposition}\label{prop:cauchy-integral-formula}
    For $\lambda\in\C$ with $|\lambda|\leq 1$ and $P(x,y) = \sum_{s,t} c_{s,t} (x^*)^sy^t$ we have
    \begin{align}
        \int \fid_P(\theta) d\Rel_{\lambda}(\theta)=c_{0,0}+\frac{k^2}{\pi^2}\sum_{n=1}^\infty\frac{\sin^2\parens*{\frac{\pi n}{k}}}{n^2}\parens*{c_{-n,-n}\lambda^n+c_{n,n}(\lambda^\ast)^n}.
    \end{align}
\end{proposition}

\begin{proof}
    By Parseval's theorem (Eq. \ref{eq:parseval-for-distributions}),
    \begin{align*}
        \int \fid_P(\theta) d\Rel_{\lambda}(\theta)=\frac{1}{2\pi}\sum_{n\in\Z}\widehat{\fid}_P(-n)\chi_{\Delta_\lambda}(n).
    \end{align*}
    The characteristic function of the wrapped Cauchy distribution $\Delta_\lambda$ is $\chi_{\Delta_\lambda}(n)=\lambda^n$ for $n\geq 0$ and $\chi_{\Delta_\lambda}(n)=(\lambda^\ast)^{-n}$ for $n<0$. Using this and the previous lemma,
    \begin{align*}
        \int \fid_P(\theta) d\Rel_{\lambda}(\theta)=c_{0,0}+\frac{1}{2\pi}\sum_{n=1}^\infty\frac{2k^2}{\pi n^2}\sin^2\parens*{\frac{\pi n}{k}}\parens*{c_{-n,-n}\lambda^n + c_{n,n}(\lambda^\ast)^n},
    \end{align*}
    giving the result
\end{proof}
\section{Algorithm} \label{sec:algorithm}

In this section, we give an approximation algorithm for noncommutative CSPs, and analyse its performance (prove its guarantees in terms of approximation ratio) on many classes of CSPs in a unified manner. We first give the algorithm and it analysis on $\ncut{k}$ and the larger class of $\nhomlin{2}{k}$ in Section~\ref{sec:algorithm-homogeneous}. In Section~\ref{sec:algorithm-smooth}, we introduce a variant of the algorithm for rounding unitary solutions of smooth CSPs to nearby feasible solutions. The analysis for the smooth variant will be similar to the analysis for $\ncut{k}$ and $\nhomlin{2}{k}$.

\subsection{Homogeneous CSPs}\label{sec:algorithm-homogeneous}
We present the algorithm for homogeneous $\nlin{2}{k}$:

\vspace{10pt}
\IncMargin{1em}
\begin{algorithm}[H]\label{alg:homlink}
\DontPrintSemicolon
Solve the canonical SDP~\eqref{eq:def-homlin2k-sdp} relaxation and obtain $n$ unit vectors $\vec{x}_i \in \R^n$.

Run the vector-to-unitary construction to obtain unitary operators $X_i = U_{\vec{x}_i}$ acting on $\C^M$ for some $M$.

Sample a unitary $U$ acting on $\C^M$ from the Haar measure.

Let $\hat{X}_i \coloneqq UX_i$ for all $i = 1,\ldots,n$.

Let $\tilde{X}_i$ be the \emph{closest} order $k$ unitary to $\hat{X}_i$.

Output the set of operators $\tilde{X}_i$ as the solution. 

\caption{Algorithm constructing a feasible solution for homogeneous $\nlin{2}{k}$.}
\end{algorithm}\DecMargin{1em}
\vspace{10pt}

\begin{itemize}
\item In Step $1$, it is crucial that the SDP relaxation of the problem is real, since these are the only vectors that when fed into the vector-to-unitary construction produce unitary operators. Luckily the canonical SDP relaxation is real for homogeneous CSPs.

\item In Step $2$, it does not matter which vector-to-unitary construction is used. Any isometric mapping from real unit vector to unitaries suffices.

\item Step $5$ can be achieved using Hoffman-Wielandt theorem \cite{hoffman}. This theorem states the following. Suppose $\hat{X}$ is a unitary and that $\sum \lambda_i\ketbra{\phi_i}{\phi_i}$ is its spectral decomposition. Then the closest order-$k$ unitary to $\hat{X}$ is the operator $\tilde{X} = \sum_i \tilde{\lambda}_i\ketbra{\phi_i}{\phi_i}$ where $\tilde{\lambda}_i$ is the closest $k$th root of unity to $\lambda_i$. The closeness is in the sense that $\|\hat{X} - \tilde{X}\| \leq \|\hat{X} - Y\|$ for any order $k$ unitary $Y$, where $\|\cdot\|$ is any unitarily invariant norm such as the Frobenius norm.
\end{itemize}

We now use the method of relative distribution from Section~\ref{sec:relative-distribution} to analyse this algorithm. First, we analyse the algorithm for $\cut{k}$ problems \eqref{eq:def-nmaxcutk}, and then finish this section by extending the analysis to general $\homlin{2}{k}$.

Since the objective function of $\cut{k}$ has the form $\sum_{i,j}w_{i,j}\parens*{1-\frac{1}{k}\sum_{s=0}^{k-1}x_i^{-s}x_j^s}$, we can find the relevant fidelity function (Definition~\ref{def:fidelity}) by considering the $\ast$-polynomial $P_k(x,y)=\frac{1}{k}\sum_{s=0}^{k-1}(x^\ast)^{s}y^s$.

\begin{corollary}[of Lemma \ref{lem:general-fidelity}]\label{cor:cut-fidelity}
    For $P_k(x,y)=\frac{1}{k}\sum_{s=0}^{k-1}(x^\ast)^{s}y^s$, the fidelity function $\fid_k:=\fid_{k,P_k}$ is
    \begin{align}
        \fid_k(\theta)=\begin{cases}1-\frac{k}{2\pi}\theta&0\leq\theta<\frac{2\pi}{k}\\0&\frac{2\pi}{k}\leq\theta<2\pi-\frac{2\pi}{k}\\1-\frac{k}{2\pi}(2\pi-\theta)&2\pi-\frac{2\pi}{k}\leq\theta<2\pi\end{cases}
    \end{align}
\end{corollary}

\begin{proof}
    First, $\fid_k(\theta)=\frac{1}{k}\sum_{t=0}^{k-1}\omega^{t\floor{\frac{k}{2\pi}\theta}}\parens*{1+(\omega^t-1)\parens*{\tfrac{k}{2\pi}\theta-\floor{\tfrac{k}{2\pi}\theta}}}$ by Lemma~\ref{lem:general-fidelity}. If $\floor{\frac{k}{2\pi}\theta}\neq 0,k-1$, the terms all cancel and $\fid_k(\theta)=0$. Else, if $\floor{\frac{k}{2\pi}\theta}=0$, we have $$\fid_k(\theta)=\frac{1}{k}\sum_t\parens*{1+(\omega^t-1)\tfrac{k}{2\pi}\theta}=1-\tfrac{k}{2\pi}\theta.$$ Finally, if $\floor{\frac{k}{2\pi}\theta}=k-1$, we have $$\fid_k(\theta)=\frac{1}{k}\sum_t\parens*{\omega^{-t}+(1-\omega^{-t})\parens*{\tfrac{k}{2\pi}\theta-(k-1)}}=\tfrac{k}{2\pi}\theta-k+1$$
    Putting these three cases together gives the result.
\end{proof}

To make use of the fidelity, we use the integral formula of Theorem~\ref{lem:noncommutative-fid-delta-integral} to get that
$$\expect\parens*{P_k(\tilde{X}_i,\tilde{X}_j)}=\int\fid_k(\theta)d\Delta_{X_i,X_j}(\theta).$$
Due to Theorem~\ref{thm:relative-distribution}, the relative distribution $\Delta_{X_i,X_j}$ tends to a wrapped Cauchy distribution (in distribution). Further, since we can always increase the dimension of the $X_i$ while preserving their $\ast$-moments (\emph{i.e.} inner products of the form $\ip{X_i^s}{X_j^s}$) by tensoring on the identity $X_i \mapsto X_i\otimes I_d$ for larger and larger $d$, we can assume that we are working in the limit of large dimension, and hence replace the relative distribution by the wrapped Cauchy distribution. 

\begin{corollary}[Of Proposition \ref{prop:cauchy-integral-formula}]\label{prop:cut-fidelity-integral}
    For $\lambda$ real we have
    \begin{align}
        \int \fid_k(\theta) d\Rel_{\lambda}(\theta)=\frac{1}{k}+\frac{2k}{\pi^2}\sum_{n=1}^\infty\frac{\sin^2(\pi n/k)}{n^2}\lambda^n=\frac{1}{k}+\frac{k}{\pi^2}\Re\squ*{\mathrm{Li}_2(\lambda)-\mathrm{Li}_2(\lambda e^{i\frac{2\pi}{k}})},
    \end{align}
    where $\mathrm{Li}_2(x)=\sum_{n=1}^\infty\frac{x^n}{n^2}$ is the dilogarithm.
\end{corollary}

\begin{proof}
    With $c_{n,n} = \frac{1}{k}$, for all $n$, in Proposition \ref{prop:cauchy-integral-formula}, the integral is expressed as the power series
    \begin{align*}
        \int \fid_k(\theta) d\Rel_{\lambda}(\theta)=\frac{1}{k}+\frac{k^2}{\pi^2}\sum_{n=1}^\infty\frac{\sin^2\parens*{\frac{\pi n}{k}}}{n^2}\parens*{\frac{1}{k}\lambda^n+\frac{1}{k}(\lambda^\ast)^n}=\frac{1}{k}+\frac{2k}{\pi^2}\sum_{n=1}^\infty\frac{\sin^2\parens*{\frac{\pi n}{k}}}{n^2}\lambda^n.
    \end{align*}
    The closed form in terms of dilogarithms can be found by expanding $\sin^2\parens*{\frac{\pi n}{k}}=\frac{1}{2}\Re\squ*{1-e^{i\frac{2\pi n}{k}}}$ in the series.
\end{proof}

We need the following technical lemma. 
\begin{lemma}\label{lemma:increasing-ratio}
    The function
    \begin{align}
        \lambda\in[-1,1]\mapsto\frac{1-\frac{1}{k}-\frac{k}{\pi^2}\Re\squ*{\mathrm{Li}_2(\lambda)-\mathrm{Li}_2(\lambda e^{i\frac{2\pi}{k}})}}{1-\lambda}
    \end{align}
    is increasing for all $k\geq 3$.
\end{lemma}

The proof of the above is elementary but quite involved, so we give it in Appendix \ref{sec:increasing-ratio}. Now, we can find the approximation ratio of our algorithm for $\ncut{k}$. 

\begin{theorem} \label{thm:maxcut-ratio}
    The approximation ratio of Algorithm \ref{alg:homlink} in the case of $\ncut{k}$ for $k\geq 3$ tends to
    \begin{align}
        \frac{k}{(k-1)(1-\lambda)}\parens*{1 -\frac{1}{k}-\frac{k}{\pi^2}\Re\squ[\Big]{\mathrm{Li}_2\parens[\Big]{-\frac{1}{k-1}}-\mathrm{Li}_2\parens[\Big]{-\frac{e^{i\frac{2\pi}{k}}}{k-1}}}}
    \end{align}
    in the limit of large dimension. 
\end{theorem}

Computing this function gives the values presented in Table~\ref{tab:ratios}. We can alternatively write this as the series
\[\frac{k}{(k-1)(1-\lambda)}\Bigparen{1-\frac{1}{k}-\frac{2k}{\pi^2}\sum_{n=1}^\infty\frac{\sin^2(\pi n/k)}{n^2(1-k)^n}}.\]
Note the rapid convergence of this series due to the exponential term $(1-k)^n$ in the denominator.

\begin{proof}
$\ncut{k}$ is defined in Eq. \ref{eq:def-nmaxcutk} and its canonical SDP relaxation is given in Eq. \ref{eq:def-maxcutk-sdp}.
From here, we use Algorithm \ref{alg:homlink} to get an order-$k$ unitary solution $\{\tilde{X}_i\}$. To calculate the expected value of our solution, by linearity of the expectation, we just need to calculate $\expect \bigparen{\frac{1}{k}\sum_{s=0}^{k-1}\ip{\tilde{X}_i^{-s}}{\tilde{X}_j^s}}$
for every $i,j$. By Theorem~\ref{lem:noncommutative-fid-delta-integral} it holds that
\begin{equation}\label{eq:objective}
\frac{1}{k} \expect \bigparen{\sum_{s=0}^{k-1} \ip{\tilde{X}_i^{-s}}{\tilde{X}_j^s}} = \int \fid_k(\theta) d\Rel_{X_i,X_j}(\theta).
\end{equation}
Therefore the contribution to the solution by edge $(i,j)$ is $w_{ij}\bigparen{1 - \int \fid_k(\theta) d\Rel_{X_i,X_j}(\theta)}$.
So the ratio of the value of the expected solution to the value of the SDP is at least the minimum of
\begin{align*}
\frac{k}{k-1}\frac{1 - \int \fid_k(\theta) d\Rel_{X_i,X_j}(\theta)}{1 - \ip{X_i}{X_j}}
\end{align*}
over all edges $(i,j)$. Note that we can always exclude the edges for which $\ip{X_i}{X_j} = 1$ from this analysis because the contribution of such edges to the SDP value is zero and the numerator is always nonnegative. 

Finally note that in all feasible solutions to SDP we have $\ip{\vec{x}_i}{\vec{x}_j} \geq -\frac{1}{k-1}$. Thus in the worst case over all instances, the quantity above is never less than
\[\min_{\substack{\lambda \in [-\frac{1}{k-1},1)\\X,Y:\ip{X}{Y}=\lambda}} \frac{k}{k-1}\frac{1-\int \fid_k(\theta) d\Rel_{X,Y}(\theta)}{1 - \lambda}\]

In the limit of large dimension, by Theorem \ref{thm:relative-distribution}, the relative distribution tends to the wrapped Cauchy distribution. As such, using Corollary~\ref{prop:cut-fidelity-integral}, the minimization becomes
\begin{align*}
    \min_{\lambda \in [-\frac{1}{k-1},1)}\frac{k}{(k-1)(1-\lambda)}\parens*{1-\frac{1}{k}-\frac{k}{\pi^2}\Re\squ*{\mathrm{Li}_2(\lambda)-\mathrm{Li}_2(\lambda e^{i\frac{2\pi}{k}})}}
\end{align*}
Finally, using Lemma \ref{lemma:increasing-ratio}, the function we are minimizing over is increasing, and as such the minimum is attained at $\lambda=-\frac{1}{k-1}$.
\end{proof}

The proof of the theorem generalizes to all homogeneous CSPs.

\begin{theorem}\label{thm:noncommutative-result} The approximation ratio of Algorithm \ref{alg:homlink} for homogeneous $\nlin{2}{k}$ for $k\geq 3$, in the limit of large dimension, is the minimum of
\[\frac{k}{(k-1)(1-\lambda)}\Bigparen{1 -\frac{1}{k}-\frac{k}{\pi^2}\Re\squ[\Big]{\mathrm{Li}_2\parens[\Big]{-\frac{1}{k-1}}-\mathrm{Li}_2\parens[\Big]{-\frac{e^{i\frac{2\pi}{k}}}{k-1}}}}\]
and 
\[\min_{\lambda \in [-\frac{1}{k-1},1)} \frac{1+\frac{k^2}{\pi^2}\Re\squ*{\mathrm{Li}_2(\lambda)-\mathrm{Li}_2(\lambda e^{i\frac{2\pi}{k}})}}{1 +(k-1)\lambda}\]
\end{theorem}

The proof is the same as for Theorem \ref{thm:maxcut-ratio} except that, in the case of equation constraints, the ratio for a given $\lambda=\ip{X_i}{X_j}$ is
$$\frac{k\int \fid_k(\theta) \Rel_{X_i,X_j}(\theta)d\theta}{1 +(k-1)\lambda}\rightarrow\frac{1+\frac{k^2}{\pi^2}\Re\squ*{\mathrm{Li}_2(\lambda)-\mathrm{Li}_2(\lambda e^{i\frac{2\pi}{k}})}}{1 +(k-1)\lambda},$$
which is not a monotone function, so the minimisation cannot be simplified in the same way as $\ncut{k}$.

This theorem gives the results presented in Table~\ref{tab:ratios-homogeneous}.

\subsection{Smooth CSPs}\label{sec:algorithm-smooth}

We can make use of a similar argument to the previous section to study smooth $\nlin{2}{k}$ defined in Section~\ref{sec:smooth-csps}. However unlike last section, we do not compare the optimal value against the SDP value. Here instead we show that the optimal and unitary values are close. The algorithm for $\nslin{2}{k}$ takes as input a solution to the unitary relaxation~\eqref{eq:ulin2k} and produces a nearby feasible solution for the original instance. 

\vspace{10pt}
\IncMargin{1em}
\begin{algorithm}[H]\label{alg:generic-smooth}
\DontPrintSemicolon

Take as input a solution $X_1,\ldots,X_N \subseteq\mrm{U}(\C^d)$, for some $d \in \N$, to the unitary relaxation~\eqref{eq:ulin2k}.

Sample a unitary $U$ acting on $\C^d$ from the Haar measure.

Let $\hat{X}_i \coloneqq UX_i$ for $i = 1,\ldots,N$.

Let $\tilde{X}_i$ be the \emph{closest} order $k$ unitary to $\hat{X}_i$.

Output the set of operators $\tilde{X}_i$ as the solution. 

\caption{Algorithm constructing a feasible solution for $\nslin{2}{k}$.}
\end{algorithm}\DecMargin{1em}
\vspace{10pt}

Since we are comparing the quality of the rounded solution against the unitary solution we define the \emph{approximation ratio} of this algorithm as $\inf_I \frac{\mathcal{A}(I)}{\mc{U}(I)}$ 
where $I$ ranges over all possible instances of the $\nslin{2}{k}$, $\mc{A}(I)$ is the expected value of the solution produced by the algorithm on $I$ and $\mc{U}(I)$ is the optimal value of the unitary relaxation.

By linearity of expectation, the approximation ratio is the minimum of
\[\frac{\expect\parens[\Big]{ 1 -\frac{2}{a_k} + \frac{2}{a_k}\Re\parens[\Big]{\omega^{-c}\ip{\tilde{X}}{\tilde{Y}}}} }{1 -\frac{2}{a_k} + \frac{2}{a_k}\Re\parens[\Big]{\omega^{-c}\ip{X}{Y}}}\]
(for equation constraints $x^*y = \omega^{c}$) and 
\[\frac{\expect\parens[\Big]{ \frac{2}{a_k} - \frac{2}{a_k}\Re\parens[\Big]{\omega^{-c}\ip{\tilde{X}}{\tilde{Y}}}} }{\frac{2}{a_k} - \frac{2}{a_k}\Re\parens[\Big]{\omega^{-c}\ip{X}{Y}}}\]
(for inequation constraints $x^*y \neq \omega^c$) over all $c\in \Z_k$ and all possible inputs $X,Y$ to the algorithm such that $\lambda = \ip{X}{Y} \in \Omega_k$. Here $\tilde{X}$ and $\tilde{Y}$ are the output of the algorithm. With a change of variable $X \rightarrow \omega^{c} X$, it is enough to consider only the case of $c = 0$. Therefore for the analysis we assume all the $c_{ij}$ in the instance (Eq. \ref{eq:nslin2k}) are zero.

As above, we analyse the algorithm using the relative distribution method of Section~\ref{sec:relative-distribution}. Here, the relevant $\ast$-polynomial is $P_k^S(x,y)=x^\ast y+xy^\ast$ and we let $\fid_k^S \coloneqq \fid_{P_k^S}$.

\begin{corollary}[of Proposition \ref{prop:cauchy-integral-formula}]\label{prop:cauchy-integral-formula-smooth}
    For $\lambda\in\C$ we have
    \begin{align}
    \begin{split}
        \int \fid_k^S(\theta)d\Delta_{\lambda}(\theta)=\frac{2k^2}{\pi^2}\sin^2\parens*{\frac{\pi}{k}}\Re\sum_{\substack{n=\pm 1\Mod{k}\\n>0}}\frac{\lambda^n}{n^2}
    \end{split}
    \end{align}
\end{corollary}

Note that, with a few simple algebraic manipulations, we can also reexpress the integral as
\begin{align}
    \begin{split}
        \int \fid_k^S(\theta)d\Delta_{\lambda}(\theta)&=\frac{2k^2}{\pi^2}\sin^2\parens*{\frac{\pi}{k}}\Re\sum_{m=0}^\infty\parens*{\frac{\lambda^{km+1}}{(km+1)^2}+\frac{\lambda^{km+k-1}}{(km+k-1)^2}}\\
        &=\frac{4k}{\pi^2}\sin^2\parens*{\frac{\pi}{k}}\Re\sum_{s=0}^{k-1}\cos\parens*{\frac{2\pi s}{k}}\mathrm{Li}_2(\lambda\omega^s)
    \end{split}
    \end{align}

\begin{proof}
    With $c_{n,n}=1$ for $n=\pm1\;(\text{mod }k)$, and $c_{n,n}=0$ otherwise, in Proposition~\ref{prop:cauchy-integral-formula},  we have
    \begin{align*}
        \int \fid_k^S(\theta) d\Rel_{\lambda}(\theta)&=\frac{k^2}{\pi^2}\sum_{\substack{n=\pm 1\Mod{k}\\n > 0}}\frac{\sin^2\parens*{\frac{\pi n}{k}}}{n^2}\parens*{\lambda^n+(\lambda^\ast)^n}\\
        &=\frac{2k^2}{\pi^2}\sin^2\parens*{\frac{\pi}{k}}\Re\sum_{\substack{n=\pm 1\Mod{k}\\n>0}}\frac{\lambda^n}{n^2}.
    \end{align*}
\end{proof}

As in the homogeneous case, we can always increase the dimension of the unitary solution by tensoring on identity. Hence, the Cauchy law (Theorem~\ref{thm:relative-distribution}) applies and $\Delta_{X_i,X_j} \rightarrow \Delta_\lambda$ where $\lambda = \ip{X_i}{X_j}$. Also by the definition of unitary relaxation, we can assume $\lambda \in \Omega_k$, where $\Omega_k$ is the simplex of $k$-th roots of unity.

\begin{theorem}\label{thm:smooth-noncommutative-result}
The approximation ratio of our algorithm for  $\nslin{2}{k}$, in the limit of large dimension, is the minimum of
\[\min_{\lambda\in\Omega_k}\frac{1-\frac{k^2}{\pi^2}\sin^2\parens*{\frac{\pi}{k}}\Re\underset{n>0,n=\pm 1\Mod{k}}{\sum}\frac{\lambda^n}{n^2}}{1-\Re\lambda}\]
and 
\[\min_{\lambda\in\Omega_k} \frac{\frac{a_k}{2}-1+\frac{k^2}{\pi^2}\sin^2\parens*{\frac{\pi}{k}}\Re\underset{n>0,n=\pm 1\Mod{k}}{\sum}\frac{\lambda^n}{n^2}}{\frac{a_k}{2}-1+\Re\lambda}\]
\end{theorem}

Computing the minima gives the values presented in Section~\ref{sec:smooth-csps}.

\begin{proof}
    Let $\{X_i\}$ be a solution to the unitary relaxation~\eqref{eq:ulin2k}, and let $\{\tilde{X}_i\}$ be the order-$k$ solution constructed using Algorithm~\ref{alg:generic-smooth}. Due to Theorem~\ref{lem:noncommutative-fid-delta-integral}, the objective values of the inequation and equation constraints are, respectively,
    \begin{align*}
        &\expect\parens[\Big]{\frac{2}{a_k}-\frac{1}{a_k}\tr\squ{P_k^S(\tilde{X}_i, \tilde{X}_j)}}=\frac{2}{a_k}-\frac{1}{a_k}\int\fid_k^S(\theta)\Delta_{X_i,X_j}(\theta)d\theta
    \end{align*}
    and
    \begin{align*}
        &\expect\parens[\Big]{1-\frac{2}{a_k}+\frac{1}{a_k}\tr\squ{P_k^S(\tilde{X}_i, \tilde{X}_j)}}=1-\frac{2}{a_k}+\frac{1}{a_k}\int\fid_k^S(\theta)\Delta_{X_i,X_j}(\theta)d\theta.
    \end{align*}
    As such, the approximation ratio is lower-bounded by the minimum of
    \begin{align*}
        &\min_{X,Y:\ip{X}{Y}\in\Omega_k}\frac{\frac{2}{a_k}-\frac{1}{a_k}\int\fid_k^S(\theta)d\Delta_{X,Y}(\theta)}{\frac{2}{a_k}-\frac{2}{a_k}\Re\ip{X}{Y}}=\min_{X,Y:\ip{X}{Y}\in\Omega_k}\frac{1-\frac{1}{2}\int\fid_k^S(\theta)d\Delta_{X,Y}(\theta)}{1-\Re\ip{X}{Y}}
    \end{align*}
    and
    \begin{align*}
        &\min_{X,Y:\ip{X}{Y}\in\Omega_k}\frac{1-\frac{2}{a_k}+\frac{1}{a_k}\int\fid_k^S(\theta)d\Delta_{X,Y}(\theta)}{1-\frac{2}{a_k}+\frac{2}{a_k}\Re\ip{X}{Y}}=\min_{X,Y:\ip{X}{Y}\in\Omega_k}\frac{\frac{a_k}{2}-1+\frac{1}{2}\int\fid_k^S(\theta)d\Delta_{X,Y}(\theta)}{\frac{a_k}{2}-1+\Re\ip{X}{Y}}.
    \end{align*}
    In the limit of large dimension, the Cauchy law~(Theorem~\ref{thm:relative-distribution}) applies and thus these minimizations only depend on $\lambda=\ip{X}{Y}$. As such, since by Corollary~\ref{prop:cauchy-integral-formula-smooth} we have \[\int \fid_k^S(\theta) d\Rel_{X,Y}(\theta)\rightarrow\frac{2k^2}{\pi^2}\sin^2\parens*{\frac{\pi}{k}}\Re\sum_{n>0,n=\pm 1\Mod{k}}\frac{\lambda^n}{n^2}\] we get the wanted result.
\end{proof}
\section{Algebraic Relative Distribution} \label{sec:algebraic-reldist}

The mixed $\ast$-moments of noncommutative random variables can be used to define new probability distributions. In particular, we consider here an algebraic generalization of the relative distribution (Definition~\ref{def:relative-measure}).

\begin{definition}\label{def:algebraic-relative-distribution}
    Let $\mc{A}$ be a $C^\ast$-algebra with tracial state $\tau$ (noncommutative probability space). Let $u,v\in\mc{A}$ be two unitary elements (noncommutative random variables). The \emph{algebraic relative distribution} of $u$ and $v$ is the probability measure $\Delta_{u,v}$ on the circle $[0,2\pi)$ with characteristic function
    \begin{align}
        \chi_{\Delta_{u,v}}(n)=\tau(u^nv^{-n}).
    \end{align}
\end{definition}

Note that this in fact describes a distribution. This can be seen using the spectral decompositions $u=\int_{0}^{2\pi}e^{i\theta}dP(\theta)$ and $v=\int_{0}^{2\pi}e^{i\varphi}dQ(\varphi)$. Using these decompositions, for measurable sets $E,F\subseteq[0,2\pi)$, we define $\tau(PQ)(E\times F)=\tau(P(E)Q(F))$. This extends to a distribution on $[0,2\pi)\times[0,2\pi)$. Then, the characteristic function $$\chi_{\Delta_{u,v}}(n)=\int e^{in(\theta-\varphi)}d\tau(PQ)(\theta,\varphi)=\chi_{\tau(PQ)}(n,-n),$$ where $\chi_{\tau(PQ)}:\Z\times\Z\rightarrow\C$ is the characteristic function of the distribution $\tau(PQ)$. Due to this, the measure with respect to the algebraic relative distribution of a set $E\subseteq[0,2\pi)$ is
$\Delta_{u,v}(E)=\tau(PQ)(\set*{(\theta,\varphi)}{\varphi-\theta\in E})\geq 0$, from which one can see that $\Delta_{u,v}$ is postive and countably additive, hence a measure.

For a finite-dimensional noncommutative probability space, the algebraic relative distribution is simply the weight distribution (Defintion~\ref{def:weight-measure}). Also, due to the inclusion of Haar-random matrices as noncommutative random variables, the relative distribution (Definition~\ref{def:relative-measure}) can also be recovered as a special case of the algebraic relative distribution.

\subsection{Infinite-Dimensional Case}\label{sec:infinite-dimensional-algebraic-rel-dist}

The first special case of this construction we consider is the case where $u=u_{\vec{x}}$ and $v=u_{\vec{y}}$, unitaries constructed using the infinite-dimensional vector-to-unitary construction of Definition~\ref{def:inf-dim-vector-to-unitary}. In this case $\mc{A}=\mc{M}$ and $\tau=\tau_{\mc{M}}$, so the strong isometry property (Proposition~\ref{prop:power-property}) gives that the characteristic function
$$\chi_{\Delta_{u_{\vec{x}},u_{\vec{y}}}}(n)=\tau_{\mc{M}}(u_{\vec{x}}^nu_{\vec{y}}^{-n})=\ip{\vec{x}}{\vec{y}}^{|n|}.$$
Since the moments fully characterise the distribution, we have the following result.

\begin{theorem}
    Let $\vec{x},\vec{y}\in\R^n$ be unit vectors and let $\lambda=\group{\vec{x},\vec{y}}$. Then, the algebraic relative distribution of $u_{\vec{x}},u_{\vec{y}}\in \mc{M}$ is the wrapped Cauchy distribution $\Delta_{u_{\vec{x}},u_{\vec{y}}}\sim\Delta_\lambda$.
\end{theorem}

That is, the algebraic relative distribution recovers exactly the Cauchy law property of the analytic relative distribution from Section~\ref{sec:relative-distribution-definition}. However, this construction has the flaw of being infinite-dimensional. Nevertheless, we can recover a weaker version of this property in finite dimensions.

\subsection{Finite-Dimensional Case}\label{sec:finite-dimensional-algebraic-rel-dist}

In finite dimensions, we can consider the representation $\pi$ given in Corollary~\ref{cor:fin-dim-vector-to-unitary}. As noted there, this gives rise to a vector-to-unitary construction $U_{\vec{x}}=\sum_ix_i\pi(\sigma_i)$. Let $u=U_{\vec{x}}$ and $v=U_{\vec{y}}$. These are finite-dimensional operators, so their algebraic relative distribution is simply the weight distribution $\Delta_{U_{\vec{x}},U_{\vec{y}}}=w_{U_{\vec{x}},U_{\vec{y}}}$. Nevertheless, the strong isometry property (Corollary~\ref{cor:fin-dim-weak-k-power}) gives a weaker form of the characterisation of the previous section. That is, for $|n|<k$,
$$\chi_{\Delta_{U_{\vec{x}},U_{\vec{y}}}}(n)=\tr\squ*{U_{\vec{x}}^n U_{\vec{y}}^{-n}}=\ip{\vec{x}}{\vec{y}}^{|n|}.$$
This gives immediately the following theorem.
\begin{theorem}\label{thm:finite-alg-rel-dist}
    Let $\vec{x},\vec{y}\in\R^n$ be unit vectors and let $\lambda=\ip{\vec{x}}{\vec{y}}$. Then, the characteristic functions of the algebraic relative distribution $\Delta_{U_{\vec{x}},U_{\vec{y}}}$ and the wrapped Cauchy distribution $\mc{W}(\theta_0,\gamma)$, for scale factor $\gamma=-\ln|\lambda|$ and peak position $\theta_0=0$ if $\lambda\geq 0$, or $\theta_0=\pi$ if $\lambda<0$) are equal up to order $k-1$.
\end{theorem}
\section{Dimension-Efficient Algorithm} \label{sec:efficient-algorithm}

The implementation of the algorithm for homogeneous CSPs introduced in Section~\ref{sec:algorithm-homogeneous} is hindered by the fact that the known approximation ratio is only attained in the limit of large dimension, with no information on the rate of convergence. With the algebraic results of the previous section in mind, consider the following variation:

\vspace{10pt}
\IncMargin{1em}
\begin{algorithm}[H]\label{alg:efficient}
\DontPrintSemicolon

    Solve the basic SDP and obtain $n$ unit vectors $\vec{x}_i \in \R^n$.
    
    Sample a random phase $\zeta\in S^1$.
    
    Run the $m\cdot 2^{(m-1)(n-1)}$-dimensional vector-to-unitary construction for $\mathrm{GWB}_n^m$ given by the representation of Corollary \ref{cor:fin-dim-vector-to-unitary} to obtain unitary operators $X_i = \zeta U_{\vec{x}_i}$.
    
    Let $\tilde{X}_i$ be the \emph{closest} order-$k$ unitary to $X_i$.
    
    Output the set of operators $\tilde{X}_i$ as the solution.

\caption{Dimension-efficient algorithm constructing feasible solution for $\nhomlin{2}{k}$.}
\end{algorithm}\DecMargin{1em}
\vspace{10pt}

We can again analyse this algorithm by means of the relative distribution method introduced in Section~\ref{sec:relative-distribution}. In this case, the relative distribution is defined as follows.

\begin{definition}\label{def:relative-measure-GWB}
    Let $\vec{x},\vec{y}\in\R^n$ and $\zeta$ be a uniformly random phase in $S^1$. Define the \emph{relative distribution} induced by vectors $\vec{x},\vec{y}$ as the measure $\Rel_{\vec{x},\vec{y}}:\scr{B}([0,2\pi))\rightarrow\R_{\geq 0}$ such that
    \begin{align}
        \Rel_{\vec{x},\vec{y}}(E)=\expect(w_{\zeta U_{\vec{x}},\zeta U_{\vec{y}}}(E)).
    \end{align}
    where the weight measure $w$ is defined in Definition \ref{def:weight-measure}.
\end{definition}

First, note that since the distribution of the $\{X_i\}$ is phase-invariant, the result of Theorem~\ref{lem:noncommutative-fid-delta-integral} still holds for this algorithm. On the other hand, since the relative distribution is the distribution over relative eigenvalues, it is independent of the phase. As such, it is exactly the algebraic relative distribution of $U_{\vec{x}}$ and $U_{\vec{y}}$ (Definition~\ref{def:algebraic-relative-distribution}). So, using the result of Theorem~\ref{thm:finite-alg-rel-dist}, we can prove a result analogous to Theorem~\ref{thm:maxcut-ratio} for Algorithm~\ref{alg:efficient}. First, we need a lemma that will allow us to exclude the cases where the approximation given by the low-order coefficients of algebraic relative distribution is not good.

\begin{lemma}
    There exists $\varepsilon>0$ independent of $k$ and $m$ such that for all $\lambda=\ip{\vec{x}}{\vec{y}}\geq 1-\varepsilon$, the approximation ratio
    $$\frac{k}{k-1}\frac{1-\int\fid_k(\theta)\Delta_{\vec{x},\vec{y}}(\theta)d\theta}{1-\lambda}\geq 1,$$
    for large enough $m$.
\end{lemma}

\begin{proof}
    We need to split the proof into a few cases, dependent on $k$. However, all the cases have the same general idea. Since we know $\int e^{il\theta}d\Delta_{\vec{x},\vec{y}}(\theta)=\lambda^{|l|}$ for $|l|< m$, if we can lower bound $1-\fid_k(\theta)\geq\sum_{l=0}^{m-1}a_l\cos(l\theta)$ for all $\theta$, then, we get the lower bound on the ratio
    $$\frac{k}{k-1}\frac{1-\int\fid_k(\theta)\Delta_{\vec{x},\vec{y}}(\theta)d\theta}{1-\lambda}\geq\frac{k}{k-1}\frac{\sum_{l=0}^{m-1}a_l\lambda^l}{1-\lambda}.$$
    As such, by maximising this near $\lambda=1$, we can find a lower bound satisfying the lemma.

    First, consider the case $k\geq 7$. Consider the function $f(\theta)=\frac{1}{2}(1-\cos(3\theta))$. For $k\geq 7$, it is easy to see that $1-\fid_k(\theta)\geq f(\theta)$ for all $\theta$. Then, by the strong isometry property, we know that since $m\geq 3$,
    \begin{align*}
        \frac{1-\int\fid_k(\theta)d\Delta_{\vec{x},\vec{y}}(\theta)}{1-\lambda}\geq\frac{\frac{1}{4}\int2-e^{3i\theta}-e^{-3i\theta}d\Delta_{\vec{x},\vec{y}}(\theta)}{1-\lambda}=\frac{1-\lambda^3}{2(1-\lambda)}=\frac{1}{2}(1+\lambda+\lambda^2).
    \end{align*}
    This is greater than or equal to $1$ for all $\lambda\geq\frac{3}{4}$, so we can take $\varepsilon=\frac{1}{4}$.

    Next, consider $k=5,6$. We can take a similar function $f(\theta)=\frac{1}{2}(1-\cos(2\theta))$, which lower bounds $1-\fid_k(\theta)$ in these cases. This gives
    $$\frac{k}{k-1}\frac{1-\int\fid_k(\theta)\Delta_{\vec{x},\vec{y}}(\theta)d\theta}{1-\lambda}\geq\frac{k}{2(k-1)}\frac{1-\lambda^2}{1-\lambda}\geq\frac{6}{10}(1+\lambda),$$
    which is $\geq 1$ for $\lambda\geq\frac{2}{3}$, so $\varepsilon=\frac{1}{3}$.

    For $k=4$, we can take the lower bound function $f(\theta)=(\sqrt{2}-\frac{3}{4})-(\sqrt{2}-1)\cos(\theta)-\frac{1}{4}\cos(2\theta)$. Then, we have $\frac{k}{k-1}\frac{1-\int\fid_k(\theta)\Delta_{\vec{x},\vec{y}}(\theta)d\theta}{1-\lambda}\geq\frac{4}{3}\parens*{\sqrt{2}-\frac{3}{4}+\frac{1}{4}\lambda}$, which is $\geq 1$ for $\lambda\geq 6-4\sqrt{2}$, so we can take $\varepsilon=\frac{1}{2}$.

    Finally, for $k=3$, we can take the lower bound function $f(\theta)=\frac{3}{5}-\frac{1}{2}\cos(\theta)-\frac{1}{10}\cos(2\theta)$, giving $\frac{k}{k-1}\frac{1-\int\fid_k(\theta)\Delta_{\vec{x},\vec{y}}(\theta)d\theta}{1-\lambda}\geq\frac{3}{2}\parens*{\frac{3}{5}+\frac{1}{10}\lambda}$. Then, the approximation ratio is $\geq 1$ for $\lambda\geq\frac{2}{5}$. Putting all the cases together, we can take $\varepsilon=\frac{1}{4}$ for any $m\geq 4$.
\end{proof}

\begin{theorem} \label{thm:efficient-maxcut-ratio}
    The approximation ratio of Algorithm \ref{alg:efficient} in the case of $\ncut{k}$ for fixed $k\geq 3$ is
    \begin{align}
        1 -\frac{1}{k}-\frac{k}{\pi^2}\Re\squ[\Big]{\mathrm{Li}_2\parens[\Big]{-\frac{1}{k-1}}-\mathrm{Li}_2\parens[\Big]{-\frac{e^{i\frac{2\pi}{k}}}{k-1}}}+O\parens*{\frac{1}{m}}
    \end{align}
\end{theorem}

\begin{proof}
    As in Theorem \ref{thm:maxcut-ratio}, the approximation ratio is the minimum over unit vectors $\vec{x},\vec{y}\in\R^n$ such that $\lambda=\ip{\vec{x}}{\vec{y}}\in[-\tfrac{1}{k-1},1)$ of 
    $$\frac{k}{k-1}\frac{1-\int\fid_k(\theta)\Delta_{\vec{x},\vec{y}}(\theta)d\theta}{1-\lambda}=\frac{k}{k-1}\frac{1-\frac{1}{2\pi}\sum_{n\in\Z}\hat{\fid}_k(n)\chi_{\Delta_{\vec{x},\vec{y}}}(n)}{1-\lambda}.$$
    Note first that, using the previous lemma, for large enough $m$, there exists a neighbourhood of $1$, of radius $\varepsilon$ independent of $m$, where the approximation ratio is greater than $1$. Hence, we need only consider the minimisation of $\lambda$ on the interval $[-\frac{1}{k-1},1-\varepsilon)$. Next, due to the order-$m$ strong isometry property (Corollary \ref{cor:fin-dim-vector-to-unitary}), we have that for $|n|<m$, $\chi_{\Delta_{\vec{x},\vec{y}}}(n)=\lambda^{|n|}$; for the other values of $n$, we have by definition that $|\chi_{\Delta_{\vec{x},\vec{y}}}(n)|\leq 1$ and that $\chi_{\Delta_{\vec{x},\vec{y}}}(n)\in\R$. As such, the difference
    \begin{align*}
        \absm[\Big]{\sum_{n\in\Z}\hat{\fid}_k(n)\chi_{\Delta_{\vec{x},\vec{y}}}(n)-\sum_{n\in\Z}\hat{\fid}_k(n)\lambda^{|n|}}&=\absm[\Big]{2\sum_{n=m}^\infty\Re(\hat{\fid}_k(n))(\chi_{\Delta_{\vec{x},\vec{y}}}(n)-\lambda^{|n|})}\\
        &\leq4\sum_{n=m}^\infty\frac{k}{\pi n^2}\parens*{1-\cos\parens*{\frac{2\pi n}{k}}}\\
        &\leq\frac{8k}{\pi}\sum_{n=m}^\infty\frac{1}{n^2}\leq\frac{8k}{\pi}\int_{m-1}^\infty\frac{1}{x^2}dx=\frac{8k}{\pi(m-1)}.
    \end{align*}
    Hence, due to Corollary~\ref{prop:cut-fidelity-integral}, the ratio for any $\lambda$ is lower-bounded by
    \begin{align*}
        \frac{k}{k-1}\frac{1-\frac{1}{2\pi}\sum_{n\in\Z}\hat{\fid}_k(n)\chi_{\Delta_{\vec{x},\vec{y}}}(n)}{1-\lambda}&\geq\frac{k}{k-1}\frac{1-\frac{1}{k}-\frac{k}{\pi^2}\Re\squ*{\mathrm{Li}_2(\lambda)-\mathrm{Li}_2(\lambda e^{i\frac{2\pi}{k}})}-\frac{4k}{\pi^2(m-1)}}{1-\lambda}\\
        &\geq\frac{k}{k-1}\frac{1-\frac{1}{k}-\frac{k}{\pi^2}\Re\squ*{\mathrm{Li}_2(\lambda)-\mathrm{Li}_2(\lambda e^{i\frac{2\pi}{k}})}}{1-\lambda}-\frac{4k^2}{\pi^2(k-1)(m-1)\varepsilon}.
    \end{align*}
    Since this is simply the ratio in Theorem \ref{thm:maxcut-ratio}, up to additive constant, we can minimise in the same way using Lemma~\ref{lemma:increasing-ratio} to get the wanted result.
\end{proof}

Although we do not analyse it here, it is also possible to extend this algorithm to the all homogeneous $\lin{2}{k}$.

\section{Connections with Previous Work}\label{sec:connection-with-previous-work}

\subsection{Classical CSPs}\label{sec:brave-new-world}
We start with a quick survey of the known results on classical variants of the CSPs studied in this work. H\r{a}stad~\cite{hastadmaxcut} proved that it is NP-hard to approximate $\maxcut$ to a ratio better than $0.94$. In the case of $\lin{2}{3}$ this inapproximability ratio stands at $\approx 0.944$~\cite{hastad}. Along these lines, the famous Unique-Games Conjecture (UGC)~\cite{khot_original,khot,mossel} states that:
\begin{conjecture}[Unique-Games Conjecture]\label{conjecture:ugc} For every $\epsilon > 0$, there exists a large enough $k$, such that it is NP-hard to decide whether, in a given instance of $\ugame{k}$, at least $1 - \epsilon$ or at most $\epsilon$ fraction of total weights can be satisfied. 
\end{conjecture}

As discussed in the introduction the celebrated Goemans-Williamson algorithm~\cite{goemansmaxcut} for $\maxcut$ gave an approximation ratio of $0.878$. This ratio turned out to be the same as the integrality gap of the SDP relaxation~\cite{feige_integrality}. Frieze and Jerrum~\cite{frieze} gave an approximation ratio of $1 - 1/k + 2(\ln k)/k^2$ for $\cut{k}$ in the limit of large $k$. Khot et al.~\cite{khot} showed that, assuming UGC, both these results are tight, meaning that there is no efficient algorithm that does better. Goemans and Williamson~\cite{goemansmax3cut} gave an approximation ratio of $0.836$ for $\cut{3}$ and $0.793$ for $\lin{2}{3}$. In particular these establish lower bounds on the integrality gap of the respective SDP relaxations. Moreover, the algorithm for $\cut{3}$ is tight assuming the ``three candidate plurality is stablest'' conjecture~\cite{khot}.

In comparison, we saw that noncommutative $\maxcut$ can be solved in P~\cite{tsirelson1}. Moreover, Kempe, Regev, and Toner~\cite{kempe} gave an approximation algorithm for noncommutative Unique Games and thereby showed that the analogue of UGC (Conjecture \ref{conjecture:ugc}) does not hold for the noncommutative variant (see Section \ref{sec:entangled-unique-games} for more on this). Along the same lines, our approximation ratios for noncommutative homogeneous $\lin{2}{k}$ are larger than those of the classical ones. 

These are all indicating a peculiar and exciting phenomenon: Whereas general noncommutative CSPs are much harder than their classical counterparts, in the case of $\lin{2}{k}$ the noncommutative problem is easier than the classical variant. What other CSPs behave this way? 
\begin{question}\label{q:which-ncsps-are-approximable}
    Which noncommutative CSPs are approximable? Which are approximable strictly better than their classical counterpart?
\end{question}

The second interesting observation in the study of $\lin{2}{k}$ is the existence of a tight connection between noncommutative solutions and classical solutions. For these CSPs, it seems that it is possible to always extract a good classical solution from a noncommutative one. 

We now propose a rounding scheme for converting solutions of norm $\lin{2}{k}$ to solutions for classical $\lin{2}{k}$. An equivalent form of the norm problem \eqref{eq:nlin2k_norm} is 
\begin{equation}\label{eq:nlin2k_state}
\openup\jot
\begin{aligned}[t]
\text{ maximize:}\quad & \bra{\phi}\parens[\Big]{\sum_{(i,j)\in \mathcal{E}} \frac{w_{ij}}{k} {\sum_{s=0}^{k-1} \omega^{-c_{ij}s} X_i^{-s} X_j^{s}} + \sum_{(i,j)\in \mathcal{I}} w_{ij} \bigparen{1 - \frac{1}{k}{\sum_{s=0}^{k-1} \omega^{-c_{ij}s} X_i^{-s} X_j^{s}}}}\ket{\phi} \\
        \text{subject to:}
                          \quad & \braket{\phi} = 1,\\
                          \quad & X_i^* X_i = X_i^k = 1.
\end{aligned}
\end{equation}
This problem is equivalent to an SDP, and we can find an optimal solution in dimension $O(N)$. Suppose a state $\ket{\phi}$ and order-$k$ unitary operators $X_1,\ldots,X_N$ form a $d$-dimensional feasible solution of an instance of the norm $\lin{2}{k}$ problem. To round this operator solution to a classical solution, sample uniformly at random an order-$k$ unitary $R \in \mathrm{U}_d(\C)$ (this can be done efficiently) and let $$x_i = \omega^{\argmax_{t \in \Z_k} \sum_{s=0}^{k-1} \omega^{-ts} \braket{\phi}{R^{-s}X_i^{s}}{\phi}}.$$ We can show that on average this classical solution achieves a value that is at least $c_k$ times the value of the noncommutative solution for some constant $c_k$. In the case of $\lin{2}{3}$ and $\cut{k}$ this rounding scheme recovers the best-known approximation ratios in \cite{frieze,goemansmaxcut,goemansmax3cut,klerk}. 

Even though we will not present the analysis of this rounding scheme in the current work, we mention it here to illustrate the idea that there is an advantage in studying noncommutative CSPs even if one's goal is a better understanding of the classical CSP landscape. Our rounding presents a unified approach to $\lin{2}{k}$ problems, whereas the known classical rounding schemes treat each problem differently.

In this paper, we focus on the tracial variant of $\lin{2}{k}$ (Problem \eqref{eq:nlin2k}). This gives rise to a natural question. Can we also round tracial solutions to good classical ones? This is further discussed in the section on future directions.

\subsection{Nonlocal Games and CSPs}\label{sec:nonlocal_games}
We phrased CSPs as polynomial optimization problems. \emph{Nonlocal games}, also known as \emph{one-round multiprover interactive proof systems}, provide another equivalent and fruitful way of looking at CSPs. In this section, we discuss the nonlocal games perspective on CSPs. The aim of this discussion is to make it clear that $2$-CSPs and $2$-player nonlocal games are equivalent and that classical (resp. noncommutative) values of CSPs correspond to classical (resp. quantum) values of nonlocal games.

For concreteness suppose we are given an instance of $\lin{2}{k}$ \eqref{eq:lin2k}, although the case of $2$-CSPs can be dealt with similarly. We define the following game played between a referee and two cooperative players Alice and Bob. Let $w$ be the total weights of constraints of the CSP instance. The referee with probability $1/2$ each performs one of the following tasks
\begin{itemize}
    \item Sample $i \in \{1,\ldots,N\}$ uniformly at random and send it to both Alice and Bob.
    \item Sample $i,j \in \{1,\ldots,N\}$ with probability $w_{ij}/w$ and send $i$ to Alice and $j$ to Bob. 
\end{itemize}
The players then each respond with a $k$-th root of unity. The referee decides if the players won according to the following two rules: 
\begin{itemize}
\item If the players received the same variable their answers must be the same (consistency check).
\item If the players received different variables their answers must satisfy the corresponding constraint in the CSP.
\end{itemize}
Crucially Alice and Bob are not allowed to communicate once they receive their questions. Hence they must agree on a \emph{strategy} beforehand that maximizes their chance of winning.

We referred to XOR nonlocal games a few times in the introduction. These are exactly those nonlocal games that are arising from $\lin{2}{2}$ CSPs. 

A \emph{classical strategy} for Alice and Bob is modelled with two functions $$x,y\colon \{1,\ldots,N\} \rightarrow \{1,\omega,\ldots,\omega^{k-1}\}.$$
Operationally, given this strategy Alice (resp. Bob) responds with $x_i$ (resp. $y_i$) when receiving question $i$. We can always write the winning probability of the strategy as a polynomial in $x_i$ and $y_j$. For example if the CSP was $\maxcut$ then the winning probability is given by
\begin{align*}
     \frac{1}{2N} \sum_{i}\frac{1+x_iy_i}{2} + \frac{1}{2w}\sum_{i,j} w_{ij}\frac{1-x_iy_j}{2} 
\end{align*}
The \emph{classical value} of a game is the maximum winning probability over all possible classical strategies, which in our example is the commutative polynomial optimization
\begin{equation}\label{eq:classical-value}
\openup\jot
\begin{aligned}[t]
\text{ maximize:}\quad & \frac{1}{2n} \sum_{i}\frac{1+x_iy_i}{2} + \frac{1}{2w}\sum_{i,j} w_{ij}\frac{1-x_iy_j}{2} \\
        \text{subject to:}\quad & x_i,y_j \in \{\pm 1\}. 
\end{aligned}
\end{equation}
This is similar to the $\maxcut$ CSP (see Eq. \ref{eq:quantummaxcut} and Eq. \ref{eq:maxcut-in-quantum-max-cut-sec} therein) apart from the fact that in \eqref{eq:classical-value} the optimization is over two sets of variables $x_i$ and $y_j$ (\emph{i.e.} strategies of Alice and Bob) where $\maxcut$ is an optimization over one set of variables $x_i$. That said, due to the consistency check, most of the times, it is in the best interest of the players to actually respond with the same strategy $x=y$. Such a strategy, where Alice and Bob respond according to the same function, is called a \emph{synchronous classical strategy}. The \emph{synchronous classical value}, the maximum winning probability over all possible synchronous classical strategies, of the game in our example is thus
\begin{equation}\label{eq:synchronous-classical-value}
\openup\jot
\begin{aligned}[t]
\text{ maximize:}\quad & \frac{1}{2} + \frac{1}{2w}\sum_{i,j} w_{ij}\frac{1-x_ix_j}{2} \\
        \text{subject to:}\quad & x_i \in \{\pm 1\}.
\end{aligned}
\end{equation}
Now except for the unimportant normalisation, this is equivalent to the $\maxcut$ CSP (Figure \ref{fig:maxcut-illustration}).

The players could use quantum entanglement to better correlate their answers. A \emph{quantum strategy} is given by a finite-dimensional Hilbert space $\cH_A\otimes \cH_B$, a shared state $\ket{\phi} \in \cH_A\otimes \cH_B$ and $k$-outcome observables $X_1,\ldots,X_N$ acting on $\cH_A$ for Alice, and $k$-outcome observables $Y_1,\ldots,Y_N$ acting on $\cH_B$ for Bob.\footnote{A $k$-outcome observable is simply an order-$k$ unitary operator. In quantum mechanics these represent $k$-outcome measurements.} Operationally Alice measures her share of the state $\ket{\phi}$ using the observable $X_i$, when she receives question $i$. Bob does the same with his share of the state and observables. The winning probability of this quantum strategy is now the expectation of a noncommutative polynomial in $X_i$ and $Y_j$ with respect to the state $\ket{\phi}$. Let us see this in our example of $\maxcut$. First, by the principle of measurement in quantum mechanics, the probability that Alice and Bob respond with $x_i,y_j \in \{\pm 1\}$ when receiving questions $i,j$, respectively, is \[\bra{\phi}\frac{1+(-1)^{x_i} X_i}{2}\otimes \frac{1+(-1)^{y_j} Y_i}{2}\ket{\phi}.\]
Thus by simple algebra the winning probability overall is
\[\bra{\phi}\Bigparen{\frac{1}{2N} \sum_{i\in [N]}\frac{1+X_i\otimes Y_i}{2} + \frac{1}{2w}\sum_{i,j\in [N]} w_{ij}\frac{1-X_i\otimes Y_j}{2}}\ket{\phi}.\]
The quantum value is defined to be the supremum of the winning probability over all strategies (supremum because we are not bounding the dimension of the Hilbert space). In our example of $\maxcut$ this is given by
\begin{equation}\label{eq:quantum-value}
\openup\jot
\begin{aligned}[t]
\text{ maximize:}\quad & \bra{\phi}\Bigparen{\frac{1}{2N} \sum_{i}\frac{1+X_i\otimes Y_i}{2} + \frac{1}{2w}\sum_{i,j} w_{ij}\frac{1-X_i\otimes Y_j}{2}}\ket{\phi}\\
        \text{subject to:}\quad & \braket{\phi}{\phi} = 1,\\
                                & X_i^*X_i = X_i^2 = 1,\\
                                & Y_j^*Y_j = Y_j^2 = 1.
\end{aligned}
\end{equation}
Finally, a \emph{synchronous quantum strategy} is the special case where Alice and Bob are forced to use fully symmetric strategies (same Hilbert space and observables). The technical definition uses tracial von Neuman algebras. Here we discuss the simpler finite dimensional case. A synchronous quantum strategy is given by a finite-dimensional Hilbert space $\cH\otimes \cH$, the maximally entangled state $\ket{\psi} \in \cH\otimes \cH$ and $k$-outcome observables $X_1,\ldots,X_N$ acting on $\cH$. Operationally, Alice uses $X_i$ and Bob (for technical reasons) uses $X_i^T$ as their observables when receiving question $i$ where $X_i^T$ is the transpose of $X_i$. The \emph{synchronous quantum value} for our example is now
\begin{equation*}
\openup\jot
\begin{aligned}[t]
\text{ maximize:}\quad & \frac{1}{2} + \bra{\psi}\Bigparen{\frac{1}{2w}\sum_{i,j} w_{ij}\frac{1-X_i\otimes X_j^T}{2}}\ket{\psi}\\
        \text{subject to:}\quad & X_i^*X_i = X_i^2 = 1.\\
\end{aligned}
\end{equation*}
Using the well-known identities \[X_i\otimes X_j^T \ket{\psi}= X_iX_j\otimes I \ket{\psi}\] and \[\bra{\psi} X_iX_j\otimes I \ket{\psi} = \tr(X_iX_j) = \ip{X_i}{X_j}\] when $X_i$ and $X_j$ are $2$-outcome observables and $\ket{\psi}$ is the maximally entangled state we can rewrite the synchronous quantum value equivalently as
\begin{equation}\label{eq:quantum-synchronous-value}
\openup\jot
\begin{aligned}[t]
\text{ maximize:}\quad & \frac{1}{2} + \frac{1}{2w}\sum_{i,j} w_{ij}\frac{1-\ip{X_i}{X_j}}{2}\\
        \text{subject to:}\quad & X_i^*X_i = X_i^2 = 1.\\
\end{aligned}
\end{equation}
Except for the unimportant normalisation this is now equivalent to the noncommutative $\maxcut$ \eqref{eq:ncmaxcut-in-intro}.

To summarize, we presented the definitions of four important quantities for a nonlocal game: the classical, synchronous classical, quantum, and synchronous quantum values. The first two are commutative polynomial optimizations and the last two are noncommutative polynomial optimizations. When the instance arises from a $2$-CSP, the synchronous classical and quantum values are the classical and noncommutative values of the CSP instance (up to a simple normalisation).

There remain a few important optimization problems associated to a nonlocal game that we have not defined here: the quantum commuting value and the non signaling value. We refer the reader to \cite{cleve_consequences_and_limits,paulsen2016estimating,helton2017algebras,kim2018synchronous,slofstra_set_of_quantum_correlations,slofstra_tsirelsons_problem_and_an_embedding_theorem,ji_mip_re} for an overview of nonlocal games and their associated optimization problems. These optimization problems relate deep ideas in mathematics, computer science and physics \cite{tsirelson1,tsirelson2,scholz_tsirelsons_problem,scholz2008tsirelson,ozawa_connes_embedding,junge_connes_embedding_problem,fritz_tsirelsons_problem_and_kirchbergs_conjecture,slofstra_set_of_quantum_correlations,slofstra_tsirelsons_problem_and_an_embedding_theorem,ji_mip_re,pi2}.

For some applications of nonlocal games in classical computer science that we have not listed already see \cite{kalai_delegation,kalai_delegation_2,kalai}. For applications in quantum computer science see \cite{colbeck,ruv,dice,device_independent,verifier_leash,kahanamoku_meyer,natarajan_bounding}.

\subsection{Approximation Algorithms for Nonlocal Games}\label{sec:entangled-unique-games}
In this section, we briefly survey known approximation algorithms for the quantum value of nonlocal games. The most notable family of nonlocal games that is efficiently solvable consists of XOR games~\cite{tsirelson1,tsirelson2}, with the CHSH game being the most well-known example. Tsirelson's work on this topic has been a major inspiration for nearly all subsequent studies on nonlocal games.

Yin et al.~\cite{harrow-free} made the insightful observation that in the CHSH game, random and free observables can saturate Tsirelson's bound. To the best of our knowledge, this is the earliest application of free probability in the study of nonlocal games.\footnote{We thank Gal Maor for bringing this to our attention.} Our work extends these ideas to a broader class of nonlocal games.

Another family of nonlocal games for which efficient algorithms exist are those with two questions and two answers per player with an arbitrary number of players. This is due to Masanes~\cite{masanes}. 

In a remarkable paper, Beigi~\cite{beigi2011lower} presented an approximation algorithms for all nonlocal games. By extending the work of Tsirelson on XOR games, and using Weyl-Brauer operators, Beigi gave a $\frac{1}{k-1}$-approximation for any nonlocal game with $k \geq 3$ answers, and a $0.68$-approximation for games with $2$ answers. 

Kempe \emph{et al.}\ \cite{kempe} gave an approximation algorithm for Unique Games. Notably, this shows that the analogue of UGC (Conjecture \ref{conjecture:ugc}) does not hold for the quantum value of Unique Games. Unique Games correspond to a subclass of $\lin{2}{k}$ CSPs which we focused on in this paper. In their work, they proved the following theorem:
\begin{theorem*}[Kempe, Regev, Toner]
There is an approximation algorithm $\mathcal{A}$ such that given an instance $I$ of $\ugame{k}$, if $\sdp(I) \geq 1-\epsilon$, the algorithm $\mathcal{A}$ is guaranteed to output a description of a noncommutative solution with value at least $1-4\epsilon$. Here $\sdp(I)$ denotes the SDP value of $I$ for the canonical SDP relaxation. 
\end{theorem*}
Remarkably, the quality of approximation does not depend on $k$. 

The theorem is stated in their paper for the quantum value of Unique-Games. Their result can be made to work for the synchronous quantum value as well. 

We note that the value of $1-4\epsilon$ in the theorem is achieved in the limit of infinite-dimensional strategies. The proof of this theorem does not extend to all of $\lin{2}{k}$ (the uniqueness property is crucial in their proof). Finally this algorithm does not provide an approximation ratio. In particular the algorithm makes no guarantees when $\epsilon \geq 1/4$.

\subsection{Grothendieck Inequalities}\label{sec:grothendieck}
In this section we explore the connection between CSPs and optimization problems known as \emph{Grothendieck problems}. One aim of this section is to rewrite the various optimization problems we encountered so far as Grothendieck-type problems. Some problems like smooth CSPs naturally fit the setting of Grothendieck problems.

\emph{Grothendieck inequalities} relate these problems with their SDP relaxations. Historically, these inequalities always led to SDP-based approximation algorithms for the Grothendieck problems. When viewed algorithmically, the \emph{constants} in the inequalities capture the notion of integrality gaps and approximation ratios (Definitions \ref{def:approximation-ratio} and \ref{def:integrality-gap}). The second aim of this section is thus to rephrase all our questions about existence of algorithms for CSPs to questions about existence of Grothendieck-type inequalities. In this section, we first give an overview of the commutative Grothendieck inequalities and then we explore their noncommutative analogues. 

For simplicity, all CSPs we consider in this section consist only of inequation constraints. The more general case of a mix of equation and inequation constraints can be dealt with similarly. A~$\lin{2}{k}$ with only inequation constraints looks like
\begin{equation}\label{eq:lin2k-inequation-only}
\openup\jot
\begin{aligned}[t]
\text{ maximize:}\quad &\sum_{(i,j)\in E} w_{ij} \bigparen{1 - \frac{1}{k}{\sum_{s=0}^{k-1} \omega^{-c_{ij}s} x_i^{-s} x_j^{s}}} \\
        \text{subject to:}\quad & x_i^k = 1 \text{ and } x_i \in \C \text{ for all } i \in V,
\end{aligned}
\end{equation}
where $V$ and $E$ are vertex and edge sets of a weighted graph $G=(V,E)$ with edge weights $w:E \to \R_{\geq 0}$ and constants $c:E\to \Z_k$. We always assume $w_{ij} = w_{ji}$ and $c_{ij} = -c_{ji}$. Let $D$ be the weighted degree matrix of $G$, \emph{i.e.} $D_{ii} = \sum_j w_{ij}$ and zero everywhere else. Let $A$ be the weighted adjacency matrix, \emph{i.e.} $A_{ij} = w_{ij}\omega^{-c_{ij}}$ if $(i,j)$ is an edge. The \emph{Laplacian} of this CSP is $D - A$. Note that the Laplacian is always positive semidefinite, as it is Hermitian and diagonally dominant.

\begin{paragraph}{Commutative Grothendieck inequalities.}
The celebrated Grothendieck inequality \cite{grothendieck} concerns two optimization problems. Given $\Gamma = [\gamma_{ij}] \in \R^{N\times N}$, the first problem is the \emph{Grothendieck problem}
\begin{align}\label{eq:real-grothendieck}
\opt(\Gamma) \coloneqq \max_{x_i,y_j \in \{\pm 1\}}\quad \sum_{i,j=1}^N \gamma_{ij} x_i y_j,
\end{align}
and the second problem is the SDP optimization
\begin{align}\label{eq:real-grothendieck-sdp}
\sdp(\Gamma) \coloneqq \max_{\substack{d \in \N\\\vec{x}_i,\vec{y}_j \in \R^d\\ \|\vec{x}_i\|=\|\vec{y}_j\| = 1}}\quad \sum_{i,j=1}^N \gamma_{ij} \ip{\vec{x}_i}{\vec{y}_j}
\end{align}
We clearly have the inequality $\opt(\Gamma) \leq \sdp(\Gamma)$. Grothendieck showed that for all $N$ and for all $\Gamma \in \R^{N\times N}$ we also have 
\begin{equation}\label{eq:real-grothendieck-inequality}
\frac{1}{K_G}\sdp(\Gamma) \leq \opt(\Gamma)
\end{equation}
for some universal constant $K_G < \infty$. The inequality is called the \emph{Grothendieck inequality} and $K_G$ is called the \emph{Grothendieck constant}. Even though we do not know its exact value, we know it is in between $1.6769$ and $1.7823$ \cite{braverman}. 

By definition, $\frac{1}{K_G}$ is the integrality gap of the SDP relaxation for $\opt(\Gamma)$. Alon and Naor \cite{alon_grothendieck} were the first to notice that $\SDP(\Gamma)$ is in fact a semidefinite program, and gave an algorithmic proof of the inequality. Raghavendra and Steurer \cite{raghavendra_steurer} took it a step further showing that the SDP optimal solution can be efficiently rounded to a solution of $\opt(\Gamma)$ with a value at least $\frac{1}{K_G}\sdp(\Gamma)$. They in fact showed that, assuming UGC, the constant $\frac{1}{K_G}$ is the best approximation ratio for $\opt(\Gamma)$, over all possible efficient algorithms (not just those based on the SDP relaxation $\sdp(\Gamma)$). 

This inequality extends naturally to the complex-valued case where $\Gamma \in\C^{N\times N}$ 
\begin{align}\label{eq:complex-grothendieck-problem}
\opt(\Gamma) \coloneqq \max_{\substack{x_i,y_j \in \C\\|x_i|=|y_j|=1}}\quad |\sum_{i,j=1}^N \gamma_{ij} x_i^* y_j|,
\end{align}
and the maximization now ranges over the unit circle. The corresponding SDP problem is 
\begin{align}\label{eq:complex-grothendieck-sdp}
\sdp(\Gamma) \coloneqq \max_{\substack{d\in \N\\\vec{x}_i,\vec{y}_j \in \C^d\\\|\vec{x}_i\|=\|\vec{y}_j\|=1}}\quad |\sum_{i,j=1}^N \gamma_{ij} \ip{\vec{x}_i}{\vec{y}_j}|.
\end{align}
The complex Grothendieck inequality states that there exists a universal constant $K_G^\C < \infty$ such that
\begin{equation}
\frac{1}{K_G^\C}\sdp(\Gamma) \leq \opt(\Gamma).
\end{equation}
It is known that the complex Grothendieck constant is in between $1.338$ and $1.4049$ \cite{davie,haagerup87}. We use the same notation for both the real and complex Grothendieck problems. We let the context and whether $\Gamma$ is real or complex indicate which one is referred.

Many variations of the Grothendieck problem, including noncommutative variants, are known. For an extensive survey see \cite{pisier} and for an incredible number of applications in theoretical computer science see \cite{khot_grothendieck}. 

We present a few of the variants that are most relevant to this paper. We first introduce the \emph{little Grothendieck problem} and see how $\maxcut$ can be phrased in this way. We then explore the link between more general classical CSPs and Grothendieck-type inequalities. Finally, we wrap up this section with the connection between noncommutative CSPs and noncommutative Grothendieck inequalities. 
\end{paragraph}

\begin{paragraph}{Little Grothendieck inequalities.}
The \emph{(real) little Grothendieck problem} \cite{alon_grothendieck} is the special case of the (real) Grothendieck problem where $\Gamma$ is positive semidefinite. In this case one can show that $\opt(\Gamma)$ simplifies to
\begin{align*}
\opt(\Gamma) = \max_{x_i \in \{\pm 1\}}\quad \sum_{i,j} \gamma_{ij} x_i x_j.
\end{align*}
where the optimization is now over one set of variables $x_1,\ldots,x_N$ as opposed to two sets in the original formulation. Grothendieck \cite{grothendieck} proved that for all $\Gamma\geq 0$
\begin{align}\label{eq:real-little-grothendieck}
\frac{2}{\pi}\sdp(\Gamma) \leq \opt(\Gamma).
\end{align} 
Nesterov \cite{nesterov} turned this into a $\frac{2}{\pi}$-approximation algorithm for the little Grothendieck problem, and this is known to be sharp unless $\Pp = \NP$ \cite{iop_tight_hardness}. 

The \emph{complex little Grothendieck problem} for any Hermitian positive semidefinite matrix $\Gamma$ is
\begin{equation}\label{eq:complex-little-grothendieck-problem}
\opt(\Gamma) = \max_{\substack{x_i \in \C\\|x_i|=1}}\quad \sum_{i,j} \gamma_{ij} x_i^* x_j.
\end{equation}
Note that the objective function is always real (as $\Gamma$ is Hermitian). The corresponding SDP is 
\begin{equation}\label{eq:sdp-for-complex-little-grothendieck}
\sdp(\Gamma) = \max_{\substack{d \in \N\\\vec{x}_i \in \C^d\\ \|\vec{x}_i\| = 1}}\quad \sum_{i,j} \gamma_{ij} \ip{\vec{x}_i}{\vec{x}_j}.
\end{equation}
The \emph{complex little Grothendieck inequality} states that for all Hermitian $\Gamma \geq 0$
\begin{equation}\label{eq:complex-little-grothendieck-inequality}
\frac{\pi}{4}\sdp(\Gamma) \leq \opt(\Gamma).
\end{equation}
Just like in the real case, this inequality also leads to a $\frac{\pi}{4}$-approximation algorithm for the complex little Grothendieck problem, and this is also known to be sharp \cite{iop_tight_hardness}. 
\end{paragraph}

\begin{paragraph}{Max-Cut as little Grothendieck problem.}
It should be clear that the classical value of $\maxcut$ as a nonlocal game \eqref{eq:classical-value} is equivalent to a real Grothendieck problem \eqref{eq:real-grothendieck}. As we show in a moment the synchronous value of this game is equivalent to a real little Grothendieck problem.

Recall the classical $\maxcut$ CSP
\begin{equation}\label{eq:maxcut-in-grothendieck-inequality}
\max_{x_i \in \{\pm 1\}}\quad \sum_{(i,j)\in E} \frac{w_{ij}}{2} \paren{1 - x_i x_j}.
\end{equation}
This is clearly the little Grothendieck problem $\opt(\Gamma)$ with $\Gamma$ that is half the Laplacian of the $\maxcut$ instance. Therefore the Goemans and Williamson result on $\maxcut$ implies a Grothendieck-type inequality: for all Laplacian $\Gamma$ we have
\begin{equation}\label{eq:laplacian-real-little-grothendieck}
\rho \sdp(\Gamma) \leq \opt(\Gamma),
\end{equation}
where $\rho = \frac{2}{\pi}\min_{0\leq \theta\leq \pi} \frac{\theta}{1-\cos(\theta)}\approx 0.878$ is the Goemans-Williamson constant. Compare this with \eqref{eq:real-little-grothendieck} and note that $\rho > \frac{2}{\pi}$.
\end{paragraph}

\begin{paragraph}{Classical CSPs as Grothendieck problems.}
The optimization problem $\lin{2}{k}$ \eqref{eq:lin2k} differs from the little Grothendieck problem in two essential ways: the domain of variables in $\lin{2}{k}$ is not $\pm 1$ (unless $k=2$) and the objective function is not a quadratic $*$-polynomial (unless $k \leq 3$). The setting of smooth CSPs, introduced in Section \ref{sec:smooth-csps}, is closer to the setting of Grothendieck problems in that at least the objective function in smooth CSPs is a quadratic polynomial. 

After dropping an unimportant normalisation factor, the problem $\slin{2}{k}$ (with only inequation constraints) becomes
\begin{equation}\label{eq:slin2k-in-grothendieck}
\max_{\substack{x_i \in \C\\x_i^k=1}}\quad \sum_{(i,j)\in E} w_{ij} \paren{1 - \omega^{-c_{ij}}x_i^*x_j}.
\end{equation}
Its unitary relaxation, defined in \ref{sec:smooth-csps}, is
\begin{equation}\label{eq:unitary-slin2k-in-grothendieck}
\max_{\substack{x_i \in \C\\|x_i|=1}}\quad \sum_{(i,j)\in E} w_{ij} \paren{1 - \omega^{-c_{ij}}x_i^*x_j}.
\end{equation}
The unitary relaxation is easily seen to be the complex little Grothendieck problem \eqref{eq:complex-little-grothendieck-problem} where $\Gamma$ is the \emph{CSP Laplacian}. Therefore by \eqref{eq:complex-little-grothendieck-inequality}, there is a $\frac{\pi}{4}$-approximation algorithm for \eqref{eq:unitary-slin2k-in-grothendieck}. 

Could the inequality \eqref{eq:complex-little-grothendieck-inequality} be improved in the special case of CSP Laplacians? This is the case for the real little Grothendieck problem as the constant in \eqref{eq:laplacian-real-little-grothendieck} is larger than the constant in \eqref{eq:real-little-grothendieck}.
\begin{question}\label{q:complex-grothendieck-for-csp-laplacian}
What is the largest real number $\rho^\C$ such that $\rho^\C\sdp(\Gamma) \leq \opt(\Gamma)$ when $\Gamma$ is a CSP Laplacian? Is $\rho^\C$ strictly larger than $\pi/4$? This question relates to Question \ref{q:classical-smooth-csp} we asked earlier.
\end{question}

Now it is natural to phrase $\slin{2}{k}$ (as opposed to its unitary relaxation) and even the non-quadratic variant $\lin{2}{k}$ as a Grothendieck-type problem. Here the variables take values in the discrete set $x_i \in \{1,\omega,\ldots,\omega^{k-1}\}$ as opposed to the unit circle. Let $W=[w_{ij}] \in (\R_{\geq 0})^{N\times N}$ be any symmetric matrix of nonnegative reals and let $C = [c_{ij}]\in \Z_k^{N\times N}$ be any antisymmetric matrix. Define two optimization problems
\begin{align*}
\sopt_k(C,W) &\coloneqq \max_{\substack{x_i\in \C\\x_i^k = 1}}\quad \sum_{(i,j)\in E} w_{ij}(1 - \omega^{-c_{ij}} x_i^*x_j),
\end{align*}
and
\begin{align*}
\opt_k(C,W) &\coloneqq \max_{\substack{x_i\in \C\\x_i^k = 1}}\quad \sum_{(i,j) \in E} w_{ij}(1 - \frac{1}{k}\sum_{s=0}^{k-1}\omega^{-c_{ij}s} x_i^{-s}x_j^s).
\end{align*}
The first one captures $\slin{2}{k}$ and the second captures $\lin{2}{k}$ \eqref{eq:lin2k}. 

Next we need to define SDP relaxations for these two problems. However, here we face at least one complication. The quadratic terms $x_i^*x_j$ in the objective function naturally turn into inner products $\ip{\vec{x}_i}{\vec{x}_j}$ in the conventional relaxations (see for example Eq. \ref{eq:complex-grothendieck-sdp}). However, it is not clear how to do this with the higher degree monomials $x_i^{-s}x_j^s$ in the second problem. We circumvent this issue by using operator relaxations as opposed to vector relaxations.\footnote{This has been a recurrent theme in this paper, see for example Section \ref{sec:brave-new-world}.} Consider the following pair of optimization problems
\begin{align}
\ssdp_k(C,W) &\coloneqq \max_{X_i^*X_i = X_i^k = 1}\quad \norm[\Big]{\sum_{(i,j)\in E} w_{ij}(1 - \omega^{-c_{ij}} X_i^*X_j)}_{\mathrm{op}}\label{eq:canonical-sdp-smooth}
\end{align}
and
\begin{align}
\sdp_k(C,W) &\coloneqq \max_{X_i^*X_i = X_i^k = 1}\quad \norm[\Big]{\sum_{(i,j) \in E} w_{ij}(1 - \frac{1}{k}\sum_{s=0}^{k-1}\omega^{-c_{ij}s} X_i^{-s}X_j^s)}_{\mathrm{op}}\label{eq:canonical-sdp-general}
\end{align}
where the optimization is over all finite-dimensional Hilbert spaces and all order-$k$ unitaries acting on them. Note that the second problem is the norm $\lin{2}{k}$ problem introduced earlier in Eq.\ \ref{eq:nlin2k_norm}. Both $\ssdp_k$ and $\sdp_k$ are indeed SDPs, and they are relaxations of $\sopt_k$ and $\opt_k$, respectively. 
\begin{question}[Order-$k$ little Grothendieck inequalities]\label{q:grothendieck-constant-for-order-k-problems}
What are the largest real numbers $\rho_{k}^{S}$ and $\rho_k$ such that for all integers $N$ and matrices $C$ and $W$ as above
\begin{align}
\rho_{k}^{S} \ssdp_k(C,W) &\leq \sopt_k(C,W),\label{eq:smooth-order-k-grothendieck-inequality}\\
\rho_k \sdp_k(C,W) &\leq \opt_k(C,W).\label{eq:order-k-grothendieck-inequality}
\end{align}
\end{question}
The first inequality \eqref{eq:smooth-order-k-grothendieck-inequality} would provide an extension of Grothendieck inequalities to nonbinary domains. The second inequality \eqref{eq:order-k-grothendieck-inequality} would provide an extension of Grothendieck inequalities to nonbinary domains and non-quadratic objective functions. Katzelnick and Schwartz~\cite{katzelnick} were the first to study Grothendieck-type inequalities for larger domains and they use their inequalities to derive approximation algorithms for the Correlation-Clustering CSP. Their formulation is different from ours. 

In Section \ref{sec:brave-new-world}, we gave an efficient rounding scheme for turning solutions of $\sdp_k(C,W)$ to solutions of $\opt_k(C,W)$. This suggests a strategy for proving these inequalities algorithmically.
\end{paragraph}

\begin{paragraph}{Noncommutative case.}
Grothendieck \cite{grothendieck} also conjectured a noncommutative analogue of his namesake inequality which since then has been proven \cite{pisier-noncommutative}. For applications of this inequality in quantum information see \cite{briet-thesis,naor-efficient-rounding,regev-vidick-xor}. Here we present the weaker noncommutative little Grothendieck problem of the form studied by Bandeira \emph{et al.}~\cite{bandeira}.\footnote{Bandeira \emph{et al.}'s version is itself a weaker version of what is known as the noncommutative little Grothendieck problem, see for example \cite{iop_tight_hardness}.} The \emph{noncommutative little Grothendieck problem of dimension $d$} is the optimization problem
\begin{equation}\label{eq:noncommutative-little-grothendieck-problem}
\max_{X_i \in \mathrm{U}_d}\quad \sum_{(i,j)\in E} \gamma_{ij} \ip{X_i}{X_j}.
\end{equation}
where $\Gamma = [\gamma_{ij}] \in \C^{N\times N}$ is positive semidefinite as in the case of the commutative little Grothendieck problem, and $\mathrm{U}_d$ denotes the set of $d$-dimensional unitary matrices. 

Recall $\ulin{2}{k}$, the unitary relaxation of $\nslin{2}{k}$, from Section \ref{sec:smooth-csps}, which is equivalent to
\begin{equation}\label{eq:unitary-nc-slin2k-in-grothendieck}
\max_{\substack{d\in \N\\X_i\in \mathrm{U}_d(\C)\\\ip{X_i}{X_j}\in \Omega_k}}\quad \sum_{(i,j)\in E} w_{ij} \paren{1 - \omega^{-c_{ij}} \ip{X_i}{X_j}}.
\end{equation}
Just like in the commutative case, using the CSP Laplacian, this problem can be written in the form of a noncommutative little Grothendieck problem \eqref{eq:noncommutative-little-grothendieck-problem}. However, the principal difference is that, whereas \eqref{eq:noncommutative-little-grothendieck-problem} is dimension-bounded, the optimization in $\ulin{2}{k}$ is over unbounded dimension. 

In all the noncommutative CSPs we introduced in this paper (as well as in the study of nonlocal games), one does not restrict the solution (or quantum strategy in the case of nonlocal games) to bounded dimension. However, noncommutative Grothendieck problems are usually defined over bounded dimension. So, in the noncommutative world, in addition to large domain and non-quadratic objective functions there is a third difference between our setting and that of Grothendieck inequalities: dimension. Nevertheless we can still express our noncommutative CSPs as Grothendieck-type problems.

Inspired by $\ulin{2}{k}$, we define the \emph{unbounded-dimensional little Grothendieck problem} as
\begin{equation}\label{eq:unbounded-noncommutative-little-grothendieck-problem}
\nopt(\Gamma) \coloneqq \max_{\substack{d\in \N\\X_i \in \mathrm{U}_d(\C)}}\quad \sum_{(i,j)\in E} \gamma_{ij} \ip{X_i}{X_j},
\end{equation}
for Hermitian and positive semidefinite $\Gamma$. We let the SDP relaxation $\sdp(\Gamma)$ be the same as the one in the complex little Grothendieck problem \eqref{eq:sdp-for-complex-little-grothendieck}. We copy it here for convenience
\begin{equation*}
\sdp(\Gamma) = \max_{\substack{d \in \N\\\vec{x}_i \in \C^d\\ \|\vec{x}_i\| = 1}}\quad \sum_{i,j} \gamma_{ij} \ip{\vec{x}_i}{\vec{x}_j}.
\end{equation*}

\begin{question}[Unbounded-dimensional little Grothendieck inequality]\label{q:unbounded-little-grothendieck-inequality} What is the largest real number $\eta$ such that for all Hermitian and positive semidefinite $\Gamma$
\begin{align*}
\eta \sdp(\Gamma) \leq \nopt(\Gamma).
\end{align*}
We could ask the same when restricting $\Gamma$ to the subset of CSP Laplacians. The restriction to Laplacians captures $\ulin{2}{k}$.\footnote{To be precise, to capture $\ulin{2}{k}$ one needs to add the additional constraint that inner products are in the simplex $\Omega_k$ when the CSP is defined over $k$-ary variables, \emph{i.e.}\  \[\nopt(\Gamma) = \max_{\substack{d\in \N\\X_i \in \mathrm{U}_d(\C)\\\ip{X_i}{X_j}\in \Omega_k}}\quad \sum_{(i,j)\in E} \gamma_{ij} \ip{X_i}{X_j}\]}
\end{question}
A remarkable property of unbounded-dimensional little Grothendieck problem is that if $\Gamma$ is any real symmetric matrix $\sdp(\Gamma) = \nopt(\Gamma)$. This follows from the vector-to-unitary construction and Tsirelson's theorem. Therefore the real unbounded-dimensional little Grothendieck inequality becomes an equality.

Finally let us phrase $\nslin{2}{k}$ and $\nlin{2}{k}$ as Grothendieck-type problems. For every $C$ and $W$ as in the set-up above, define two optimization problems
\begin{align}
\nsopt_k(C,W) &\coloneqq \max_{X_i^*X_i = X_i^k = 1}\quad \sum_{(i,j)\in E} w_{ij}(1 - \omega^{-c_{ij}} \ip{X_i}{X_j})\label{eq:nc-smooth-in-grothendieck}\\
\nopt_k(C,W) &\coloneqq \max_{X_i^*X_i = X_i^k = 1}\quad \sum_{(i,j) \in E} w_{ij}(1 - \frac{1}{k}\sum_{s=0}^{k-1}\omega^{-c_{ij}s} \ip{X_i^{s}}{X_j^s})\label{eq:nc-in-grothendieck}
\end{align}
where the optimization is over order-$k$ unitaries of any dimension. The first problem captures $\nslin{2}{k}$ and the second captures $\nlin{2}{k}$. It is easy to see that $\sdp_k(C,W)$ is also an SDP relaxation of $\nopt_k(C,W)$. This is also true of $\ssdp_k(C,W)$ and $\nsopt_k(C,W)$. 

\begin{question}[Order-$k$ unbounded-dimensional little Grothendieck inequalities]\label{q:order-k-unbounded-little-grothendieck}
What are the largest real numbers $\eta_{k}^{S}$ and $\eta_k$ such that for all integers $N$ and matrices $C$ and $W$ as above
\begin{align}
\eta_{k}^{S} \ssdp_k(C,W) &\leq \nsopt_k(C,W),\label{eq:nc-smooth-order-k-grothendieck-inequality}\\
\eta_k \sdp_k(C,W) &\leq \nopt_k(C,W).\label{eq:nc-order-k-grothendieck-inequality}
\end{align}
\end{question}

\end{paragraph}

\subsection{Quantum Max-Cut}\label{sec:quantum-max-cut}
The famous \emph{local Hamiltonian} problem is another physically-motivated generalization of CSPs to the noncommutative setting. Being the most natural $\QMA$-complete problem \cite{kitaev}, this problem has been studied extensively in the literature. See the survey \cite{quantum-hamiltonian-complexity} on quantum Hamiltonian complexity for an introduction to this area. 

The goal of this section is to compare this generalization of CSPs with the setting of noncommutative CSPs we study in this paper. To help with this comparison, we give the definition of quantum $\maxcut$, a $\QMA$-complete special case of the local Hamiltonian problem \cite{montanaro_quantum_max_cut,sevag}. See \cite{sevag,anshu,parekh,king,watts_quantum_max_cut} for some of the recent progress on approximation algorithms (often SDP-based) for this problem.

We begin by rephrasing the familiar classical $\maxcut$ in a language that is closer to the setting of quantum $\maxcut$. First, recall that $\maxcut$ is
\begin{equation}\label{eq:maxcut-in-quantum-max-cut-sec}
\openup\jot
\begin{aligned}[t]
\text{ maximize:}\quad &\sum_{(i,j) \in E} \frac{w_{ij}}{2} (1-x_ix_j) \\
        \text{subject to:}\quad & x_i^2 = 1 \text{ and } x_i \in \mathbb{R},
\end{aligned}
\end{equation}
An equivalent restatement of \eqref{eq:maxcut-in-quantum-max-cut-sec} using Pauli matrices $\sigma_z$ is
\begin{equation}\label{eq:maxcut-paulis}
\openup\jot
\begin{aligned}[t]
\text{ maximize:}\quad &\bra{\phi}\Bigparen{\sum_{(i,j) \in E} \frac{w_{ij}}{2} (1-\sigma_{z,i}\sigma_{z,j})}\ket{\phi} \\
        \text{subject to:}\quad & \ket{\phi} \in (\C^{2})^{\otimes N}, 
            \braket{\phi}{\phi} = 1,
\end{aligned}
\end{equation}
where the optimization is ranging over all $N$-qubit states $\ket{\phi}$
and where $\sigma_{z,i}$ is our notation for the operator \[\underbrace{\overbrace{I\otimes \cdots \otimes I \otimes \sigma_z}^{i} \otimes I \otimes \cdots \otimes I}_{N},\] that acts as identity everywhere except on the $i$-th qubit. Given an assignment $x_1,\ldots,x_N$ in \eqref{eq:maxcut-in-quantum-max-cut-sec}, the state $\ket{\phi} = \ket{x_1}\otimes\cdots\otimes\ket{x_N}$ achieves the same objective value in $\eqref{eq:maxcut-paulis}$. The converse can also be verified, and thus the two problems are equivalent.

Quantum $\maxcut$ is a generalization where one takes into consideration the other two Pauli matrices as well
\begin{equation}\label{eq:quantummaxcut}
\openup\jot
\begin{aligned}[t]
\text{ maximize:}\quad &\bra{\phi}\Bigparen{\sum_{(i,j) \in E} \frac{w_{ij}}{2} (1-\sigma_{x,i}\sigma_{x,j}-\sigma_{y,i}\sigma_{y,j}-\sigma_{z,i}\sigma_{z,j})}\ket{\phi} \\
        \text{subject to:}\quad & \ket{\phi} \in (\C^{2})^{\otimes N}, 
            \braket{\phi}{\phi} = 1.
\end{aligned}
\end{equation}
Whereas in quantum $\maxcut$ the operators are fixed and one is optimizing over states in a Hilbert space of dimension $2^N$, in noncommutative $\maxcut$ as we study in this paper we have
\begin{equation}\label{eq:nc-maxcut-in-quantum-max-cut-sec}
\openup\jot
\begin{aligned}[t]
\text{ maximize:}\quad &\tr\Bigparen{\sum_{i,j} w_{ij} \frac{1-X_iX_j}{2}} \\
        \text{subject to:}\quad & X_i^*X_i = X_i^2 = 1.
\end{aligned}
\end{equation}
where one optimizes over operators (of unbounded dimension). If in noncommutative $\maxcut$, one were also to optimize over the state (so that one is not limited to the tracial state) in addition to the operators, one recovers the norm $\maxcut$ problem 
\begin{equation}\label{eq:nmaxcut_state}
\openup\jot
\begin{aligned}[t]
\text{ maximize:}\quad &\bra{\phi}\Bigparen{\sum_{i,j} w_{ij} \frac{1-X_iX_j}{2}}\ket{\phi} \\
        \text{subject to:}\quad & X_i^*X_i = X_i^2 = 1,\\
                                & \braket{\phi}{\phi} = 1.
\end{aligned}
\end{equation}
Deriving this from the definition of norm variant \eqref{eq:nlin2k_norm} is straightforward. As mentioned earlier this variant is an efficiently-solvable SDP.

\appendix
\section{Proofs from Section \ref{sec:construct-gwb}}\label{sec:representation-calculations}

\begin{lemma}[Restatement of Lemma \ref{lem:claim1}] Every element of $G_{n,f}^k$ is of the form $J^{b}c^sp^\alpha$ for some $(b,s,\alpha) \in \Z_2\times\Z_k\times\mathrm{M}_{[n-1]\times\Z_k}(\Z_2)$.
\end{lemma}
In the following lemma, we will show these are the normal form words of $G_{n,f}^k$
\begin{proof}
Elements of $G_{n,f}^k$ are arbitrary products of terms of the form $c^{-t}p_ic^{t}$ and powers of $J$ and $c$. Since $J$ commutes with everything we can always bring terms that are powers of $J$ to the beginning. Similarly, powers of $c$ can always be moved to the beginning because $(c^{-t}p_i c^t)c^s = c^s(c^{-(t+s)}p_i c^{t+s})$.
Finally, using the commutation relations of $G_{n,f}^k$, we have $[c^{-s} p_i c^{s},c^{-t} p_j c^{t}] = J^{f_{i,j,t-s}}$, and therefore elements of the form $c^{-t} p_i c^{t}$ for all $i$ and $t$, can always be moved past each other, which changes the exponent of $J$.
\end{proof}

\begin{lemma}[Restatement of Lemma \ref{lem:unique-normal-form}]
    $|G_{n,f}^k| = 2k\cdot 2^{k(n-1)}$. In particular every element in $G_{n,f}^k$ has a unique normal form.
\end{lemma}
\begin{proof}
    Let $X$ be the set of all $2k\cdot 2^{k(n-1)}$ normal form words $J^bc^sp^\alpha$. We let $\pi$ be the action of the group $G_{n,f}^k$ on the set $X$ defined on the generators so that
    \begin{align*}
        &\pi_J(J^{b}c^sp^\alpha)=J^{b+1}c^sp^\alpha,\\
        &\pi_c(J^{b}c^sp^\alpha)=J^{b}c^{s+1}p^\alpha,\\
        &\pi_{p_i}(J^{b}c^sp^\alpha)=J^{b+\alpha_{i,s}+\sum_{(j,t)\leq(i,s)}\alpha_{j,t}f_{i,j,t-s}}c^sp^{\alpha+E_{i,s}},
    \end{align*}
    where $E_{i,s}$ is the standard basis matrix that is $1$ at the entry $(i,s)$ and $0$ everywhere else. To show that $\pi$ indeed extends to an action of the group $G_{n,f}^k$ on $X$ we need to show that every relator in the presentation acts as the identity map on $X$.
    \begin{claim*} For every relator $r$ in the presentation of $G_{n,f}^k$, the map $\pi_r$ is the identity map on $X$.
    \end{claim*}
    \begin{proof}
    We can directly verify that 
    \begin{align*}
    &\pi_{J^2}(J^{b}c^sp^\alpha)=(J^{b}c^sp^\alpha)\\
    &\pi_{c^k}(J^{b}c^sp^\alpha)=(J^{b}c^sp^\alpha)\\
    &\pi_{p_iJ}(J^{b}c^sp^\alpha)=\pi_{Jp_i}(J^{b}c^sp^\alpha)\\
    &\pi_{cJ}(J^{b}c^sp^\alpha)=\pi_{Jc}(J^{b}c^sp^\alpha)
    \end{align*}
    Similarly we have \[\pi_{Jp_i^2}(J^{b}c^sp^\alpha)=\pi_J(J^{b + 1}c^sp^\alpha) =J^b c^sp^\alpha.\] Finally, a tedious but straightforward calculation verifies that
    $$\pi_{[p_i,c^{-t}p_jc^t]}(J^{b}c^sp^\alpha)=J^{b+\delta_{(i,s)\leq(j,s+t)}f_{i,j,t}+\delta_{(j,s+t)\leq(i,s)}f_{j,i,-t}}c^sp^\alpha,$$
    which is $J^{b+f_{i,j,t}}c^sp^\alpha$, as desired, owing to our assumptions on the coefficients $f_{i,j,t}$.
    \end{proof}
    Let $\epsilon = J^0 c^0 p^0 \in X$ denote the empty word and note that $J^bc^sp^\alpha=\pi_{J^bc^sp^\alpha}(\epsilon)$ are all distinct words in $X$. Therefore distinct normal form words in $X$ are also distinct as elements of $G_{n,f}^k$. We conclude that  $|G_{n,f}^k| = |X|$.
\end{proof}

\begin{lemma}[Restatement of Lemma \ref{lem:centrality}]
$H$ is a central subgroup, and thus is trivially normal in $G_{n,f}^k$.
\end{lemma}
 \begin{proof}   
    Using the commutation relations, we can write
    \begin{align*}
        &Jr_iJ^{-1}=r_i\\
        &c^{-1}r_ic=J^{k}(c^{-1}p_ic)\cdots (c^{-(k-1)}p_ic^{k-1})p_i=J^3p_i(c^{-1}p_ic)\cdots (c^{-(k-1)}p_ic^{k-1})=r_i\\
        &p_jr_ip_j^{-1}=J^{k + \delta_{i\neq j} + \delta_{i\leq j} +\delta_{j\leq i}}p_i(c^{-1}p_ic)\cdots (c^{-(k-1)}p_ic^{k-1})=r_i.
    \end{align*}
\end{proof}
\begin{theorem}[Restatement of Theorem \ref{thm:gnk-order}]
    The group $\gwb_n^k$ is finite of order $2k\cdot 2^{(k-1)(n-1)}$. Further, the words $p^{\alpha}$ for $\alpha\in\mathrm{M}_{[n-1]\times\Z_k}(\Z_2)$ with $\alpha_{i,k-1}=0$ are equal to $e$ or $J$ if and only if $\alpha=0$.
\end{theorem}
\begin{proof}
    Since $\gwb_n^k=G_{n,f}^k/H$ we have $|\gwb_n^k|=|G_{n,f}^k|/|H|$. We calculate $|H|$. First note that from the commutation relations of $G_{n,f}^k$, we have
    \begin{align*}
    [c^{-s}p_ic^s,c^{-t}p_i c^t] &= c^{-s}[p_i,c^{-(t-s)}p_ic^{t-s}]c^s=\begin{cases} J &\quad\text{if } t-s = 1,k-1 \Mod{k},\\
    e &\quad\text{otherwise.}
    \end{cases}
    \end{align*}
    Now a repeated application of these commutation relations gives that $r_i^2 = 1$ for all $i \in [n-1]$.
    Therefore $|H| \leq 2^{n-1}$ as it is generated by $n-1$ commuting generators of order $2$.
    However by Lemma \eqref{lem:unique-normal-form} we know that all the normal form words are distinct in $G_{n,f}^k$. Thus $\Pi_{i=1}^{n-1} r_i^{e_i}$ for $e_i\in \{0,1\}$ are all distinct, and $|H| = 2^{n-1}$. This gives the result $|\gwb_n^k| = |G_{n,f}^k|/|H|=2k\cdot 2^{k(n-1)}/2^{n-1}=2k\cdot 2^{(k-1)(n-1)}$.
    
    For the second statement, none of the $p^\alpha$ considered are contained in $H$ or $JH$, so they are not equal to $e$ or $J$ in the quotient group $\gwb_n^k$.
\end{proof}

\section{Proof of Lemma \ref{lemma:increasing-ratio}}\label{sec:increasing-ratio}

\begin{lemma}\label{lem:convex-condition}
    Let $f:[-1,1]\rightarrow\R$ be a smooth function. If $f(1)\leq 1$ and $f$ is convex, then $g(x):=\frac{1-f(x)}{1-x}$ is increasing.
\end{lemma}

\begin{proof}
    Since $f$ is smooth and convex, $f''(x)\geq 0$ for all $x$. Therefore, $h(x):=1-f(x)-(1-x)f'(x)$ is decreasing, as its derivative $h'(x)=-f'(x)+f'(x)-(1-x)f''(x)=-(1-x)f''(x)\leq 0$. Also, we have that $h(1)=1-f(1)\geq 0$. Putting these together, $h(x)\geq 0$ for all $x\in[-1,1]$. The derivative of $g$ is $g'(x)=\frac{-f'(x)(1-x)+(1-f(x))}{(1-x)^2}=\frac{h(x)}{(1-x)^2}\geq 0$, and as such $g$ is increasing. 
\end{proof}

\begin{lemma}\label{lem:positive-endpoint}
    Let $f:[-1,1]\rightarrow\R$ be $f(x)=\frac{1}{k}+\frac{k}{\pi^2}\Re\squ*{\mathrm{Li}_2(x)-\mathrm{Li}_2(x e^{i\frac{2\pi}{k}})}$ for $k\geq 3$. Then $f(1)=1$.
\end{lemma}

\begin{proof}
    Consider the Fourier series of the function $p(x)=x(1-x)$ on the interval $[0,1]$. The Fourier coefficients are
    \begin{align*}
        \hat{p}(n)=\int_0^1x(1-x)e^{-2\pi i n x}dx=-\frac{1}{2\pi n^2},
    \end{align*}
    for $n\neq 0$, and $\hat{p}(0)=\frac{1}{6}$. Thus, $x(1-x)=\frac{1}{6}-\frac{1}{\pi^2}\sum_{n=1}^{\infty}\frac{\cos(2\pi n x)}{n^2}$. In particular, as $\frac{1}{k}\in[0,1]$, $-\sum_{n=1}^{\infty}\frac{\cos(2\pi n/k)}{n^2}=\pi^2\parens*{\frac{1}{k}\parens*{1-\frac{1}{k}}-\frac{1}{6}}$. Using this and the fact that $\sum_{n=1}^{\infty}\frac{1}{n^2}=\frac{\pi^2}{6}$,
    \begin{align*}
        f(1)=\frac{1}{k}+\frac{k}{\pi^2}\sum_{n=1}^{\infty}\frac{1-\cos(2\pi n/k)}{n^2}=\frac{1}{k}+\frac{k}{\pi^2}\squ*{\frac{\pi^2}{6}+\frac{\pi^2}{k}\parens*{1-\frac{1}{k}}-\frac{\pi^2}{6}}=1.
    \end{align*}
\end{proof}

\begin{lemma}\label{lem:convex-function}
    Let $f:[-1,1]\rightarrow\R$ be as in the previous lemma. Then, $f$ is convex.
\end{lemma}

\begin{proof}
    It suffices to show that the second derivative of $f$ is positive. Since the dilogarithm has integral expression $\mathrm{Li}_2(x)=-\int_0^x\frac{\ln(1-t)}{t}dt$, the first derivative of $f$ is
    \begin{align*}
        f'(x)=\frac{k}{\pi^2}\Re\squ*{\frac{-\ln(1-x)+\ln(1-xe^{i\frac{2\pi}{k}})}{x}}=\frac{k}{2\pi^2}\frac{1}{x}\ln\parens*{1+C\frac{x}{(1-x)^2}},
    \end{align*}
    where $C:=4\sin^2\parens*{\frac{\pi}{k}}\in(0,3]$. Differentiating again,
    \begin{align*}
        f''(x)=\frac{k}{2\pi^2}\squ*{-\frac{1}{x^2}\ln\parens*{1+C\frac{x}{(1-x)^2}}+C\frac{1+x}{x(1-x)\parens*{(1-x)^2+Cx}}}.
    \end{align*}
    This is positive if and only if the multiple by $2\pi^2x^2/k$, $g(x):=-\ln\parens*{1+C\frac{x}{(1-x)^2}}+C\frac{x(1+x)}{(1-x)\parens*{(1-x)^2+Cx}}$ is. To show $g$ is positive, we need only show that it is positive at the boundary and at any inflection point. At $x=-1$,
    $$g(-1)=-\ln\parens*{1-\frac{C}{4}}\geq 0,$$
    as $C\leq 3$ so $1-\frac{C}{4}\leq 1$. In the limit $x\rightarrow 1$, $-\ln\parens*{1+C\frac{x}{(1-x)^2}}\sim 2\ln(1-x)$ and $C\frac{x(1+x)}{(1-x)\parens*{(1-x)^2+Cx}}\sim\frac{2}{1-x}$, so $\lim_{x\rightarrow 1}g(x)=+\infty$. To find the values of $g$ at the inflection points, consider the derivative
    \begin{align*}
        g'(x)=\frac{Cx(2x^3+Cx^2+(2C-6)x-C+4)}{(1-x)^2((1-x)^2+Cx)^2};
    \end{align*}
    the inflection points are at the values of $x$ where $g'(x)=0$. First, we have a zero at $x=0$, in which case $g(0)=0\geq 0$. I claim that $g'$ has no other zero in the interval $[-1,1]$, and hence no other inflection points we must consider. In fact, the only zeros of $g'$ are the zeros of $q(x):=2x^3+Cx^2+(2C-6)x-C+4$. We can show that $q$ only has one real root at some $x\leq -1$ by showing that it is positive at its inflection points, and that its concave inflection point is at $x=-1$. To do so, $q'(x)=6x^2+2Cx+2C-6$, which has roots $-1$ and $1-\frac{C}{3}$. Thus, the values of $q$ at its inflection points are $q(-1)=-2+C-(2C-6)-C+4=8-2C\geq 0$, and $q(1-C/3)=\frac{C}{27}\parens*{C^2-18C+54}\geq 0$ as its smallest nonzero root is $9-3\sqrt{3}>3$ (seeing $C$ as a variable). As such, the only inflection point of $g$ in $[-1,1]$ is at $x=0$, where it is positive, so $g(x)\geq 0$ for all $x\in[-1,1]$, and so $f$ is convex.
\end{proof}

\begin{lemma}[Restatement of Lemma \ref{lemma:increasing-ratio}]
    The function
    \begin{align}
        \lambda\in[-1,1]\mapsto\frac{1-\frac{1}{k}-\frac{k}{\pi^2}\Re\squ*{\mathrm{Li}_2(\lambda)-\mathrm{Li}_2(\lambda e^{i\frac{2\pi}{k}})}}{1-\lambda}
    \end{align}
    is increasing for all $k\geq 3$.
\end{lemma}

\begin{proof}
    Let $f(x)=\frac{1}{k}+\frac{k}{\pi^2}\Re\squ*{\mathrm{Li}_2(x)-\mathrm{Li}_2(x e^{i\frac{2\pi}{k}})}$. Then, using Lemma \ref{lem:convex-condition}, if $f(1)\leq 1$ and $f$ is convex, we get that $\frac{1-f(x)}{1-x}$ is increasing. To get these two properties, we can use Lemmas~\ref{lem:positive-endpoint} and~\ref{lem:convex-function}.
\end{proof}
\section{Vector Relative Distribution}\label{sec:classical-rel-dist}
In this section, we introduce the vector analogue of the (operator) relative distribution introduced in Section~\ref{sec:relative-distribution}. This provides an analogous relative distribution method for analysing approximation algorithms for classical CSPs.

Fix two unit vectors $\vec{a}$ and $\vec{b}$ in some Hilbert space $\C^d$. Consider the distribution $d_{\vec{a},\vec{b}}$ of pairs of unit vectors $(U\vec{a},U\vec{b})$ where $U$ is a Haar random unitary acting on $\C^d$. Indeed $d_{\vec{a},\vec{b}}$ is the uniform distribution over all pairs of unit vectors with the fixed inner product $\ip{\vec{a}}{\vec{b}}$.
\begin{definition}[Vector relative distribution]\label{def:vector-relative-distribution-using-process}
    The relative distribution of $(\vec{a},\vec{b})$, denoted by $\delta_{\vec{a},\vec{b}}$, is the distribution of the random variable $\theta$ in the following process:
\begin{enumerate}
    \item Sample $(\vec{x},\vec{y})$ from $d_{A,B}$.
    \item Sample a complex Gaussian vector $\vec{r}$ and let $(\alpha,\beta) = (\ip{\vec{r}}{\vec{x}},\ip{\vec{r}}{\vec{y}})$.
    \item Let $\theta \in [0,2\pi)$ be the relative angle between $\alpha$ and $\beta$, that is $\theta \coloneqq \measuredangle \alpha^*\beta$.\footnote{There is a simpler process that gives the same distribution for $\theta$:
    Sample a complex Gaussian vector $\vec{r}$ and let $\theta \coloneqq \measuredangle \ip{\vec{r}}{\vec{a}}^*\ip{\vec{r}}{\vec{b}}$.}
\end{enumerate}
\end{definition}

\begin{theorem}[Cauchy law for vectors]\label{thm:vector-relative-dist-pdf}
    Let $\lambda=|\lambda|e^{i\theta_0}=\ip{ \vec{a}}{ \vec{b}}$. If $|\lambda|\neq 1$, the PDF of the relative distribution induced by $\vec{a},\vec{b}$ exists with respect to the Lebesgue measure on $[0,2\pi)$ and is
    \begin{align}
        p_{\rel_{\vec{a},\vec{b}}}(\theta)=\frac{1-|\lambda|^2}{2\pi(1-|\lambda|^2\cos^2(\theta-\theta_0))}\squ*{\frac{|\lambda|\cos(\theta-\theta_0)}{\sqrt{1-|\lambda|^2\cos^2(\theta-\theta_0)}}\arccos(-|\lambda|\cos(\theta-\theta_0))+1}.
    \end{align}
    Else, the relative distribution is the Dirac delta distribution supported at $\theta=\theta_0$.
\end{theorem}

\begin{proof}
    First, note that there exists a unit vector $\vec{a}_\perp\in\C^d$ orthogonal to $\vec{a}$ such that
    \begin{align*}
        \vec{b}=\lambda \vec{a}+\sqrt{1-|\lambda|^2}\vec{a}_\perp.
    \end{align*}
    Extending these vectors to an orthonormal basis, we see that $a=\ip{\vec{r}}{\vec{a}}$ and $b=\ip{\vec{r}}{\vec{b}}$ depend only on the first two components of $\vec{r}$: $r_1=\ip{\vec{r}}{\vec{a}}$ and $r_2=\ip{\vec{r}}{\vec{a}_\perp}$, which are disributed as independent standard complex normals. With this notation, $a=r_1$ and $b=\lambda r_1+\sqrt{1-|\lambda|^2}r_2$. As such, we get that the random variable
    \begin{align*}
        \measuredangle\left(a^\ast b\right)&=\measuredangle\left(|r_1|^2\lambda+r_1^\ast r_2\sqrt{1-|\lambda|^2}\right)\\
        &=\measuredangle\left(\lambda+\frac{r_2}{r_1}\sqrt{1-|\lambda|^2}\right).
    \end{align*}
    Thus, to find the relative distribution where $|\lambda|<1$, we first derive the distribution of $\lambda+\frac{r_2}{r_1}\sqrt{1-|\lambda|^2}$ on $\C$. As noted, the joint PDF $p_{(r_1,r_2)}(z_1,z_2)=\frac{1}{(2\pi)^2}e^{-\frac{|z_1|^2+|z_2|^2}{2}}$, so the PDF of quotient distribution is
    \begin{align*}
        p_{\frac{r_2}{r_1}}(u)=\int|z|^2p_{(r_1,r_2)}(z,zu)dz=\frac{1}{\pi(1+|u|^2)^2}.
    \end{align*}
    Effecting the affine transformation,
    \begin{align*}
        p_{\lambda+\frac{r_2}{r_1}\sqrt{1-|\lambda|^2}}(u)=\frac{1-|\lambda|^2}{\pi\parens*{1-|\lambda|^2+|u-\lambda|^2}^2}=\frac{1-|\lambda|^2}{\pi\parens*{1+|u|^2-2\Re(\lambda^\ast u)}^2}.
    \end{align*}
    Expanding in polar coordinates $u=re^{i\theta}$, we have that the distribution of the argument is
    \begin{align*}
        p_{\rel_{\vec{a},\vec{b}}}(\theta)=\int_0^\infty p_{\lambda+\frac{r_2}{r_1}\sqrt{1-|\lambda|^2}}(re^{i\theta})rdr=\frac{1-|\lambda|^2}{\pi}\int_0^\infty \frac{r}{\parens*{1+r^2-2|\lambda|\cos(\theta-\theta_\lambda)r}^2}dr.
    \end{align*}
    Note that an antiderivative of $\frac{r}{(r^2-2\beta r+1)^2}$ is $\frac{\beta}{2(1-\beta^2)^{3/2}}\arctan\parens*{\frac{r-\beta}{\sqrt{1-\beta^2}}}+\frac{\beta r-1}{2(1-\beta^2)(r^2-2\beta r+1)}$ for any $\beta\in(-1,1)$. Computing and simplyfing, one finds
    \begin{align*}
        \int_0^\infty \frac{r}{\parens*{1+r^2-2\beta r}^2}dr=\frac{1}{2(1-\beta^2)}\squ*{\frac{\beta}{\sqrt{1-\beta^2}}\arccos(-\beta)+1}.
    \end{align*}
    Replacing $\beta$ with $|\lambda|\cos(\theta-\theta_\lambda)$ and reinstating the normalisation constant $\frac{1-|\lambda|^2}{\pi}$ gives the wanted result.

    On the other hand, if $|\lambda|=1$, we have immediately that $\measuredangle(a^\ast b)=\measuredangle(\lambda)=\theta_0$, so $\delta_{\vec{a},\vec{b}}=\delta_{\theta_0}$.
\end{proof}

Similarly to the operator relative distribution, the vector relative distribution can be used in the integral formula of Section~\ref{sec:fidelity-integral-formula} to analyse the approximation algorithms for classical CSPs.

\bibliographystyle{unsrt}
\bibliography{main}

\end{document}